%% file: main.tex
\documentclass[12pt, reqno]{amsart}
\usepackage{amsthm, amscd, amsfonts, amssymb, graphicx, xcolor,mathtools}
\usepackage{amsmath}
\usepackage[a4paper, left=0.85in, right=0.85in, top=1in, bottom=1in]{geometry}

\usepackage[round,authoryear]{natbib}
\usepackage{xcolor}      % optional, for color names
\usepackage[ bookmarksnumbered, plainpages,
            colorlinks=true,   % enable colored links
            citecolor=blue,    % set citation color
            linkcolor=red,    % set internal links color
            urlcolor=blue]{hyperref} % URL links
\usepackage{bookmark}
\usepackage{subcaption}
\usepackage{setspace}

\usepackage{floatrow}         % for [H] placement specifier
\usepackage{adjustbox}     % for \begin{adjustbox}{...}
\usepackage[font = tiny]{caption}       % for \captionof{table}
\usepackage{tabularray}    % for tblr environment (tabularray)
\usepackage{booktabs}  
\usepackage{ragged2e}
\usepackage{multirow}
\usepackage[shortlabels]{enumitem}
\usepackage{pdflscape}
\usepackage{rotating}

\usepackage{csquotes}
\newtheorem{theorem}{Theorem}[section]
\newtheorem{lemma}{Lemma}[section]
\newtheorem{proposition}{Proposition}[section]
\newtheorem{corollary}{Corollary}[section]
\theoremstyle{definition}

\newtheorem{assumption}[theorem]{Assumption}
\newtheorem{condition}{Condition}[section]

\newtheorem{remark}{Remark}[section]
\numberwithin{equation}{section}


\begin{document}

\setcounter{page}{1}

\title[\tiny Two Gaussians, Too Many]{Two Gaussians, Too Many: A bootstrap-based approach to assess identifiability in non-Gaussian Structural Vector Autoregressions}

\author[]{Paritosh Shankarrao Junare}
\thanks{Correspondence: Department of Economics, University of Bologna. Email: paritosh.junare@unibo.it. I am immensely grateful to Luca Fanelli, Giuseppe Cavaliere and Giovanni Angelini for their guidance and support throughout this work. I am also thankful to Markku Lanne, Jani Luoto, Mika Meitz and Savi Virolainen for their insightful comments. Useful feedback from Alessandra Luati and other participants of 2026 RCEA Time Series Workshop, Madrid; Katerina Petrova and other participants of the IAAE Annual Conference 2026, Lisbon is gratefully acknowledged.}

\dedicatory{October 2026}

\begin{abstract}
\noindent Normality pre-tests on reduced-form innovations cannot detect if two or more underlying structural shocks are Gaussian, and hence cannot detect a failure of identification in non-Gaussian SVARs. We propose a bootstrap-based approach that reveals such a failure. Under the null of valid identification and certain regularity conditions, the conditional bootstrap distribution of the maximum likelihood estimator of the impact matrix is asymptotically normal, so the diagnostic reduces to a normality test on the bootstrap replications. The impact-matrix estimator remains consistent and asymptotically normal when one shock is Gaussian and its shape parameter lies on the boundary of the parameter space, so the diagnostic is valid across the entire null hypothesis. When the number of bootstrap replications and the sample size diverge jointly at an appropriate rate, the test statistic, conditional on the data, is asymptotically pivotal. This implies that conditioning on this diagnostic does not induce pre-testing bias asymptotically. In simulations with Normal-inverse Gaussian shocks, the multivariate diagnostic is close to its nominal size under the null, its power against an additional Gaussian shock increases with the sample size, and its univariate version localizes the failure in large samples. Based on the estimates of a SVAR in
the macroeconomic and financial uncertainty literature, we demonstrate its potential as a
practical, robust tool for assessing non-Gaussian identification.
\end{abstract} \maketitle
\section{Introduction}{\label{sec:intro}}
\noindent \input{00_intro.tex}
\section{Identification in non-Gaussian SVARs}\label{sec:ICA}
\noindent \input{01_ICA.tex}

\section{Bootstrap Diagnostic}\label{sec:BSdiag}
\noindent \input{04_BSdiag.tex}

\section{Monte Carlo Evidence}\label{sec:MC}
\noindent \input{05_MC.tex}
% \section{The Case of Weak Identification}\label{sec:WKId}
% \noindent  \input{05_WKId.tex}
\section{Empirical Application}\label{sec:EMP}
\noindent \input{06_EMP.tex}

\section{Conclusion}\label{sec:concl}
\noindent \input{11_con.tex}
\pagebreak
\bibliographystyle{apalike}
\bibliography{Ref1}
\pagebreak

\appendix
\noindent \input{12_app.tex}
\end{document}

%% file: 00_intro.tex
How do we identify the latent structural drivers of the economy? This question has long been a cornerstone of research in empirical macroeconometrics. A prevalent approach with structural vector autoregressions (SVARs) involves linking observed macroeconomic variables to structural
shocks. Since the covariance matrix of the reduced-form innovations alone is not sufficient, additional
restrictions are necessary to recover the structural shocks from the reduced-form innovations. Traditional approaches to identification rely on economic theory, which include
short-run exclusion restrictions \citep{blanchard_empirical_2002}, cointegration-based constraints 
\citep{blanchard_dynamic_1989}, sign restrictions \citep{uhlig_what_2005}, or employ external information as instruments for exogenous sources of variations \citep{stock_identification_2018}. Another strand of literature exploits statistical properties of the structural shocks, while avoiding any economic restrictions, to achieve identification. Derived from the independent component analysis (ICA) literature, see \citet{comon_independent_1994, hyvarinen_fast_1999}, identification based on non-Gaussianity of the structural shocks focuses on the information beyond the second-order moments of the shocks, such as skewness and excess kurtosis, see \citet{lanne_identification_2017, gourieroux_statistical_2017, keweloh_generalized_2021}. Moreover, given the nature of statistical identification, we still require economic labeling to gain meaningful interpretation.

 If the structural shocks are mutually independent and at most one of them follows a Gaussian distribution, the shocks can be identified up to sign and permutation indeterminacies. However, if two or more independent shocks are Gaussian, rotations within the Gaussian subspace leave the joint distribution of the reduced-form innovations unchanged \citep{gourieroux_statistical_2017}. This implies that the structural model is not uniquely determined, which makes ensuring the validity of these assumptions essential for valid asymptotic inference. In empirical research, most studies verify these identification conditions indirectly through univariate tests of Gaussianity on the reduced-form residuals, see \citet{keweloh_generalized_2021, lanne_identification_2017}, among others. We show that this residual normality pre-test is insufficient to ensure the validity of these conditions. Since the reduced-form innovations are simply linear combinations of the structural shocks, they remain non-Gaussian even with two or more Gaussian shocks. This may lead a practitioner to presume the validity of the identification conditions when in fact, they are violated.

 Instead, we propose a novel bootstrap-based approach to assess the identifiability of non-Gaussian SVARs. The proposed method, under certain regularity conditions, tests the null hypothesis of strong identification, where at most one of the independent shocks is Gaussian, against the alternative that two or more structural shocks are Gaussian. The diagnostic measures the divergence of the bootstrap distribution of the impact-matrix estimator from its asymptotic benchmark of normality: the two coincide under valid identification and diverge systematically under its failure. We build on a growing literature in bootstrap-based diagnostic testing: \citet{angelini_bootstrap_2022, angelini_identification_2024} show that bootstrap can be used to detect weak identification and model misspecification in dynamic state-space models, and SVARs identified with external instruments. Moreover, \citet{cavaliere_bootstrap_2025} provide a framework for diagnosing general forms of misspecification such as non-stationarity and weak instruments, among others, which we specialize and make operational for non-Gaussian identification.

We provide three main results: \emph{first}, we turn the bootstrap-validity principle into an
identification diagnostic for non-Gaussian SVARs. Under valid identification and certain regularity conditions, the
conditional bootstrap distribution of the impact-matrix estimator obtained by the non-Gaussian maximum likelihood (NGML) estimator of \citet{lanne_identification_2017} is asymptotically standard
normal. This implies that the check reduces to a normality
test on the bootstrap replications. \emph{Second}, we show that this validity extends to the single-Gaussian
boundary of the null, where the Gaussian shock's shape parameters lie on the
boundary of the parameter space and the full information matrix is singular.
But the impact-matrix estimator remains consistent and asymptotically normal, so the diagnostic is valid across the entire null. Beyond the
diagnostic, this carries a practical benefit: practitioners who suspect a Gaussian
shock need not determine which shock it is, nor adopt a different specification. 
They can estimate all shocks under the Normal-inverse Gaussian (NIG) family, which nests the Gaussian at its boundary, and obtain consistent estimates with valid inference
regardless of which shock, if any, is Gaussian. \emph{Third}, when we apply the univariate version of the diagnostic to each element of the impact matrix, we can localize where the identification fails. Under partial identification
\citep{maxand_identification_2020}, the non-Gaussian shocks remain identified
while the Gaussian ones do not. The diagnostic rejects only for the columns of
the impact matrix that belong to the unidentified shocks and retains nominal size
on the others. A practitioner can deduce which shocks, and hence which impulse
responses, are affected by the identification failure.
% Under the alternative of two or more Gaussian shocks
% the bootstrap distribution diverges from its limit and the test is consistent, with
% non-negligible power.

Apart from the residual normality pre-test, which is insufficient, current literature for verifying the identification conditions includes a moment-based rank test on reduced-form residuals
\citep{guay_identification_2021} or a specification test on the estimated shocks \citep{amengual_moment_2022,amengual_specification_2024}. But such specification tests may
induce severe pre-testing bias, since the finite-sample distribution of a
post-selection estimator is not uniformly close to its asymptotic benchmark
\citep{danilov_harm_2004, leeb_model_2005, roth_pretest_2022}. Crucially, we show that this bootstrap-based approach does not induce any pre-testing bias asymptotically. This is because, unlike standard bootstrap asymptotic theory where the number of bootstrap replications ($M$), and the sample size ($T$) diverge to infinity sequentially, first $M \rightarrow \infty$, and then $T \rightarrow \infty$, in this framework they diverge jointly, i.e., $M,T \rightarrow \infty$ and $M/T \rightarrow 0$ at an appropriate rate. The joint divergence ensures that under the null hypothesis, the test statistic, conditional on the data, is asymptotically pivotal and does not distort subsequent inference. Hence, the bootstrap allows not only to refine first-order approximations, but
also to detect identification failure in non-Gaussian SVARs without distorting
post-test inference.

We validate the theoretical results with a Monte Carlo study across three different sets of distributional specifications for the structural shocks: (i) \textbf{ICA 0}, 
where all three shocks are non-Gaussian, (ii) \textbf{ICA 1}, where one shock is Gaussian and the other two are non-Gaussian, and (iii) \textbf{No ICA}, where two of the shocks are Gaussian and one is non-Gaussian. Under \textbf{ICA 0} and \textbf{ICA 1}, the estimates of the impact matrix are
centered on the true values and the confidence intervals have close to nominal
coverage, consistent with asymptotic normality. Under \textbf{No ICA}, the
columns associated with the two Gaussian shocks are not consistently estimated:
their estimates and bootstrap distributions are bimodal, and coverage falls
well below the nominal level.
The bootstrap diagnostic test of normality, under the appropriate rate of joint divergence of $M$ and $T$, 
demonstrates correct empirical size under the null and substantial  
power against the alternative of identification failure; $M$ governs the size-power trade-off in finite samples. In contrast, the residual normality pre-test cannot distinguish between valid identification and identification failure caused by an additional Gaussian shock. The bootstrap-based approach thus provides a reliable and 
implementable tool for evaluating identification conditions for non-Gaussian SVARs in empirical applications. Furthermore, under weak identification, where one structural shock is
near-Gaussian in the presence of another Gaussian shock, the univariate
diagnostic localizes the failure to the columns of the impact matrix
associated with these two shocks, and its rejection frequency increases with
the sample size. The residual normality pre-test, in contrast, flags fewer
samples as unidentified with increasing sample size. This is because it focuses on non-Gaussianity in linear combinations of the
shocks rather than on the strength of identification.

We apply the diagnostic to the trivariate SVAR of financial uncertainty, real
activity and macroeconomic uncertainty of \citet{ludvigson_uncertainty_2021},
which asks whether uncertainty is a cause or a consequence of business cycle
fluctuations. The diagnostic does
not reject the validity of the bootstrap approximation at the $5\%$
level. The impulse responses,
obtained without economic identifying restrictions, show that financial
uncertainty depresses real activity with a delay and is largely
unaffected by the other shocks, whereas macroeconomic uncertainty has no
sustained real effects and responds persistently to both other shocks.
This supports the evidence of \citet{ludvigson_uncertainty_2021} which shows
financial uncertainty is an exogenous driver of business cycle fluctuations and
macroeconomic uncertainty largely an endogenous response to them.

The remainder of the paper is organized as follows. Section~\ref{sec:ICA} discusses identification in non-Gaussian SVARs, the NGML estimator and its asymptotic properties, including the single-Gaussian boundary, and the asymptotic distribution of the bootstrap estimator. Section~\ref{sec:BSdiag} introduces the bootstrap diagnostic. Section~\ref{sec:MC} presents Monte Carlo results, including for weak identification scenarios. Section~\ref{sec:EMP} contains the empirical application and Section~\ref{sec:concl} concludes. The notation used throughout the paper are defined in Appendix \ref{appendix:a}. 

%% file: 01_ICA.tex
This section elaborates on identification in non-Gaussian SVARs through the NGML estimator of \citet{lanne_identification_2017}. The structural shocks follow a Normal-inverse Gaussian (NIG) distribution, which nests the Gaussian distribution as a limiting case. We extend the consistency and asymptotic normality of the impact-matrix estimator to the single-Gaussian boundary, at which one shock is exactly Gaussian (Proposition~\ref{prop:onegaussian}). Furthermore, under the null of valid identification, we also provide the asymptotic distribution of the bootstrap impact-matrix estimator (Proposition~\ref{prop:BSvalidity}).

 Consider a VAR($p$) model in its companion form:
\begin{equation}
    \mathbf{Y}_t = \mathbf{\Pi} \mathbf{Y}_{t-1} + \mathbf{U}_t 
    \label{eq:companionVAR}
\end{equation}
where,
\begin{equation}
    \mathbf{Y}_t = \underset{np \times 1}{\begin{pmatrix}
    y_t \\
    y_{t-1} \\
    \vdots \\
    y_{t-p+1}
    \end{pmatrix}}, \quad
    \mathbf{\Pi} = \underset{np \times np}{\begin{pmatrix}
    \Pi_1 & \Pi_2 & \cdots & \Pi_p \\
    I_n & 0 & \cdots & 0 \\
    0 & I_n & \cdots & 0 \\
    \vdots & \vdots & \ddots & \vdots \\
    0 & 0 & \cdots & I_n
    \end{pmatrix}}, \quad
    \mathbf{U}_t = \underset{np \times 1}{\begin{pmatrix}
    u_t \\
    0 \\
    0 \\
    \vdots \\
    0
    \end{pmatrix}}
    \label{eq:setup_companion_matrices}
\end{equation}

and $u_t = B\varepsilon_t$ with,
\begin{equation}
    \underset{n \times n}{B} = \begin{pmatrix}
    \beta_{11} & \beta_{12} & \cdots & \beta_{1n} \\
    \beta_{21} & \beta_{22} & \cdots & \beta_{2n} \\
    \vdots & \vdots & \ddots & \vdots \\
    \beta_{n1} & \beta_{n2} & \cdots & \beta_{nn}
    \end{pmatrix}
    \label{eq:setup_Bmatrix}
\end{equation}

Here, $y_t$ is an $n \times 1$ vector of endogenous variables,
$\Pi_1, \ldots, \Pi_p$ are $n \times n$ autoregressive coefficient matrices, and
$u_t$ is an $n \times 1$ vector of reduced-form innovations. The innovations are
linked to the $n \times 1$ vector of structural shocks $\varepsilon_t$ through
$u_t = B\varepsilon_t$, where the nonsingular $n \times n$ impact (mixing)
matrix $B$ is the object of identification. Since the structural shocks are
mutually uncorrelated with diagonal covariance matrix $\Sigma_\varepsilon$, the
reduced-form covariance matrix is $\mathbb{E}(u_t u_t') = B \Sigma_\varepsilon B'$.
We assume that the VAR in \eqref{eq:companionVAR} is stable, that is, the
companion matrix $\mathbf{\Pi}$ satisfies:

\begin{assumption}$\det(I_{np} - \mathbf{\Pi}z) \neq 0, \ \forall z \in \mathbb{C}, |z| \leq 1$.    
\label{as:stability}
\end{assumption}

Beyond a stable VAR, we assume that the structural shocks are mutually
\textit{independent} and \textit{non-Gaussian}, and exploit the information in addition to the unconditional variances of the reduced-form innovations, i.e., skewness and/or excess kurtosis of the structural shocks. This yields identification of $B$ up to sign and column-permutation
indeterminacies, without imposing any additional identifying restrictions. This approach follows from the \citet{darmois_analyse_1953, skitovich_linear_1953} theorem, which states 
that if  $X = (X_1, X_2, \ldots, X_n)$ are independent random variables and $a'X$ and $b'X$ are independent for $a_i, b_i \in \mathbb{R}\setminus\{0\}$, then all $X_i, \ i = 1,2, \ldots, n$ are Gaussian. Extending this result and formalizing it in the ICA literature, \citet{comon_independent_1994} shows that we can extract a unique decomposition of $u_t = B\varepsilon_t$ up to sign and column permutation, if $\varepsilon_t$ has at most one Gaussian component and all
components are independent\footnote{It should be noted that the assumption of independence of structural shocks is more restrictive than uncorrelatedness.
This is because it is not guaranteed to recover independent structural shocks from non-Gaussian reduced-form innovations solely through linear
transformation \citep{kilian_structural_2017}.}. However, recent literature shows that it is not necessary to assume complete independence of the shocks
to achieve identification. \citet{jarocinski_estimating_2024} relaxes the assumption of complete independence and assumes that the shocks are drawn from
partially dependent multivariate $t$-distributions. Other 
moment-based estimators, instead of independence, assume restrictions on higher-order moments and co-moments of the shocks, see \citet{guay_identification_2021}, \citet{lanne_identifying_2023}, or on cumulant tensors
\citep{mesters_non-independent_2024}. 

\subsection{NGML estimation}\label{subsec:NGMLest}
Apart from non-parametric algorithms, see \citet{hyvarinen_independent_2000}, \citet{hyvarinen_independent_2013} for a detailed review, other likelihood-based approaches can be used for estimation. With NGML we assume that the structural shocks, $\varepsilon_{i,t}$, are i.i.d.\ sequences\footnote{\citet{lanne_identification_2017} relax the temporal independence assumption to temporal uncorrelatedness, allowing for conditional heteroscedasticity in the shocks.}, with at most one of them Gaussian.
\begin{assumption}
    \textit{The structural shocks, $\varepsilon_{i,t}, i = 1,2,\ldots, n$, are mean-zero, mutually independent and identically distributed (i.i.d.) processes,
    with diagonal covariance matrix, i.e., $\mathbb{E}(\varepsilon_t\varepsilon_t') = \operatorname{diag}(\sigma_i^2), \  \sigma_i^2 > 0$.}
    \label{as:indep}
\end{assumption}
\begin{assumption}
    \textit{At most one of the components of $\varepsilon_t$ is Gaussian.}
    \label{as:gau}
\end{assumption}
The densities of the structural shocks may depend on their individual set of parameters, $\theta_i$. In the likelihood, the density of $\varepsilon_{i,t}$ is
$f_{i,\sigma_i}(x;\theta_i) \coloneqq \sigma_i^{-1} f_i(\sigma_i^{-1}x;\theta_i)$, where $f_i(\cdot; \theta_i)$ has
mean zero and unit variance, so that $\sigma_i$ is the standard deviation of
$\varepsilon_{i,t}$. As noted before, ICA ensures identification of the impact matrix up to sign and column permutations. We
select a unique representative from each equivalence class, as in
\citet{lanne_identification_2017}: columns are scaled to unit norm, permuted such
that $|c_{ii}|>|c_{ij}|$ for $i<j$, and rescaled to a unit diagonal, with the
scale absorbed by $\sigma_i$. The resulting $B(\beta)$ is unique outside a set
of Lebesgue measure zero. We refer to $B(\beta)$ as the normalized impact matrix. Its column order is a statistical convention, we still require economic labeling to provide interpretation to the statistically identified shocks.

In our framework, the structural shocks are assumed to be from the NIG distribution, a member of the generalized hyperbolic distribution
family. It is suitable for modeling processes where the probability of obtaining extreme values is higher, relative to the Gaussian distribution \citep{barndorff-nielsen_normal_1997}. \citet{jarocinski_estimating_2024, andrade_higher-order_2025}, among others,
note non-trivial excess kurtosis in the distribution of macroeconomic shocks, and their proxies respectively\footnote{\citet{jarocinski_estimating_2024},
\citet{lanne_identification_2017}
and other studies employ independent Student-$t$ distributions to capture excess kurtosis. The Student-$t$ and the NIG distributions both belong to the generalized hyperbolic family; the NIG has finite moments of all orders and allows for skewness through $\gamma$ \citep{barndorff-nielsen_normal_1997}.}. The NIG 
distribution is defined
by four parameters: $\alpha$ (tail heaviness), $\gamma$ (asymmetry), $\delta$ (scale) and $\mu$ (location), and is denoted 
as $NIG(\alpha, \gamma, \delta, \mu)$. Its density function is given by:

\begin{equation}
    f(x; \alpha, \gamma, \delta, \mu) = \frac{\alpha\delta K_1(\alpha\sqrt{\delta^2 + (x - \mu)^2})}{\pi\sqrt{\delta^2 + (x - \mu)^2}} \exp(\delta\sqrt{\alpha^2 - \gamma^2} + \gamma(x - \mu))
    \label{eq:NIGpdf}
\end{equation}
where, $K_1$ is the modified Bessel function of the second kind with index 1. The Gaussian distribution $\mathcal{N}(\mu,\sigma^2)$ is its limit as $\alpha \rightarrow \infty$ with $\gamma = 0$ and $\delta = \sigma^2\alpha$.

% Standard asymptotic-normality results for these estimators, whether the fully parametric NGML estimator of \citet{lanne_identification_2017} or the pseudo-maximum likelihood estimator of \citet{gourieroux_statistical_2017}, are derived under a nonsingular information matrix. This condition is violated at the single-Gaussian boundary, where the estimated shape parameters of the Gaussian shock lie on the edge of the parameter space and its information block vanishes, even though $B$ itself remains identified. We show that the impact-matrix estimator remains consistent and asymptotically normal, which eventually allows the bootstrap diagnostic to be valid across the null hypothesis of \emph{at most} one Gaussian shock.

Before setting up the estimation procedure, we summarize the parameters of the model: with the unit-diagonal normalization of $B$,
the off-diagonal parameters of $B$ are collected in an $n(n-1) \times 1$ vector, $\beta \coloneqq \operatorname{vecd}(B), \ \beta \in \mathbb{R}^{n(n-1)}$, 
where $\operatorname{vecd}(\cdot)$ is the operator that 
stacks the columns of an $n \times n$ matrix, excluding the diagonal elements from its $\operatorname{vec}(\cdot)$ form. We can collect the parameters of the individual shocks' densities $(\alpha_i, \gamma_i, \delta_i, \mu_i)$ in a vector $\theta_i$, i.e.,
$\theta_i \coloneqq (\alpha_i, \gamma_i, \delta_i, \mu_i), \ i = 1,2, \ldots, n$. Furthermore, for brevity, all the model parameters (including the autoregressive coefficients) 
are collected in a vector $\lambda \coloneqq (\Pi,\beta,\sigma,\theta)$, where $\Pi = \operatorname{vec}(\Pi_1, \Pi_2, \ldots, \Pi_p)$ are the vectorized-autoregressive coefficients,
$\theta = (\theta_1', \theta_2', \ldots, \theta_n')'$, where each $\theta_i$ is a $4 \times 1$ vector, are the shape parameters and 
$\sigma = (\sigma_1, \ldots, \sigma_n)$ are the standard deviations of the structural 
shocks, $i = 1, 2, \ldots, n$. We denote $\lambda_0$ as the true parameter vector.

NGML is a two-step estimation procedure\footnote{\citet{lanne_identification_2017} also
consider a three-step estimation method, for large $n$ and small $T$, provided the distribution of the underlying $\varepsilon_{i,t}$
is symmetric. This estimator is
asymptotically efficient.}. First, the least
squares estimates of the reduced-form innovations, $\hat{u}_t$, are obtained. With the estimated $\hat{u}_t$, the log-likelihood: 
\begin{equation}
    \mathcal{L}_T(\lambda) = T^{-1}\sum_{t=1}^{T}\ell_t(\lambda),
    \label{eq:loglik}
\end{equation}
where, 
\begin{equation}
    \ell_t(\lambda) = \sum_{j=1}^{n} \log f_j(\sigma_j^{-1}\imath_j'B(\beta)^{-1}\hat{u}_t; \theta_j) - \log |\det(B(\beta))| -
     \sum_{j=1}^{n} \log \sigma_j
     \label{eq:lik}
\end{equation}
is maximized with respect to the parameter vector $(\beta',\sigma,\theta)'$. Here, $\imath_j$ is the $j^{th}$ unit vector, i.e., $\imath_j = (0, \ldots, 1, \ldots, 0)'$, where $1$ is in the $j^{th}$ position. The estimator is given by: 
\begin{equation}
    \hat{\lambda}_T = (\hat{\Pi}_T, \hat{\beta}_T, \hat{\sigma}_T, \hat{\theta}_T) = \underset{\beta, \sigma, \theta}{\arg\max} \ 
     \mathcal{L}_T(\lambda)
    \label{eq:NGMLestimator}
\end{equation}
The estimated structural shocks are obtained by pre-multiplying a transposed $j^{th}$ unit vector, $\imath_j'$,
\begin{equation}
\hat{\varepsilon}_{j,t} = \hat{\sigma}_j^{-1}\imath_j'B(\hat{\beta}_T)^{-1}\hat{u}_t, \ j = 1,2,\ldots, n. 
\label{eq:est_struc_shock}
\end{equation}
Here, $B(\hat{\beta}_T)$ is the estimated normalized matrix $B$ as a function of its off-diagonal elements $\hat{\beta}_T$.
Once the structural shocks are identified and estimated, the impulse responses are computed as 
\begin{equation}
\underset{n \times n}{\hat{\Psi}_i} = R \ \hat{\mathbf{\Pi}}^i \ R' \ B(\hat{\beta}_T)\operatorname{diag}(\hat{\sigma}_T), \quad i = 0,1,2,\ldots, h,
\label{eq:IRF}
\end{equation}
where, $R \coloneqq (I_n , \mathbf{0}_{n \times n(p-1)})$ is a selection matrix, $\hat{\mathbf{\Pi}}$ is the companion form of the
estimated autoregressive coefficients, $\hat{\Pi}_1, \hat{\Pi}_2, \ldots, \hat{\Pi}_p$.

The on-impact responses to one-standard-deviation shocks are given by $B(\hat{\beta}_T)\operatorname{diag}(\hat{\sigma}_T)$. 
 Henceforth, for brevity, we refer to it as the impact matrix, with true value $B_0 \coloneqq B(\beta_0)\operatorname{diag}(\sigma_0)$, and denote its estimator by:
 \begin{equation}
 \hat{B}_T \coloneqq B(\hat{\beta}_T)\underset{n \times n}{\operatorname{diag}}(\hat{\sigma}_T)
 \label{eq:BT} 
\end{equation}
We summarize the set of regularity conditions on the non-Gaussian SVAR which allow for standard asymptotic and bootstrap inference for NGML 
estimates in Appendix~\ref{subsec:regcond}. These regularity conditions include the stability conditions for the SVAR model, that the true parameter values lie in a permissible parameter space, conventional differentiability assumptions on the density functions, suitable integrability conditions ensuring the score function has zero mean and finite covariance matrix when evaluated at the true parameter value, and the covariance matrix of the score is positive definite. The last condition, (7), fails when one shock is Gaussian; condition (7a) replaces it in that case.

Standard asymptotic-normality results for these estimators, whether the fully parametric NGML estimator of \citet{lanne_identification_2017} or the pseudo-maximum likelihood estimator of \citet{gourieroux_statistical_2017}, are derived under a nonsingular information matrix. This condition is violated at the single-Gaussian boundary, where the estimated shape parameters of the Gaussian shock lie on the edge of the parameter space and its information block vanishes, even though $B$ itself remains identified. We show that the impact-matrix estimator remains consistent and asymptotically normal, which eventually allows the bootstrap diagnostic to be valid across the null hypothesis of \emph{at most} one Gaussian shock.
\begin{remark}[Identification Failure and the Fisher Information Matrix]\label{rem:identification_failure}
     Assumptions~\ref{as:indep} and \ref{as:gau} are jointly necessary for the positive definiteness of the Fisher information matrix $\mathcal{I}(\lambda_0)$ stated in condition (7) of Appendix~\ref{subsec:regcond}. To see this, suppose Assumption~\ref{as:gau} fails and two shocks, say $\varepsilon_{i,t}$ and $\varepsilon_{j,t}$ are Gaussian. Then any rotation of the structural space within the Gaussian subspace $\{i,j\}$ leaves the joint distribution of $u_t$ unchanged \citep{gourieroux_statistical_2017}. This implies that the log-likelihood is invariant along the rotation of the Gaussian subspace. Hence, $\mathcal{I}(\lambda_0)$ is singular in that
    rotational direction and this causes the NGML estimator to be inconsistent and its bootstrap distribution to deviate from normality: the mechanism which our diagnostic exploits.
\end{remark}

 \subsection{Consistency and asymptotic normality} Consider the model representation in equation \eqref{eq:companionVAR}, 
 satisfying the regularity conditions of Appendix~\ref{subsec:regcond} and Assumptions~\ref{as:stability}--\ref{as:gau}, then the NGML estimator,
    $\hat{\lambda}_T = (\hat{\Pi}_T, \hat{\beta}_T, \hat{\sigma}_T, \hat{\theta}_T)$, obtained by maximizing the log-likelihood function 
    \eqref{eq:loglik} and \eqref{eq:lik}, satisfies\footnote{For detailed proofs, see \citet{lanne_identification_2017}, Theorem 1. When one structural shock is Gaussian, \eqref{eq:asymp} holds
    for the autoregressive and impact parameters $\psi=(\operatorname{vec}(\Pi)',\beta',\sigma')'$
    with the robust covariance of Proposition~\ref{prop:onegaussian} in place of
    $\mathcal{I}(\lambda_0)^{-1}$; the impact-matrix estimator $\hat{B}_T$
    remains asymptotically normal, which is the object of interest.}:

 \begin{equation}
    T^{1/2}(\hat{\lambda}_T - \lambda_0) \xrightarrow{d} \mathcal{N}(0, \mathcal{I}(\lambda_0)^{-1}), \quad \text{as} \ T \rightarrow \infty,
    \label{eq:asymp}
 \end{equation}
 where $\xrightarrow{d}$ denotes convergence in distribution, $\mathcal{I}(\lambda_0)$ is the Fisher information matrix, defined as,
 $\mathcal{I}(\lambda_0) \coloneqq -\mathbb{E}[\nabla^2 \ell_{\lambda\lambda,t}(\lambda_0)]$, where 
 $\nabla^2 \ell_{\lambda\lambda,t}(\lambda_0) = {\partial^2 \ell_t(\lambda_0)}/{\partial \lambda \partial \lambda'}$. 
 Furthermore,
 a consistent estimator of the asymptotic covariance matrix can be obtained from the NGML estimator, $\hat{\lambda}_T$ and the Hessian matrix of the log-likelihood function, i.e.,
    \begin{equation}
       - \mathcal{L}_{\lambda\lambda,T}^{-1}(\hat{\lambda}_T) \coloneqq -\left(T^{-1} \sum_{t=1}^{T}\nabla^2 \ell_{\lambda\lambda,t}(\hat{\lambda}_T)\right)^{-1} \rightarrow \quad \mathcal{I}(\lambda_0)^{-1}, \quad \text{as} \ T \rightarrow \infty,
        \label{eq:fisher}
    \end{equation}

We extend these results of consistency and asymptotic normality of the estimator to the single-Gaussian boundary, at which exactly one shock is Gaussian. The identification conditions permit at most one Gaussian shock, and the boundary of
this null is precisely where the standard asymptotics require attention. When one shock's
fitted NIG density reaches its Gaussian limit ($\alpha_k\to\infty$, equivalently
$\rho_k=1/\alpha_k=0$), that shock's shape parameters are no longer identified: the
corresponding block of the Fisher information collapses, so the full-parameter
information matrix $\mathcal{I}(\lambda_0)$ is singular. This boundary is not an isolated knife-edge; the full-parameter information matrix is already ill-conditioned
whenever a single shock is merely near-Gaussian. The single-Gaussian boundary is simply the limit of this
near-singular region. The case is therefore a genuine
part of the null hypothesis over which the diagnostic must remain valid.

Proposition~\ref{prop:onegaussian} establishes that this ill-conditioning is confined
to the Gaussian shock's own shape parameters and does not propagate to the
impact-matrix estimator. Although the full-parameter information matrix is singular at the boundary, the
autoregressive and impact parameters $\psi$ retain a nonsingular profile information $\mathcal{I}_\psi$. Since $\mathcal{I}_\psi$ is nonsingular at the
boundary, the impact-matrix
estimator $\hat{B}_T$ is consistent and asymptotically normal as a
shock ranges from strongly non-Gaussian to (and including) the Gaussian limit.
Inference for the impact-matrix estimates is thus stable across the entire null, not merely at its
interior. The estimator loses identification only when a \emph{second} shock also
approaches Gaussianity, which is the rotational near-indeterminacy of the weak-identification
regime studied in Section~\ref{sec:WKId}.

 \begin{proposition}[]\label{prop:onegaussian}
Let conditions (1)--(6) and (7a) of Appendix~\ref{subsec:regcond} and Assumptions
\ref{as:stability}--\ref{as:indep} hold, and suppose exactly one structural shock is
Gaussian. Then the NGML estimator of the autoregressive and impact parameters $\psi=(\operatorname{vec}(\Pi)',\beta',\sigma')'$ is consistent and asymptotically normal, and so is the impact-matrix estimator $\hat{B}_T=B(\hat\beta_T)\operatorname{diag}(\hat\sigma_T)$,
\begin{equation}
  T^{1/2}\bigl(\hat{B}_T-B_0\bigr)\;\xrightarrow{d}\;\mathcal{N}\bigl(0,\;\mathcal{V}_B\bigr),
  \label{eq:Bblock_normal}
\end{equation}
where $\mathcal{V}_B$ is obtained from the robust estimator
$\mathcal{I}_\psi^{-1}\mathcal{J}_\psi\mathcal{I}_\psi^{-1}$, and $\mathcal{I}_\psi$ is positive definite by Lemma~\ref{lem:profileinfo}. When all
structural shocks are non-Gaussian the information identity
$\mathcal{I}_\psi=\mathcal{J}_\psi$ holds and $\mathcal{V}_B$ reduces to the delta-method
transformation \eqref{eq:SigmaB} of $\mathcal{I}_\psi^{-1}$.
\end{proposition}
\begin{proof}
    See Appendix~\ref{proof:onegaussian}.
\end{proof}
\begin{remark}\label{rem:onegaussian_bridge}
Because the diagnostic of Section~\ref{sec:BSdiag} enters only through $\hat{B}_T$
and requires only the asymptotic normality \eqref{eq:Bblock_normal} of its bootstrap
distribution, Proposition~\ref{prop:onegaussian} is exactly what makes the procedure
valid under the entire null hypothesis. 
\end{remark}

Proposition~\ref{prop:onegaussian} has an implication for practitioners that is
independent of the diagnostic. A practitioner who believes one structural shock may be
Gaussian need not decide \emph{which} one, nor switch to a partly Gaussian or
otherwise restricted specification. Since the NIG family nests the Gaussian at the boundary of its parameter space,
estimating all $n$ shocks as NIG leaves $\hat{B}_T$ consistent and
asymptotically normal, with standard errors from $\mathcal{V}_B$ in
\eqref{eq:Bblock_normal}, even when one shock is Gaussian.

  By the delta method, see \citet{kilian_structural_2017, bruggemann_inference_2016}, the covariance matrix of $\hat{B}_T = B(\hat{\beta}_T)\operatorname{diag}(\hat{\sigma}_T)$, 
 denoted as $\hat{\Sigma}_{{B}_T}$, can be obtained as: for $i = 1,2,\ldots, n$ and $j = 1,2,\ldots, n$,
 
\begin{align}
    \hat{\Sigma}_{{B}_T}(i,j) &= \hat{\Sigma}_{{\sigma}_T}(i,j), \quad \text{for} \ i = j, \\
    \hat{\Sigma}_{{B}_T}(i,j) &= \operatorname{Var}(B(\hat{\beta}_T)_{i,j})(\hat{\sigma}_T(j))^2 + \operatorname{Var}(\hat{\sigma}_T(j))(B(\hat{\beta}_T)_{i,j})^2 \nonumber\\ 
                            &+2\operatorname{Cov}(B(\hat{\beta}_T)_{i,j}, \hat{\sigma}_T(j))B(\hat{\beta}_T)_{i,j}\hat{\sigma}_T(j),
    \quad \text{for} \ i \neq j,
 \label{eq:SigmaB}
\end{align}
 where $B(\hat{\beta}_T)_{i,j}$ is the $(i,j)^{th}$ element of the matrix $B(\hat{\beta}_T)$. Equation \eqref{eq:SigmaB} gives the variances of the elements of $\hat{B}_T$; the covariances follow from the same delta-method expansion. Throughout, $\hat{\Sigma}_{B_T}$ denotes the $n\times n$ robust estimator of the asymptotic covariance $\mathcal{V}_B$ of $T^{1/2}\operatorname{vec}(\hat{B}_T-B_0)$ (Proposition~\ref{prop:onegaussian}). When all shocks are non-Gaussian it reduces to the delta-method transformation of $\mathcal{I}_\psi^{-1}$.

 \subsection{Bootstrap inference}\label{sec:BSINF}
 Here we discuss the bootstrap estimator, $\hat{B}_T^*$ and provide its asymptotic distribution under the null of valid identification. We strictly follow the residual-based moving-block bootstrap (MBB) algorithm of \citet{bruggemann_inference_2016}, see Appendix~\ref{subsec:RMBB}, where the residual-based MBB yields asymptotically valid inference for the reduced-form VAR under conditional heteroscedasticity of unknown form. In the Monte Carlo simulations, the structural shocks are 
drawn i.i.d.\ from a non-Gaussian
distribution which allows for parametric i.i.d.\ bootstrap algorithms as well. However, in empirical applications with small samples, the general setup of residual-based MBB
and drawing sample from the empirical distribution of the residuals is expected to be more robust.
% (the NGML score conditions being verified under the i.i.d.\ design; see Appendix~\ref{subsec:proofBSvalidity})
Let $\hat{B}_T^{\ast}=B(\hat{\beta}_T^{\ast})\operatorname{diag}(\hat{\sigma}_T^{\ast})$
denote the bootstrap analog of the impact-matrix estimator $\hat{B}_T$, obtained from
Algorithm~\ref{subsec:RMBB}, and define the studentized bootstrap statistic
\begin{equation}
Q_T^{\ast} \coloneqq \hat{\Sigma}_{B_T}^{-1/2}\,T^{1/2}\operatorname{vec}\big(\hat{B}_T^{\ast}-\hat{B}_T\big),
\end{equation}
where $\hat{\Sigma}_{B_T}$ is a consistent estimator of the covariance matrix $\mathcal{V}_B$, and $\hat{B}_T$ is estimated impact matrix. The distribution of
$Q_T^{\ast}$ conditional on the observed data is denoted by
$\hat{G}_T^{\ast}(\cdot)$. It is used to approximate the sampling distribution
$G_T(\cdot)$ of the corresponding sample statistic
\begin{equation}
 Q_T \coloneqq \hat{\Sigma}_{B_T}^{-1/2}\,T^{1/2}\operatorname{vec}\big(\hat{B}_T-B_0\big),
\end{equation}
which, under valid identification, is asymptotically standard normal by \eqref{eq:asymp} and \eqref{eq:Bblock_normal}. The distribution $\hat{G}_T^{\ast}$ can be estimated with arbitrary precision from the bootstrap replications
$\{Q_{T,b}^{\ast}\}_{b=1}^{N}$ of Algorithm~\ref{subsec:RMBB} by the empirical
distribution function:
\begin{equation}\label{eq:edf}
  \hat{G}_{T,N}^{\ast}(x)
  \coloneqq \frac{1}{N}\sum_{b=1}^{N}
  \mathbf{1}\!\big(\operatorname{vec}(Q_{T,b}^{\ast})\le x\big),
  \qquad x\in\mathbb{R}^{n^{2}},
\end{equation}
where $\mathbf{1}(\cdot)$ is the indicator function. By the Glivenko--Cantelli theorem, $\sup_x\big|\hat{G}_{T,N}^{\ast}(x)-\hat{G}_T^{\ast}(x)\big|
\xrightarrow{\text{a.s.}}0$ conditionally on the data as $N\to\infty$.
The following proposition and corollary establish the asymptotic validity of this bootstrap
approximation; under valid identification $\hat{G}_T^{\ast}$ converges to the
standard normal, the property exploited by the diagnostic of
Section~\ref{sec:BSdiag}.

\begin{proposition}\label{prop:BSvalidity}
     Consider the NGML estimator, $\hat{\lambda}_T = (\hat{\Pi}_T, \hat{\beta}_T, \hat{\sigma}_T, \hat{\theta}_T)$, defined in \eqref{eq:NGMLestimator} where, $\hat{B}_T = B(\hat{\beta}_T)\operatorname{diag}(\hat{\sigma}_T)$, and its bootstrap analog
     $\hat{\lambda}_T^{\ast}$ (and $\hat{B}_T^{\ast}$), defined in Algorithm \ref{subsec:RMBB}. Under Assumptions~\ref{as:stability}--\ref{as:gau}, the results \eqref{eq:asymp} and 
     \eqref{eq:fisher} with the regularity conditions of Appendix~\ref{subsec:regcond}, as $T \rightarrow \infty$:
       \begin{equation*}
        \hat{\Sigma}^{-1/2}_{\lambda_T}T^{1/2}(\hat{\lambda}_T^* - \hat{\lambda}_T) 
        \xrightarrow{d^*}_p\  \mathcal{N}(\mathbf{0}_{\dim (\lambda)}, I_{\dim (\lambda)}), \ \text{in probability}
       \end{equation*}  
        where, $\hat{\Sigma}_{\lambda_T}$ is a consistent estimator of the asymptotic covariance 
        matrix of $\hat{\lambda}_T$.     
\end{proposition}
\begin{proof}
  See Appendix~\ref{subsec:proofBSvalidity}.
\end{proof}
\begin{corollary}\label{cor:Bblock}
Under Assumptions~\ref{as:stability}--\ref{as:gau} (in particular, at most one
Gaussian structural shock) and the regularity conditions of Appendix~\ref{subsec:regcond}
holding for the autoregressive and impact parameters $\psi$ (Proposition~\ref{prop:onegaussian}), the bootstrap impact-matrix estimator satisfies, as
$T\to\infty$,
\begin{equation*}
  \hat{\Sigma}^{-1/2}_{B_T}\,T^{1/2}\operatorname{vec}\bigl(\hat{B}_T^{\ast}-\hat{B}_T\bigr)
  \;\xrightarrow{d^*}_p\;\mathcal{N}\bigl(\mathbf{0}_{n^2},\,I_{n^2}\bigr),
  \qquad\text{in probability},
\end{equation*}
where $\hat{\Sigma}_{B_T}$ is a consistent estimator of the robust covariance $\mathcal{V}_B$. 
\end{corollary}

Proposition~\ref{prop:BSvalidity} is stated for the full parameter vector
$\lambda$. When one structural shock is Gaussian, the regularity conditions fail
for $\lambda$, because the information about the Gaussian shock's shape
parameters is singular. They continue to hold, however, for the autoregressive
and impact parameters $\psi$ (Proposition~\ref{prop:onegaussian}).
Corollary~\ref{cor:Bblock} states the corresponding result for the impact matrix,
which is the one the diagnostic uses. Throughout, $\xrightarrow{p^*}_p$ and $\xrightarrow{d^*}_p$ denote convergence
in probability and in distribution under the bootstrap measure, conditional on
the data, in probability; see Appendix~\ref{appendix:a}.

%% file: 04_BSdiag.tex
This section develops the bootstrap diagnostic for the identification conditions of a
non-Gaussian SVAR. Proposition~\ref{prop:BSvalidity} and Corollary~\ref{cor:Bblock} established that, under valid
identification, the studentized bootstrap statistic $Q_T^{\ast}$ is asymptotically standard
normal; equivalently, its bootstrap distribution satisfies
\begin{equation}\label{eq:BSconsist}
  \sup_{x\in\mathbb{R}^{n^2}}\bigl|\hat{G}_T^{\ast}(x)-\Phi(x)\bigr|
  \xrightarrow{p^{\ast}}_{p}0,\qquad T\to\infty,
\end{equation}
where $\Phi$ is the standard normal distribution function. The
diagnostic exploits the converse: when identification fails, $\hat{G}_T^{\ast}$
departs from normality, so a test of the normality of the bootstrap replications
detects the failure.

The null hypothesis $\mathcal{H}_0$ is strong identification: Assumptions~\ref{as:stability}--\ref{as:gau} hold, i.e. at
most one shock is Gaussian. The alternative $\mathcal{H}_1$ contains the processes with two or more
Gaussian shocks. Section~\ref{sec:WKId} also considers processes in which a second
shock approaches Gaussianity at rate $T^{-1/2}$. Along such sequences the normal approximation to the
distribution of $\hat{B}_T$ fails. We report the rejection frequency of the diagnostic
there as a measure of its sensitivity to weak identification.

The mechanism underlying the diagnostic is as follows. Under valid identification and the maintained
regularity conditions, the log-likelihood has a unique maximum. The NGML
estimator then concentrates at rate $T^{1/2}$, and the studentized bootstrap
replications are asymptotically standard normal. When two or more shocks are
Gaussian, the log-likelihood is flat along rotations within the Gaussian
subspace, and hence the affected columns of $\hat{B}_T$ are not identified (Remark~\ref{rem:identification_failure}).
Instead, the replications follow the non-degenerate distribution induced by the
unidentified rotation, which is not normal (Lemma~\ref{lemma:altlaw}). Given
Assumptions~\ref{as:stability}--\ref{as:indep} and the regularity conditions in Appendix \ref{subsec:regcond}, a
rejection therefore indicates identification failure, that is, two or more
Gaussian shocks.

To measure the deviation of $\hat{G}^{\ast}_T$ from $\Phi$, we construct the statistic:
\begin{equation}\label{eq:teststat}
  d^{*}_{T,M}(x)\coloneqq M^{1/2}\,\hat{\Omega}_T^{-1/2}(x)
  \bigl(\hat{G}^{*}_{T,M}(x)-\Phi(x)\bigr),
\end{equation}
\begin{align}
  d^{*}_{T,M}(x)
  &= M^{1/2}\,\hat{\Omega}^{-1/2}_{T}(x)
     \Bigl(\hat{G}^{*}_{T,M}(x) - \hat{G}^{*}_{T}(x)\Bigr)
  \notag \\
  &+ M^{1/2}\,\hat{\Omega}^{-1/2}_{T}(x)
     \Bigl(\hat{G}^{*}_{T}(x) - \Phi(x)\Bigr).
  \label{eq:decomp2}
\end{align}
where $\hat{G}^{*}_{T,M}$ is the empirical distribution function
\eqref{eq:edf} computed from $M$ of the $N$ stored replications
$\{Q^{\ast}_{T,b}\}$, and $\hat{\Omega}_T(x)$ is a consistent estimator of
$\Omega_T(x) \coloneqq \hat{G}^{*}_{T}(x)\bigl(1 - \hat{G}^{*}_{T}(x)\bigr)$.

At finite $T$, the bootstrap distribution $\hat{G}_T^{\ast}$ is only approximately
standard normal: under the null, it differs from $\Phi$ by an Edgeworth
term of order $T^{-\rho}$ (Condition~\ref{cond:BErate}), and this discrepancy is a
function of the observed data. The statistic decomposes into a \emph{resampling}
term, $M^{1/2}\hat{\Omega}_T^{-1/2}(\hat{G}_{T,M}^{\ast}-\hat{G}_T^{\ast})$, and a
\emph{centering} term, $M^{1/2}\hat{\Omega}_T^{-1/2}(\hat{G}_T^{\ast}-\Phi)$ in \eqref{eq:decomp2}, which makes the problem
transparent. If $M \to \infty$ for fixed $T$, the empirical distribution recovers
$\hat{G}^{*}_{T}$ exactly, and the test checks whether $\hat{G}^{*}_{T} = \Phi$,
which does not hold at a finite $T$. The centering term, of order $M^{1/2}T^{-\rho}$,
then diverges. The statistic reflects the $O(T^{-\rho})$ approximation error
rather than identification failure, and the test over-rejects under the null.
Moreover, since $\hat{G}^{*}_{T}$ depends on the data, conditioning subsequent
inference on such a test induces pre-testing bias \citep{roth_pretest_2022}, particularly for the impulse response functions.

Following \citet{cavaliere_bootstrap_2025}, we therefore let $M$ and $T$ diverge
\emph{jointly}, with $MT^{-2\rho} \to 0$ for some $\rho>0$.\footnote{Under the
sequential regime in which $M\to\infty$ after $T\to\infty$ the pre-testing bias
likewise vanishes asymptotically; we adopt the joint regime because the practitioner can better
control $M$ than the sample size.} This requirement is exactly
$M^{1/2}T^{-\rho}\to0$, which forces the centering term to vanish. The asymptotic
randomness of $d^{\ast}_{T,M}$ then derives solely from the resampling term, which is
free of the data, so the statistic is pivotal and the test attains correct asymptotic
size (Remark~\ref{rem:pivotal}).

\begin{proposition}\label{prop:bootstraptest}
Let Assumptions~\ref{as:stability}--\ref{as:gau} and the regularity conditions of Appendix~\ref{subsec:regcond} hold, and let $d^{*}_{T,M}(x)$ be given by \eqref{eq:teststat}.
As $T,M\to\infty$ jointly with $MT^{-2\rho} \to 0$ for some $\rho>0$, 
\begin{enumerate}
\item[\textbf{(i)}] Under the null $\mathcal{H}_0$ of valid identification, i.e., at most one structural
shock is Gaussian, for each $x\in\mathbb{R}^{n^2}$,
\begin{equation}\label{eq:testdist}
  d^{*}_{T,M}(x)\xrightarrow{d^{*}}_{p}\mathcal{N}(0,1);
\end{equation}
\item[\textbf{(ii)}] Under $\mathcal{H}_1$ where there are two or more
Gaussian shocks,
\begin{equation}\label{eq:altdist}
  \sup_{x\in\mathbb{R}^{n^2}}\bigl|d^{*}_{T,M}(x)\bigr| \;\xrightarrow{p}\; \infty .
\end{equation}
\end{enumerate}
\end{proposition}

\begin{proof}
    See Appendix~\ref{subsec:proofbootstraptest}.
\end{proof}

\begin{corollary}[Valid size-$\alpha$ diagnostic]\label{cor:sizetest}
Let $\mathcal{T}_{T,M}$ be any consistent test statistic for the normality of the
bootstrap replications, computed on the raw $\{\hat{B}^{\ast}_{T,b}\}_{b=1}^{M}$ or
equivalently the studentized $\{Q^{\ast}_{T,b}\}_{b=1}^{M}$, and let $c_{1-\alpha}$
be the $(1-\alpha)$-quantile of its standard ($\chi^{2}$) null distribution. Under
the conditions of Proposition~\ref{prop:bootstraptest}(i), the rule
$\phi_{T,M}\coloneqq\mathbf{1}(\mathcal{T}_{T,M}>c_{1-\alpha})$ has asymptotic size
$\alpha$ and, by Proposition~\ref{prop:bootstraptest}(ii), power tending to one
against all fixed alternatives in $\mathcal{H}_1$.
\end{corollary}

\begin{proof}
    See Appendix~\ref{subsec:proofcorroboottest}.
\end{proof}

\begin{remark}[Asymptotic pivotality and ancillarity]\label{rem:pivotal}
The limit $\mathcal{N}(0,1)$ obtained above is free of the data-generating
parameters, so $d^{\ast}_{T,M}$ is {asymptotically pivotal}, and the test has
asymptotically correct size. The surviving resampling term,
$M^{1/2}\hat{\Omega}_T^{-1/2}(\hat{G}^{\ast}_{T,M}-\hat{G}^{\ast}_T)$, has,
conditional on the data, an asymptotic distribution that does not depend on the data, so the
diagnostic is asymptotically ancillary. 
\end{remark}
% By
% \citet[Theorem~3.1]{cavaliere_bootstrap_2025}, a diagnostic with this property does not
% distort subsequent inference asymptotically, unlike a conventional residual normality pre-test.
\begin{remark}[The rate $\rho$ and the choice of $M$]\label{rem:Mchoice}
The choice of $M$ relative to $T$ is governed by the rate $\rho$ in the
joint requirement $MT^{-2\rho}=o_p(1)$. When the bootstrap admits an Edgeworth
expansion, $\rho=\frac{1}{2}$ (Condition~\ref{cond:BErate}) and the requirement reduces
to $M/T\to0$. If $M$ is too large relative to $T^{2\rho}$, the centering term is
not negligible anymore, and $d^{\ast}_{T,M}$ fails to converge to $\mathcal{N}(0,1)$ even
when the bootstrap is consistent, leading to false rejections
\citep{angelini_bootstrap_2022}. $M/T$ should therefore be kept small in finite
samples.
\end{remark}

\begin{remark}[Implementation: a normality test on the replications]\label{rem:DH}
Since, by Corollary~\ref{cor:Bblock} and Proposition~\ref{prop:bootstraptest}, the
studentized replications $\{Q^{\ast}_{T,b}\}_{b=1}^{M}$ are, conditional on the
data, asymptotically i.i.d.\ standard normal under $\mathcal{H}_0$, any consistent normality test, multivariate,
or univariate version element-wise, inherits the implication of the proposition:
its standard ($\chi^{2}$) null distribution controls asymptotic size, while the
statistic diverges under $\mathcal{H}_1$. Because the Doornik--Hansen and Jarque--Bera statistics are invariant to nonsingular affine maps (respectively to
location and scale), they take the same value on the raw replications
$\{\hat{B}^{\ast}_{T,b}\}$ and on the studentized $\{Q^{\ast}_{T,b}\}$. We apply
them to the raw $\{\hat{B}^{\ast}_{T,b}\}$, the studentized version being equivalent
and the size properties of Corollary~\ref{cor:sizetest} transferring accordingly. We
use the multivariate omnibus\footnote{See
\citet{fresoli_bootstrap_2022,angelini_bootstrap_2022} for the Doornik--Hansen test
in this context; the univariate test is applied to each sequence
$\{\hat{B}^{\ast}_{i,j,T,b}\}_{b=1}^{M}$, $i,j=1,\dots,n$.} test of \citet{doornik_omnibus_2008} and the univariate
test of \citet{jarque_test_1987}; henceforth, DH and JB tests respectively.
\end{remark}

%% file: 05_MC.tex
\subsection{SVAR design} This section studies the finite-sample properties of the bootstrap-based approach to detect identification failure due to two or more Gaussian shocks. Consider the SVAR model in \eqref{eq:companionVAR}, with $n=3$ and $p=1$,

\begin{equation}
y_t = \Pi_1 y_{t-1} + u_t, \quad {t = 1,2,\ldots,T}, \quad u_t = B \varepsilon_t,
\end{equation}
where $y_t = (y_{1,t}, y_{2,t}, y_{3,t})'$, $\Pi_1$ is a $3 \times 3$ matrix of autoregressive coefficients, $u_t = (u_{1,t}, u_{2,t}, u_{3,t})'$ is 
the vector of reduced-form innovations, $B$ is the $3 \times 3$ normalized impact matrix, 
and $\varepsilon_t = (\varepsilon_{1,t}, \varepsilon_{2,t}, \varepsilon_{3,t})'$ is the vector of structural shocks. And,

\begin{equation}
    \Pi_1 = \begin{pmatrix}
    0.30 & 0.60 & -0.40 \\
    -0.50 & 0.20 & 0.10 \\
    0.70 & 0.10 & -0.30
    \end{pmatrix}, \quad
    B(\beta_0) = \begin{pmatrix}
    1.00 & 0.20 & 0.00 \\
    0.63 & 1.00 & -0.60 \\
    -0.51 & 0.55 & 1.00
    \end{pmatrix}, \quad
    \label{eq:model_param}
\end{equation}

As noted in \eqref{eq:setup_Bmatrix}, the normalized impact matrix $B$ is nonsingular and column-wise normalized to have a unit diagonal, $\beta_0$ is a vector containing the off-diagonal
elements of $B$, of 
dimension $n(n-1) = 6$. Hence, $B = B(\beta_0)$ is parameterized as a function of $\beta_0$. We focus on the impact-matrix estimator, which collects the on-impact responses to one-standard-deviation shocks:
\begin{equation}
    \hat{B}_T = B(\hat{\beta}_T) \operatorname{diag}(\hat{\sigma}_T),
\end{equation}
where $\hat{\sigma}_T$ is the estimate of the standard deviations of structural shocks.
The VAR is estimated by least squares. The structural parameters
$(\beta_0, \sigma_0)$ and the shape parameters of the non-Gaussian densities are
estimated by NGML. Inference uses the residual-based moving block bootstrap of
Section~\ref{sec:BSINF}, with block length
$\ell = \lfloor 5.03\,T^{1/4} \rfloor$ \citep{jentsch_proxy_2016} and
$N = 1999$ bootstrap replications. All results are based on $N_S = 1000$ Monte
Carlo simulations, and confidence intervals have nominal level $90\%$.

\subsection{Distributional specifications}\label{sec:dist_spec}
 We consider three different specifications for
 the structural shocks' distributions. Specifically:

 \begin{itemize}
    \item \textbf{ICA 0:} $\varepsilon_{1,t} \sim NIG(0.34,0.57), \varepsilon_{2,t} \sim NIG(0.51,1.00)$, and
     $\varepsilon_{3,t} \sim NIG(0.83,1.79)$.

    \item \textbf{ICA 1:} $\varepsilon_{1,t} \sim NIG(0.34,0.57), \varepsilon_{2,t} \sim NIG(0.51,1.00)$, and
     $\varepsilon_{3,t} \sim \mathcal{N}(0,1)$.

    \item \textbf{No ICA:} $\varepsilon_{1,t} \sim NIG(0.34,0.57), \varepsilon_{2,t} \sim \mathcal{N}(0,1)$, and $\varepsilon_{3,t} \sim \mathcal{N}(0,1)$.
    \label{dis spec}
\end{itemize}

 Though the NIG distribution is defined by four parameters, we fix the location and asymmetry\footnote{We allow for asymmetric shocks for robustness in the empirical application.} parameters at $\mu = \gamma = 0$, to center the respective structural shocks around zero. Hence, the above-mentioned parameterization is of the form: $NIG(\alpha, \delta)$, see \eqref{eq:NIGpdf}. 
 For a symmetric NIG, $\operatorname{Var}(\varepsilon_{i,t}) = \delta_i/\alpha_i$ and the excess kurtosis is $3/(\alpha_i \delta_i)$. Since our object of interest is the impact matrix, computed from the estimates of $\hat{\beta}_T$ and $\hat{\sigma}_T$, the standard deviations $\sigma_i$ of
 the structural shocks $\varepsilon_t$ are:
 \begin{itemize}
    \item \textbf{ICA 0:} $\sigma_1 = 1.29, \sigma_2 = 1.40, \sigma_3 = 1.47$.
    \item \textbf{ICA 1:} $\sigma_1 = 1.29, \sigma_2 = 1.40, \sigma_3 = 1.00$.
    \item \textbf{No ICA:} $\sigma_1 = 1.29, \sigma_2 = 1.00, \sigma_3 = 1.00$.
\end{itemize}
 such that $B_0 = B(\beta_0)\operatorname{diag}(\sigma_0)$, with $\sigma_0 = (\sigma_1, \sigma_2, \sigma_3)'$ as listed above.
 From Assumptions~\ref{as:indep} and \ref{as:gau}, we can infer that the first two
 specifications satisfy the identification conditions. However, in the third specification, two of the independent structural shocks are Gaussian, thereby violating the identification conditions. In \textbf{ICA 0} and \textbf{ICA 1} the likelihood is correctly specified: the number of shocks assigned the standard normal density in estimation matches the number of Gaussian shocks in the data-generating process. In \textbf{No ICA} the likelihood continues to declare a single Gaussian shock while the data contain two, so the second shock is fitted with a NIG density whose shape parameters drift towards the Gaussian limit. If both Gaussian shocks were assigned the Gaussian density in estimation, the likelihood would be exactly flat along their rotation and the estimator would not be defined. Hence, the design above lets the failure of identification manifest itself.

\subsection{Insufficient verification with reduced-form residuals} \label{subsec:insuff_test} Given the relation $u_t = B \varepsilon_t$, it can be tempting to verify the assumptions
of non-Gaussianity of structural shocks by testing $\hat{u}_t$ for non-Gaussianity. 
Indeed, many studies, including \citet{lanne_identification_2017}, \citet{andrade_higher-order_2025}, \citet{jarocinski_deconstructing_2020},
use normality tests on the reduced-form residuals as a pre-test for valid identification. Under the null of non-Gaussian structural
shocks, the reduced-form innovations naturally inherit the non-Gaussianity. However, as noted earlier,
given that valid identification allows \textit{at most} one Gaussian shock, the presence of non-Gaussianity in the reduced-form innovations
does not preclude two or more Gaussian shocks. 
\input{Tables/unireducedformtest_v2.tex}
\input{Tables/reducedformtest_v2.tex}

Tables~\ref{tab:unireducedformtest} and \ref{tab:reducedformtest} report the empirical rejection frequencies of the multivariate and univariate DH and JB tests for normality 
of $\hat{u}_t$.

Table~\ref{tab:unireducedformtest} shows that, under both \textbf{ICA 1} and
\textbf{No ICA}, the rejection frequencies of the univariate normality tests
increase with the sample size. The row \textit{all three} reports the
conventional conditioning rule, which retains a sample only when all three
residual series reject normality. Under \textbf{No ICA} this rule retains
$0.19$, $0.37$ and $0.53$ of the simulations at $T = 150$, $300$ and $500$, so
the share of unidentified samples that pass the pre-test rises with $T$. The pre-test therefore fails to
screen out identification failure while discarding many identified samples,
and conditioning inference on it induces pre-testing bias.
Table~\ref{tab:reducedformtest} shows that the multivariate test rejects
normality of $\hat{u}_t$ in essentially every simulation, at every sample size
and under all three specifications. A multivariate normality test on the reduced-form residuals therefore carries
no information about the number of Gaussian shocks.

\subsection{Estimates and probability coverages}
Tables~\ref{tab:EstandCovICA1} and \ref{tab:EstandCovnoICA} report the estimates of the impact matrix, from the two distributional specifications \textbf{ICA 1} and \textbf{No ICA} respectively\footnote{Since we are concerned with identification failure due to an additional Gaussian shock, in this section we show results pertaining only to specifications \textbf{ICA 1} and \textbf{No ICA}. Results for the specification \textbf{ICA 0}, where all three structural shocks are non-Gaussian are in Appendix~\ref{subsec:resultsICA0}.}, with sample size $T = 150, 500$. The tables also report the probability coverage of
asymptotic, bootstrap percentile and bootstrap studentized confidence
intervals.

 Table~\ref{tab:EstandCovICA1} shows that for the specification with one Gaussian shock, the estimates $\hat{B}_T$ are centered at $B_0$. The asymptotic intervals cover close to the nominal 90\% level. The coverage of bootstrap confidence intervals (CIs) improves with sample size. The bootstrap standard errors closely match those from NGML estimation, in line with the Corollary~\ref{cor:Bblock}.

In contrast, Table~\ref{tab:EstandCovnoICA} shows that when two Gaussian shocks are present, the estimates of $\hat{B}_T$ becomes inconsistent and the empirical CI coverages fall well below the nominal level. The average percentile coverage is $0.69$ at $T = 150$ and deteriorates to $0.62$ at $T = 500$, so the distortion does not vanish as the sample grows. Moreover, the failure is confined to the columns associated with the two Gaussian shocks: at $T = 500$ the percentile coverages for the three elements of the first column, the on-impact responses to the single non-Gaussian shock $\varepsilon_{1,t}$, are $0.83$, $0.82$ and $0.84$, against $0.39$ to $0.78$ for the other two columns. Also, the column matching described in the notes to Table~\ref{tab:EstandCovnoICA} realigns around $60\%$ of the simulations under \textbf{No ICA}, against $0.5\%$ under \textbf{ICA 1}. The mean Amari distance between $\hat{B}_T$ and $B_0$  is $0.21$, $0.19$ and $0.18$ at $T = 150, 300$ and $500$ under \textbf{No ICA}, whereas under \textbf{ICA 1} it halves from $0.08$ to $0.04$.

 These findings are consistent with the literature on partial identification of non-Gaussian structural shocks in the presence of two or more Gaussian shocks \citep{maxand_identification_2020}. However, inference from such partially identified models requires prior knowledge of which shocks are Gaussian, which is generally infeasible in empirical applications. Nonetheless, the bootstrap diagnostic can detect such an identification failure, and its univariate form additionally identifies {which} columns of $\hat{B}_T$ are affected.

\input{Tables/EstandCovICA1_v2.tex}

\input{Tables/EstandCovnoICA_v2.tex}
\subsection{Bootstrap diagnostic}
We consider the multivariate and univariate (over each element of $\hat{B}_T$) DH normality test for the sequences,
\begin{equation*}
 {\{\hat{B}_{T,1}^*, \hat{B}_{T,2}^*, \ldots, \hat{B}_{T,M}^*\}}_{s}, s = 1, 2, \ldots, S, \ b = 1,2,\ldots,M 
\end{equation*}
where $M$ is the length of a sequence, and $S$ is the number of sequences/tests. As noted in Remark~\ref{rem:Mchoice}, $M$ must satisfy $M/T \to 0$. We take rules\footnote{The two univariate tests differ only in the constant $a$, so both satisfy
$M/T \to 0$. The choice of constant reflects how well each test's null
distribution is approximated in small samples. The DH test transforms sample
skewness and kurtosis so that its $\chi^2_2$ approximation is accurate in small
samples \citep{doornik_omnibus_2008}. It can therefore be applied to short
sequences, and a small $a$ keeps the centering term in \eqref{eq:decomp2}
small. The JB test compares the untransformed moments with their asymptotic
$\chi^2_2$ distribution. This approximation is poor in short sequences \citep{jarque_test_1987}, so it requires a larger $a$. Moreover, the constant $a$ is a tuning parameter: any fixed $a > 0$ satisfies the rate
condition, so the asymptotic properties of the diagnostic do not depend on it. Its choice affects only finite-sample properties.} of the form $M = \lceil a\,T^{3/5} \rceil$, which meet that condition for any fixed $a > 0$, and report three of them: $a = 1/3$ and $a = 1/2$ for the DH tests, and $a = 5/3$ for the univariate JB test.

We consider the empirical distribution function of the $p$-values, $p^*_{M,T,s}$, $s = 1, 2, \ldots, S$, across the simulations.
Under the null hypothesis, for continuous test statistics, it is known that the $p$-values are uniformly distributed. Specifically, under valid identification, \\
$p^*_{M,T,s} \xrightarrow{d^*}_p \mathbb{U}(0,1)$, as $S \to \infty$ and the average rejection frequencies can be computed as $\pi^*_{M,T,S}(x) = \frac{1}{S} \sum_{s=1}^S \mathbf{1}(p^*_{M,T,s} \leq x)$. Hence, under valid specification of \emph{at most} one Gaussian shock, $\pi^*_{M,T,S}(x) \xrightarrow{p} x$, as $S \to \infty$.
However, under the alternative hypothesis, the $p$-values are stochastically smaller than $\mathbb{U}(0,1)$ and hence, $\pi^*_{M,T,S}(x) > x$ for some $x \in (0,1)$, see
\citet{hung_behavior_1997,tang_general_2021}.

\subsubsection{Multivariate diagnostic} 

\input{Figures/fanchart_v2.tex}
Figure~\ref{fig:fanchart500} plots the empirical distribution function of the
$p$-values $p^*_{M,T,s}$ of the multivariate normality test under \textbf{ICA 0},
\textbf{ICA 1} and \textbf{No ICA}, for $T = 500$ and
$M = \lceil (1/2)T^{3/5} \rceil$.\footnote{The corresponding plots for $T = 300$
are in Appendix~\ref{subsec:resultsICA0} (Figure~\ref{fig:fanchart300}).} The
$x$-axis is the nominal level and the $y$-axis the rejection frequency; the
dashed gray line is the $45^{\circ}$ line. The solid blue (orange) line is the
median (mean) rejection frequency across Monte Carlo simulations, and the dark
and light shaded areas are the $50\%$ and $90\%$ bands. Under \textbf{ICA 0} and
\textbf{ICA 1} the $p$-values are close to uniform. Under \textbf{No ICA}, with two Gaussian shocks, they are stochastically smaller
than uniform, so the rejection frequency exceeds the nominal level at every
level. Table~\ref{tab:mvnormtest} in Appendix \ref{subsec:dhmultit300} reports the rejection frequencies at the $5\%$ and
$1\%$ levels for both choices of $M$.

\subsubsection{Univariate diagnostic} The diagnostic can also be applied to each element of the impact matrix, using the univariate DH
test on each element of $\hat{B}_T$. Results for the univariate JB test are in
Appendix~\ref{subsec:JBuni}.
Figures~\ref{fig:uni_NIGOGT500} and \ref{fig:uni_NIGWGT500} show the empirical
distribution function of $p^*_{M,T,s}$ for the sequences
$\{\hat{B}^*_{ij,T,b}\}_{b=1}^{M}$, $i, j = 1, 2, 3$, under \textbf{ICA 1} and
\textbf{No ICA}, for $T = 500$ and $M = \lceil (1/2)T^{3/5} \rceil$. The $j$-th
row of panels shows the three elements of the $j$-th column of $\hat{B}_T$,
that is, the impact responses to $\varepsilon_{j,t}$.

Under \textbf{ICA 1}, the empirical distribution functions lie
close to the $45^{\circ}$ line in all nine panels, so the $p$-values are
approximately uniform. Under
\textbf{No ICA}, the distribution functions lie well above the $45^{\circ}$ line
only in the second and third rows, which correspond to the two Gaussian shocks
$\varepsilon_{2,t}$ and $\varepsilon_{3,t}$. The first row, for the non-Gaussian
shock, stays close to the $45^{\circ}$ line. This matches the observation pertaining to partial
identification in Table~\ref{tab:EstandCovnoICA}, where the probability coverage for the
column of $\varepsilon_{1,t}$ remains close to the nominal $90\%$, whereas the
columns of the two Gaussian shocks are not consistently estimated.
\input{Figures/uni_NIGOGT500_v2.tex}

\input{Figures/uni_NIGWGT500_v2.tex}

\subsection{Weak identification}\label{sec:WKId}
We also consider the weak identification scenario, where the standard asymptotic approximation used for inference is invalid \citep{hoesch_locally_2024}. We define weak identification as a sequence of data-generating processes under \textbf{ICA 1} where one of the non-Gaussian structural shocks ($\varepsilon_{2,t}$) has excess kurtosis $\kappa_{2,T} \to 0$ as $T \to \infty$:
\begin{equation}
    \kappa_{2,T} = c/\sqrt{T}, \qquad c = 10,
    \label{eq:wkidp_drift}
\end{equation}
which gives $\kappa_{2,T} = 0.577$ at $T = 300$ and $0.447$ at $T = 500$, against $15.48$ for $\varepsilon_{1,t}$ and zero for $\varepsilon_{3,t}$. The rotation between $\varepsilon_{2,t}$ and $\varepsilon_{3,t}$ is identified
through fourth-order cumulants, and the information about it is of order
$T\kappa_{2,T}^2$. Along this sequence $T\kappa_{2,T}^2 = c^2$ is fixed, so the
rotation is not consistently estimated and the normal approximation for the
affected columns of $\hat{B}_T$ fails. We study whether the diagnostic detects such weak identification.

The variance of $\varepsilon_{2,t}$ is held at its \textbf{ICA 1} value, so $B_0$ is the same impact matrix as in Section~\ref{sec:dist_spec} and the estimates are comparable with Tables~\ref{tab:EstandCovICA1} and \ref{tab:EstandCovnoICA}. The following results are computed over $N_S = 1500$ Monte Carlo simulations at $T = 300$ and $T = 500$, with $N = 1999$ bootstrap replications per simulation, and nominal $90\%$ confidence intervals are maintained throughout. The diagnostic is applied at nominal significance level of $5\%$.

\subsubsection{Decision rule with the diagnostic}\label{subsubsec:wkidp_rule}
We recall that $p^*_{M,T,ij}$ denotes the $p$-value of the univariate test applied
to a sequence of $M$ bootstrap replications of the $(i,j)$-th element of
$\hat{B}_T$, $i, j = 1, \ldots, n$. Under valid identification this
$p$-value is uniformly distributed. However, for a given sample the $p$-values also depend on which $M$ of the $N$ bootstrap replications enter
the test, and this dependence does not in general vanish as $M$ and $T$
grow \citep{cavaliere_bootstrap_2025}, so two researchers analyzing the same sample may report different values.

We address this concern by forming $L$ sequences from disjoint blocks
of $M$ stored replications, which requires $LM \leq N$. Since the blocks
are disjoint, the resulting $p$-values $p^*_{M:l,T,ij}$, $l = 1, \ldots, L$,
are independent and identically distributed conditional on the data. As in
the previous subsection, their empirical distribution function is
\begin{equation}
  \pi^*_{M,T,L,ij}(x)
  \;=\; \frac{1}{L} \sum_{l=1}^{L}
        \mathbf{1}\!\left( p^*_{M:l,T,ij} \leq x \right),
  \qquad x \in (0,1),
  \label{eq:wkidp_picurve}
\end{equation}
Under valid identification
$\pi^*_{M,T,L,ij}(x) \xrightarrow{p} x$, whereas under the alternative hypothesis, the $p$-values
are stochastically smaller than $\mathbb{U}(0,1)$, so
$\pi^*_{M,T,L,ij}(x)$ lies above the $45^{\circ}$ line. We measure the
departure from uniformity for each element by the Kolmogorov--Smirnov
distance
\begin{equation}
  D_{ij} \;=\; \sqrt{L}\, \sup_{x \in (0,1)}
  \left| \pi^*_{M,T,L,ij}(x) - x \right| 
  \label{eq:wkidp_ks}
\end{equation}
For independent, uniform $p$-values, $D_{ij}$ converges to the supremum of a Brownian bridge, i.e., to the Kolmogorov
distribution, as $L \to \infty$. We therefore take critical values from
the null distribution of $D_{ij}$ at the adopted values of $(M, L)$. The diagnostic
rejects if
\begin{equation}
  \phi_{M,L} = \mathbf{1}\left\{ \max_{i,j} D_{ij}
  > c_{1-\alpha/n^2}(M,L) \right\}
  \label{eq:wkidp_report}
\end{equation}
equals one, where $c_{1-\alpha/n^2}(M,L)$ is the $(1-\alpha/n^2)$-quantile
of the null distribution. The Bonferroni adjustment over the $n^2 = 9$
elements of $\hat{B}_T$ bounds the asymptotic size of the decision rule by $\alpha$. With $\phi_r$ as the test outcome in simulation $r = 1, \ldots, N_S$, we report the rejection
frequency as $N_S^{-1} \sum_{r=1}^{N_S} \phi_r$. 

We set $M = \lceil (1/3)\,T^{3/5} \rceil$ and $L = \lfloor T/(2M) \rfloor$. The
bound on $L$ follows from \eqref{eq:decomp2}. All blocks are drawn from the same
bootstrap distribution $\hat{G}^{*}_{T}$, so the $L$ $p$-values share a common
departure from uniformity, $\Delta_{M,T} = O_p(M^{1/2}T^{-\rho})$, inherited
from the centering term. By averaging over the blocks, we can remove the resampling noise but
not $\Delta_{M,T}$, which the normalization in \eqref{eq:wkidp_ks} inflates to
$\sqrt{L}\,\Delta_{M,T}$. The critical values therefore remain valid only if
\begin{equation}
  L\,M\,T^{-2\rho} \to 0 ,
  \label{eq:wkidp_rate}
\end{equation}
that is, $LM/T \to 0$ for $\rho = 1/2$. This is the rate condition of
Remark~\ref{rem:Mchoice}, applied to all $LM$ replications used in the decision
rather than to a single sequence of length $M$.

% The sequence length follows the smaller of the two rules of the previous
% subsection, $M = \lceil (1/3)\,T^{3/5} \rceil$, and the number of blocks
% is $L = \lfloor T/(2M) \rfloor$. This gives $(M, L) = (11, 13)$ at
% $T = 300$ and $(14, 17)$ at $T = 500$, so the rule uses $143$ and $238$ of
% the $N = 1999$ stored replications. Averaging over $L$ blocks reduces the
% resampling component of \eqref{eq:decomp2} by a factor $\sqrt{L}$ but
% leaves the centering component unaffected. The rate condition of
% Remark~\ref{rem:Mchoice} therefore applies to the total number of draws
% $LM$ rather than to $M$.

We compare the diagnostic with the conventional conditioning rule based on the residual normality pre-test. The rule retains a sample only when the
univariate test rejects normality for all three reduced-form residual
series $\hat{u}_{i,t}$, $i = 1, 2, 3$, at the same level $\alpha$. This corresponds to the row \textit{all three} of
Table~\ref{tab:unireducedformtest}.
\subsubsection{Results}\label{subsubsec:wkidp_results}

Table~\ref{tab:weakiduni} reports the empirical rejection frequencies of the univariate
diagnostic for each element of $\hat{B}_T$. Under weak identification the
rejections concentrate in the columns of $\varepsilon_{2,t}$ and $\varepsilon_{3,t}$, as their rotation is weakly identified, while the column of $\varepsilon_{1,t}$ keeps close to
nominal size. The rejection frequencies in the affected columns rise with $T$.
\input{Tables/weakiduni_v2.tex}
Table~\ref{tab:wkidp_size_power} reports the overall rejection frequencies of
the decision rule. Under \textbf{ICA 1} the rule moderately over-rejects in finite samples,
at $0.121$ for $T = 300$ and $0.089$ for $T = 500$ against a nominal $0.05$,
and the over-rejection declines with $T$. Under weak identification the
rejection frequency is $0.595$ and $0.763$, close to the $0.641$ and $0.779$
under \textbf{No ICA}.

In finite samples, the diagnostic responds to weak identification much as it
does to identification failure. This is an expected response, since the
normal approximation to the distribution of $\hat{B}_T$ is invalid in both
cases. Although $B_0$ is identified, standard inference is unreliable as noted by \citet{hoesch_locally_2024}, and a rejection by the diagnostic signals such failure.
\input{Tables/wkidp_retention_v2.tex}

With the residual pre-test, the share of samples it
flags as not identified falls with $T$ in all three specifications: from $0.399$ to
$0.204$ under \textbf{ICA 1}, from $0.727$ to $0.603$ under weak identification,
and from $0.615$ to $0.461$ under \textbf{No ICA}. The reason is that each
residual is a mixture of all shocks and inherits the non-normality of the most
non-Gaussian one, $\varepsilon_{1,t}$. As $T$ grows, the univariate tests detect
that non-normality in the residuals more often, irrespective of whether the rotation between
$\varepsilon_{2,t}$ and $\varepsilon_{3,t}$ is identified. The pre-test's rejection of normality therefore cannot be a valid signal of identification.

%% file: Tables/unireducedformtest_v2.tex
\begin{table}[H]
\centering
\renewcommand{\arraystretch}{1}
\scalebox{0.9}{
    \begin{tabular}{|c|ccc|ccc|}
     \hline
                & \multicolumn{6}{c|}{\textbf{Specification}} \\
                & \multicolumn{3}{c|}{\textbf{ICA 1}} & \multicolumn{3}{c|}{\textbf{No ICA}} \\
        \textbf{Sample Size} & \textbf{T = 150} & \textbf{T = 300} & \textbf{T = 500} & \textbf{T = 150} & \textbf{T = 300} & \textbf{T = 500}\\
        \hline
        \textit{DH} & & & & & & \\
        \rule{0pt}{2.5ex}
        $\hat{u}_{1,t}$ & 1.00 & 1.00 & 1.00 & 1.00 & 1.00 & 1.00\\
        $\hat{u}_{2,t}$ & 0.85 & 0.97 & 1.00 & 0.43 & 0.62 & 0.78\\
        $\hat{u}_{3,t}$ & 0.41 & 0.59 & 0.80 & 0.30 & 0.46 & 0.59\\
        \textit{all three} & 0.37 & 0.59 & 0.80 & 0.19 & 0.37 & 0.53\\
        \hline
        \textit{JB} & & & & & & \\
        \rule{0pt}{2.5ex}
        $\hat{u}_{1,t}$ & 1.00 & 1.00 & 1.00 & 1.00 & 1.00 & 1.00\\
        $\hat{u}_{2,t}$ & 0.86 & 0.97 & 1.00 & 0.43 & 0.62 & 0.78\\
        $\hat{u}_{3,t}$ & 0.42 & 0.61 & 0.80 & 0.30 & 0.48 & 0.61\\
        \textit{all three} & 0.38 & 0.60 & 0.80 & 0.19 & 0.39 & 0.54\\
        \hline
    \end{tabular}
}
\caption{\scriptsize Empirical rejection frequencies of the univariate Doornik--Hansen and Jarque--Bera tests for normality of the reduced-form residuals, $\hat{u}_t$, for different sample sizes, ${T} = 150, 300, 500$, across ${N_S} = 1000$ Monte Carlo simulations, at the nominal 5\% level, for the two distributional specifications, \textbf{ICA 1} and \textbf{No ICA}, described in Section~\ref{sec:dist_spec}. The row \textit{all three} reports the frequency with which the test rejects for all three residual series simultaneously.}
\label{tab:unireducedformtest}
\end{table}

%% file: Tables/reducedformtest_v2.tex
\begin{table}[H]
\centering
\renewcommand{\arraystretch}{1}
\scalebox{1}{
    \begin{tabular}{|c|ccc|}
     \hline
                & \multicolumn{3}{c|}{\textbf{Specification}} \\
        \textbf{Sample Size} & \textbf{ICA 0} & \textbf{ICA 1} & \textbf{No ICA}\\
        \hline
        \rule{0pt}{2.5ex}
        \textbf{150} & 1.00 & 1.00 & 0.99\\
        \textbf{300} & 1.00 & 1.00 & 1.00\\
        \textbf{500} & 1.00 & 1.00 & 1.00\\
        \hline
    \end{tabular}
}
\caption{\scriptsize Empirical rejection frequencies of the multivariate Doornik--Hansen test for normality of the reduced-form residuals, $\hat{u}_t$, for different sample sizes, ${T} = 150, 300, 500$, across ${N_S} = 1000$ Monte Carlo simulations, at the nominal 5\% significance level, for the three distributional specifications described in Section~\ref{sec:dist_spec}.}
\label{tab:reducedformtest}
\end{table}

%% file: Tables/EstandCovICA1_v2.tex
\begin{table}[H]
\centering

    \renewcommand{\arraystretch}{0.9}

\scalebox{0.8}{

         \begin{tabular}{|c|c|cc|c|ccc|}
         \hline
         & & & & & & \multicolumn{2}{c|}{\textbf{Bootstrap Coverage}}\\

         & $\mathbf{B_{0}}$ & $\mathbf{\hat{B}_T}$ & $\mathbf{\hat{B}^*_T}$ & $\mathbf{90\%}$\textbf{CI} & \textbf{Asymp. Cov}. & \textbf{Studentized} & \textbf{Percentile} \\
         & (a) & (b) & (c) & (d) & (e) & (f) & (g) \\

         \hline
        \rule{0pt}{1ex}
                           & 1.29 & 1.26 & 1.23 & [0.95, 1.48] & 0.84 & 0.83 & 0.74 \\
                           && (0.19) & (0.19) &  & (0.61) & (0.63) & (0.52)\\
                           & 0.82 & 0.80 & 0.78 & [0.49, 1.05] & 0.90 & 0.82 & 0.85 \\
                           && (0.18) & (0.18) &  & (0.55) & (0.53) & (0.56)\\
                           & -0.66 & -0.64 & -0.63 & [-0.87, -0.37] & 0.90 & 0.81 & 0.86 \\
                           && (0.16) & (0.17) &  & (0.51) & (0.46) & (0.50)\\
                           & 0.28 & 0.27 & 0.26 & [0.12, 0.41] & 0.91 & 0.87 & 0.91 \\
                           && (0.08) & (0.09) &  & (0.25) & (0.24) & (0.28)\\
        $\mathbf{T = 150}$ & 1.40 & 1.37 & 1.33 & [1.04, 1.60] & 0.93 & 0.82 & 0.82 \\
                           && (0.19) & (0.21) &  & (0.60) & (0.57) & (0.58)\\
                           & 0.77 & 0.75 & 0.72 & [0.45, 1.00] & 0.89 & 0.81 & 0.85 \\
                           && (0.17) & (0.19) &  & (0.52) & (0.48) & (0.54)\\
                           & 0.00 & 0.00 & 0.00 & [-0.15, 0.17] & 0.88 & 0.85 & 0.91 \\
                           && (0.09) & (0.10) &  & (0.27) & (0.28) & (0.31)\\
                           & -0.60 & -0.58 & -0.55 & [-0.83, -0.26] & 0.86 & 0.85 & 0.87 \\
                           && (0.17) & (0.19) &  & (0.47) & (0.49) & (0.55)\\
                           & 1.00 & 0.97 & 0.93 & [0.74, 1.12] & 0.89 & 0.83 & 0.85 \\
                           && (0.11) & (0.14) &  & (0.33) & (0.32) & (0.36)\\
\hline
        \rule{0pt}{1ex}
                           & 1.29 & 1.28 & 1.28 & [1.10, 1.44] & 0.89 & 0.87 & 0.84 \\
                           && (0.11) & (0.11) &  & (0.36) & (0.36) & (0.33)\\
                           & 0.82 & 0.81 & 0.80 & [0.65, 0.95] & 0.89 & 0.85 & 0.85 \\
                           && (0.09) & (0.09) &  & (0.31) & (0.29) & (0.30)\\
                           & -0.66 & -0.65 & -0.65 & [-0.78, -0.51] & 0.92 & 0.86 & 0.88 \\
                           && (0.08) & (0.09) &  & (0.27) & (0.26) & (0.27)\\
                           & 0.28 & 0.28 & 0.28 & [0.21, 0.34] & 0.90 & 0.88 & 0.90 \\
                           && (0.04) & (0.04) &  & (0.13) & (0.13) & (0.14)\\
        $\mathbf{T = 500}$ & 1.40 & 1.39 & 1.38 & [1.22, 1.53] & 0.90 & 0.85 & 0.86 \\
                           && (0.10) & (0.10) &  & (0.32) & (0.30) & (0.31)\\
                           & 0.77 & 0.77 & 0.76 & [0.62, 0.90] & 0.90 & 0.86 & 0.88 \\
                           && (0.09) & (0.09) &  & (0.28) & (0.27) & (0.28)\\
                           & 0.00 & 0.00 & 0.00 & [-0.07, 0.07] & 0.89 & 0.87 & 0.90 \\
                           && (0.04) & (0.04) &  & (0.14) & (0.14) & (0.14)\\
                           & -0.60 & -0.60 & -0.59 & [-0.73, -0.46] & 0.89 & 0.85 & 0.89 \\
                           && (0.08) & (0.08) &  & (0.26) & (0.26) & (0.27)\\
                           & 1.00 & 0.99 & 0.98 & [0.89, 1.08] & 0.89 & 0.88 & 0.86 \\
                           && (0.06) & (0.06) &  & (0.18) & (0.18) & (0.18)\\
\hline
\hline
    \end{tabular}
    }
    \caption{\scriptsize Estimates of the true impact matrix, $B_{0}$, from the process: \\
    \textbf{ICA 1:} $\epsilon_{1,t} \sim NIG(0.34,0.57), \epsilon_{2,t} \sim NIG(0.51,1.00)$, and \\
     $\epsilon_{3,t} \sim \mathcal{N}(0,1)$ \\
     (see Section~\ref{sec:dist_spec}),
    with ${T} = 150, 500$, Monte Carlo simulations ${N_S} = 1000$ and residual-based moving-block bootstrap replications ${N} = 1999$.\\
    (a) True parameter values, $\mathrm{vec}(B_0)$.\\
    (b) NGML estimates, averaged across Monte Carlo simulations; the parentheses below contain the average estimated standard error.\\
    (c) Bootstrap estimates, averaged across Monte Carlo simulations; the parentheses below contain the average bootstrap standard error.\\
    (d) Median lower and upper limits of the bootstrap $90\%$ percentile confidence intervals, across Monte Carlo simulations.\\
    (e) Empirical probability coverage of the asymptotic $90\%$ confidence intervals.\\
    (f) Bootstrap (studentized) probability coverage of nominal $90\%$ confidence intervals.\\
    (g) Bootstrap (percentile) probability coverage of nominal $90\%$ confidence intervals.\\
    For (e), (f) and (g) the parentheses below contain the median length of the respective confidence interval, across Monte Carlo simulations.
      Coverage is reported after resolving the column labeling of $\hat{B}_T$: each replication is matched to $B_0$ by its best signed column permutation, measured by Amari distance (for $T=150$, 0.5\% of replications were realigned; for $T = 500$, 0\% replications were realigned). The realignment uses $B_0$ and is therefore a Monte Carlo convention rather than a procedure available to a practitioner.}
    \label{tab:EstandCovICA1}

\end{table}

%% file: Tables/EstandCovnoICA_v2.tex
\begin{table}[H]
\centering

    \renewcommand{\arraystretch}{0.9}

\scalebox{0.8}{

         \begin{tabular}{|c|c|cc|c|ccc|}
         \hline
         & & & & & & \multicolumn{2}{c|}{\textbf{Bootstrap Coverage}}\\

         & $\mathbf{B_{0}}$ & $\mathbf{\hat{B}_T}$ & $\mathbf{\hat{B}^*_T}$ & $\mathbf{90\%}$\textbf{CI} & \textbf{Asymp. Cov}. & \textbf{Studentized} & \textbf{Percentile} \\
         & (a) & (b) & (c) & (d) & (e) & (f) & (g) \\

         \hline
        \rule{0pt}{1ex}
                           & 1.29 & 1.26 & 1.22 & [0.95, 1.47] & 0.84 & 0.79 & 0.73 \\
                           && (0.19) & (0.19) &  & (0.63) & (0.61) & (0.51)\\
                           & 0.82 & 0.80 & 0.78 & [0.52, 1.03] & 0.89 & 0.95 & 0.83 \\
                           && (0.18) & (0.18) &  & (0.54) & (0.50) & (0.53)\\
                           & -0.66 & -0.64 & -0.63 & [-0.85, -0.38] & 0.89 & 0.84 & 0.83 \\
                           && (0.16) & (0.16) &  & (0.50) & (0.44) & (0.47)\\
                           & 0.20 & 0.17 & 0.08 & [-0.06, 0.25] & 1.00 & 0.95 & 0.86 \\
                           && (0.16) & (0.38) &  & (0.24) & (0.26) & (0.29)\\
        $\mathbf{T = 150}$ & 1.00 & 0.95 & 0.59 & [0.17, 1.01] & 0.47 & 0.47 & 0.67 \\
                           && (0.81) & (1.67) &  & (0.44) & (0.40) & (0.45)\\
                           & 0.55 & 0.35 & -0.06 & [-0.53, 0.38] & 0.26 & 0.32 & 0.42 \\
                           && (0.70) & (1.56) &  & (0.64) & (0.74) & (0.59)\\
                           & 0.00 & 0.02 & 0.05 & [-0.09, 0.23] & 0.79 & 0.58 & 0.76 \\
                           && (0.16) & (0.40) &  & (0.27) & (0.26) & (0.31)\\
                           & -0.60 & -0.39 & 0.05 & [-0.45, 0.53] & 0.37 & 0.32 & 0.45 \\
                           && (0.70) & (1.65) &  & (0.64) & (0.64) & (0.61)\\
                           & 1.00 & 0.95 & 0.57 & [0.15, 0.99] & 0.58 & 0.63 & 0.67 \\
                           && (0.78) & (1.60) &  & (0.37) & (0.37) & (0.40)\\
\hline
        \rule{0pt}{1ex}
                           & 1.29 & 1.28 & 1.27 & [1.11, 1.42] & 0.90 & 0.84 & 0.83 \\
                           && (0.11) & (0.11) &  & (0.39) & (0.36) & (0.32)\\
                           & 0.82 & 0.81 & 0.80 & [0.65, 0.93] & 0.84 & 0.77 & 0.82 \\
                           && (0.09) & (0.10) &  & (0.32) & (0.29) & (0.29)\\
                           & -0.66 & -0.65 & -0.65 & [-0.77, -0.52] & 0.87 & 0.77 & 0.84 \\
                           && (0.08) & (0.08) &  & (0.28) & (0.26) & (0.26)\\
                           & 0.20 & 0.18 & 0.08 & [0.03, 0.18] & 0.77 & 0.74 & 0.78 \\
                           && (0.09) & (0.17) &  & (0.15) & (0.15) & (0.15)\\
        $\mathbf{T = 500}$ & 1.00 & 0.98 & 0.58 & [0.15, 1.02] & 0.48 & 0.39 & 0.49 \\
                           && (0.52) & (0.96) &  & (0.30) & (0.29) & (0.30)\\
                           & 0.55 & 0.34 & -0.07 & [-0.58, 0.39] & 0.52 & 0.48 & 0.39 \\
                           && (0.32) & (0.67) &  & (0.68) & (0.69) & (0.54)\\
                           & 0.00 & 0.03 & 0.06 & [0.00, 0.18] & 0.65 & 0.68 & 0.53 \\
                           && (0.08) & (0.15) &  & (0.16) & (0.17) & (0.16)\\
                           & -0.60 & -0.38 & 0.09 & [-0.45, 0.60] & 0.58 & 0.48 & 0.39 \\
                           && (0.32) & (0.68) &  & (0.69) & (0.68) & (0.56)\\
                           & 1.00 & 0.97 & 0.56 & [0.15, 0.99] & 0.65 & 0.65 & 0.50 \\
                           && (0.51) & (0.95) &  & (0.34) & (0.30) & (0.30)\\
\hline
\hline
    \end{tabular}
    }
    \caption{\scriptsize Estimates of the true impact matrix, $B_{0}$, from the process: \\
    \textbf{No ICA:} $\epsilon_{1,t} \sim NIG(0.34,0.57), \epsilon_{2,t} \sim \mathcal{N}(0,1)$, and \\
     $\epsilon_{3,t} \sim \mathcal{N}(0,1)$ \\
     (see Section~\ref{sec:dist_spec}),
    with ${T} = 150, 500$, Monte Carlo simulations ${N_S} = 1000$ and residual-based moving-block bootstrap replications ${N} = 1999$.\\
    (a) True parameter values, $\mathrm{vec}(B_0)$.\\
    (b) NGML estimates, averaged across Monte Carlo simulations; the parentheses below contain the average estimated standard error.\\
    (c) Bootstrap estimates, averaged across Monte Carlo simulations; the parentheses below contain the average bootstrap standard error.\\
    (d) Median lower and upper limits of the bootstrap $90\%$ percentile confidence intervals, across Monte Carlo simulations.\\
    (e) Empirical probability coverage of the asymptotic $90\%$ confidence intervals.\\
    (f) Bootstrap (studentized) probability coverage of nominal $90\%$ confidence intervals.\\
    (g) Bootstrap (percentile) probability coverage of nominal $90\%$ confidence intervals.\\
    For (e), (f) and (g) the parentheses below contain the median length of the respective confidence interval, across Monte Carlo simulations.
    Coverage is reported after resolving the column labeling of $\hat{B}_T$: each replication is matched to $B_0$ by its best signed column permutation, measured by Amari distance (for $T=150$, 59.6\% of replications were realigned; for $T = 500$, 62.5\% of replications were realigned). The realignment uses $B_0$ and is therefore a Monte Carlo convention rather than a procedure available to a practitioner.}
    \label{tab:EstandCovnoICA}

\end{table}

%% file: Figures/fanchart_v2.tex
\begin{figure}[h]
    \centering
          \begin{subfigure}{0.45\textwidth}
          \centering
          \includegraphics[width = 0.75\textwidth]{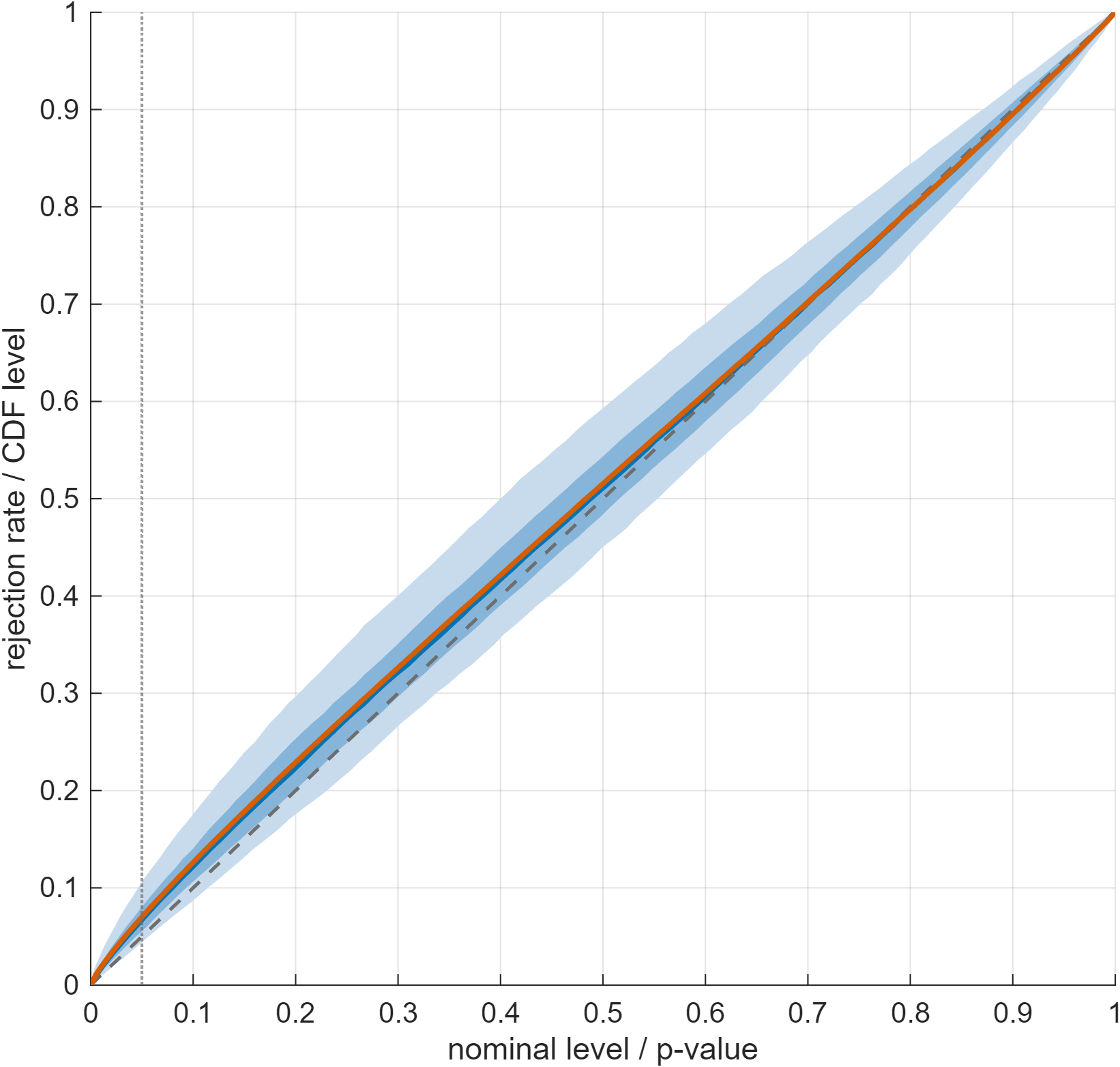}
          \subcaption{\textbf{ICA 0}}
          \label{fig:500ICA0}
          \end{subfigure}
      \quad
          \begin{subfigure}{0.45\textwidth}
          \centering
          \includegraphics[width = 0.75\textwidth]{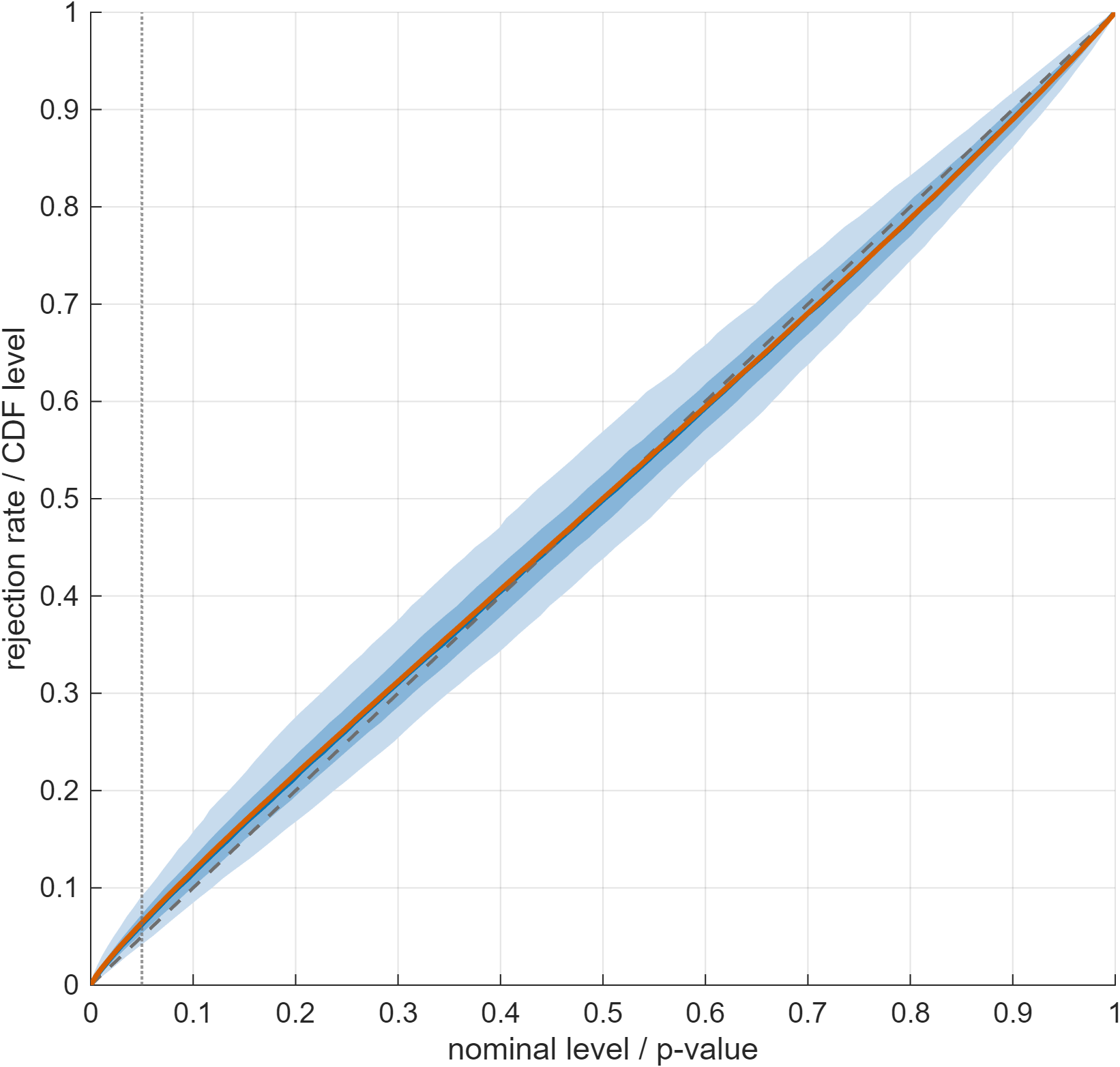}
          \subcaption{\textbf{ICA 1}}
          \label{fig:500ICA1}
          \end{subfigure}
      \quad
      % \vspace{0.2cm}
          \begin{subfigure}{0.45\textwidth}
          \centering
          \includegraphics[width = 0.75\textwidth]{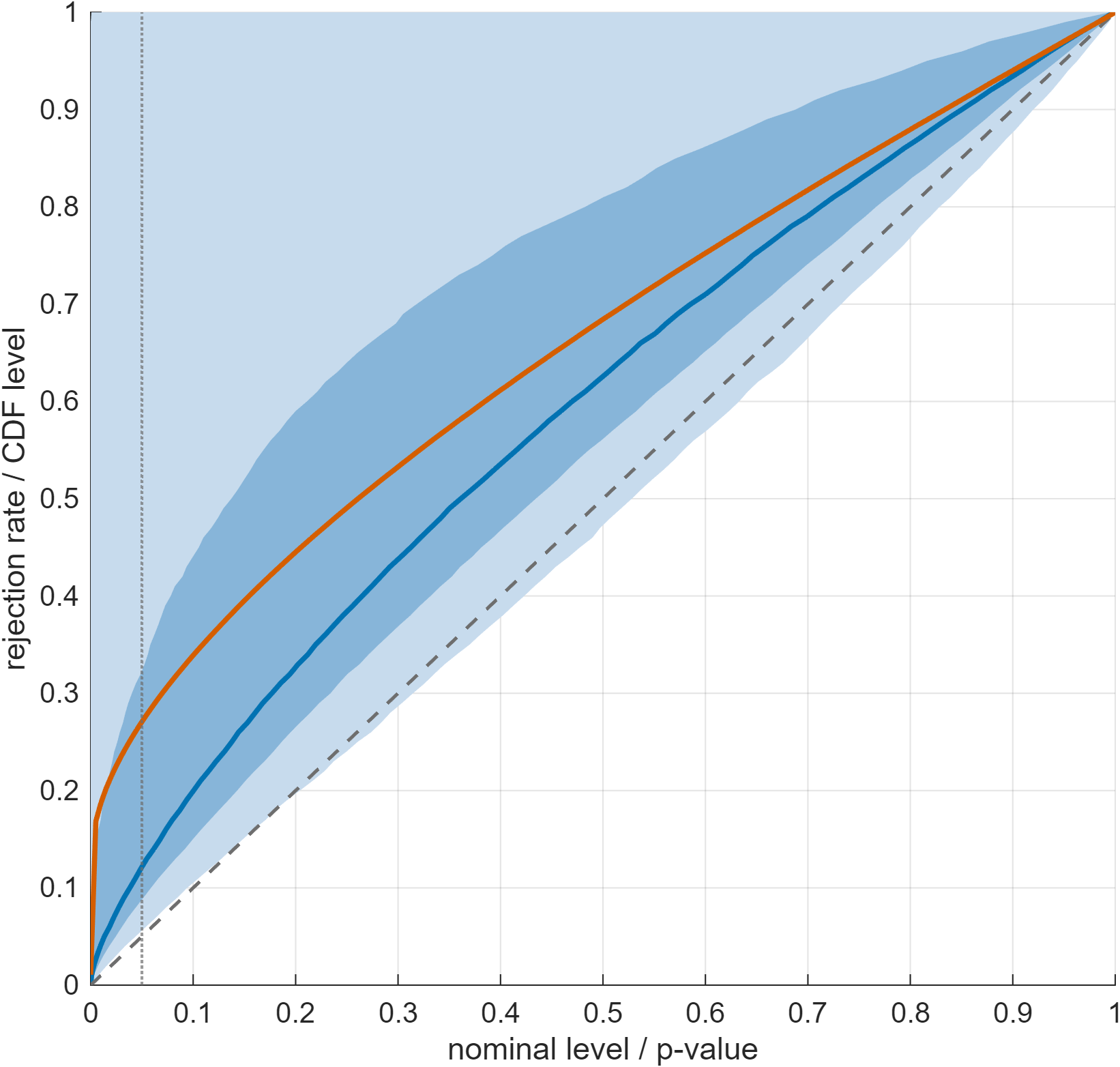}
          \subcaption{\textbf{No ICA}}
          \label{fig:500NoICA}
          \end{subfigure}
      \caption{Empirical distribution function of the $p$-values, $p^*_{M,T,s}$, $s = 1,2,\cdots, \mathbf{S}$ of the multivariate Doornik--Hansen normality test for the bootstrap estimates of $\hat{B}_T^{*}$ for the three specifications - \textbf{ICA 0}, \textbf{ICA 1} and \textbf{No ICA} - with sample size $T = 500$ and $M = \lceil (1/2)T^{3/5} \rceil$, for the number of sequences, $S = 1000$, across Monte Carlo simulations, $N_S = 1000$. The solid blue (orange) line is the median (mean) of the empirical rejection frequencies, the dashed gray line is the $45^{\circ}$ line, and the dark and light shaded areas indicate the $50$ and $90$-percentile intervals, respectively, across $N_S$. The dotted vertical line marks the nominal $5\%$ level, with the $p$-values on the $x$-axis and their corresponding empirical frequencies on the $y$-axis.}
      \label{fig:fanchart500}
\end{figure}

%% file: Figures/uni_NIGOGT500_v2.tex
\begin{figure}[h]
    \centering
    \includegraphics[width = 0.65\textwidth]{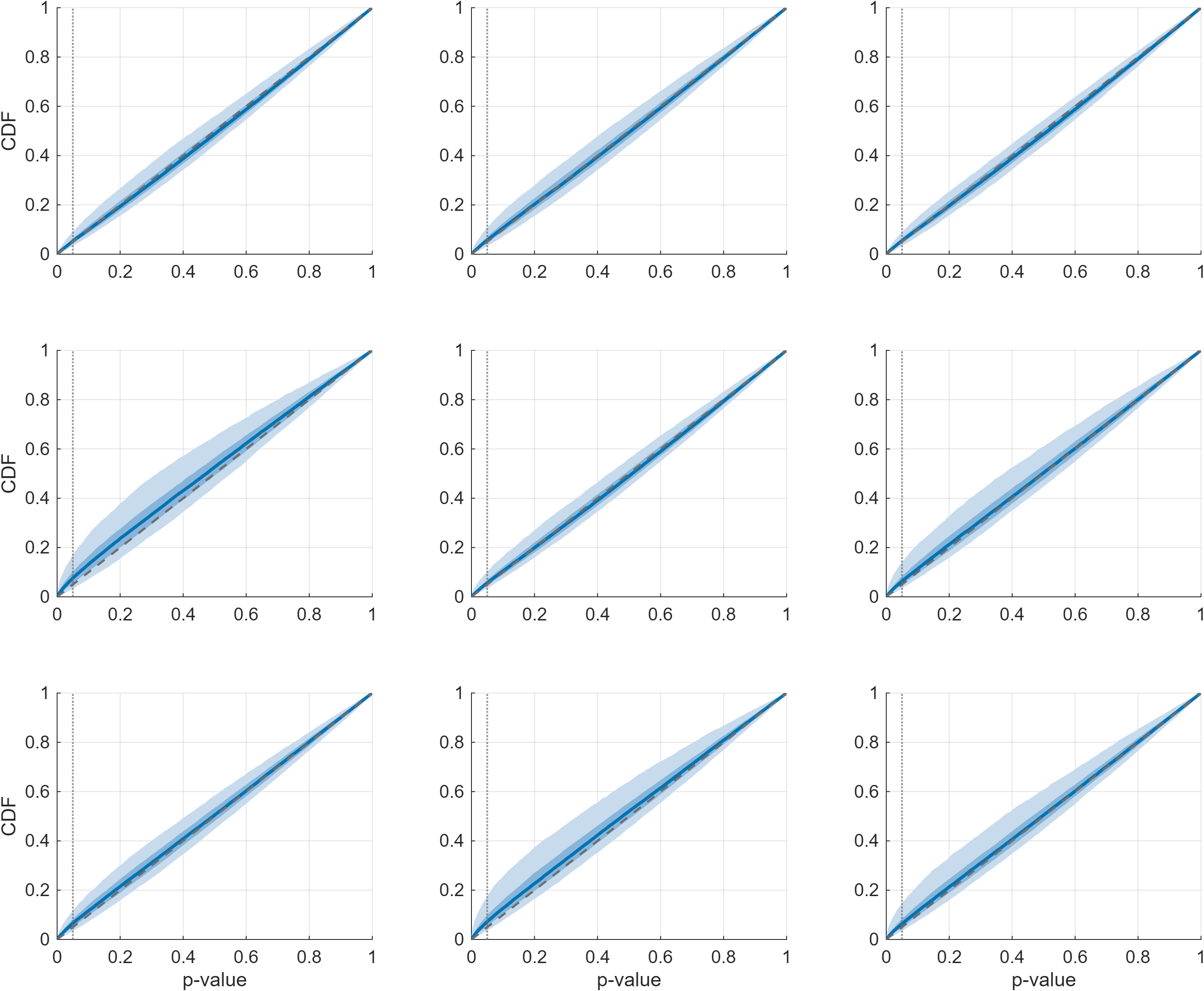}
    \caption{Univariate \textbf{ICA 1}: Empirical distribution function of $p^*_{M,T,s}$ from the univariate Doornik--Hansen test for normality of each element of the bootstrap estimates of
 $\hat{B}_{T}^*$ for the specification \textbf{ICA 1} (see Section~\ref{sec:dist_spec}), with sample size ${T} = 500$, the sequence length ${M} = \lceil (1/2)T^{3/5} \rceil$,
    and the number of sequences, ${S} = 1000$, across Monte Carlo simulations ${N_S} = 1000$. The $j$-th row of panels collects the three elements of the $j$-th column of $\hat{B}_T$, that is, the on-impact responses to the structural shock $\varepsilon_{j,t}$. The dashed gray line indicates the $45^{\circ}$ line and the dotted vertical line the nominal $5\%$ level. The solid blue line is the median of the empirical distribution of the $p$-values, $\pi^*_{M,T,S}(x)$, while the dark and light shaded areas indicate the $50$ and $90$-percentile intervals, respectively, with the $p$-values on the $x$-axis and their corresponding empirical frequencies on the $y$-axis.}
    \label{fig:uni_NIGOGT500}
\end{figure}

%% file: Figures/uni_NIGWGT500_v2.tex
\begin{figure}[H]
    \centering
    \includegraphics[width = 0.65\textwidth]{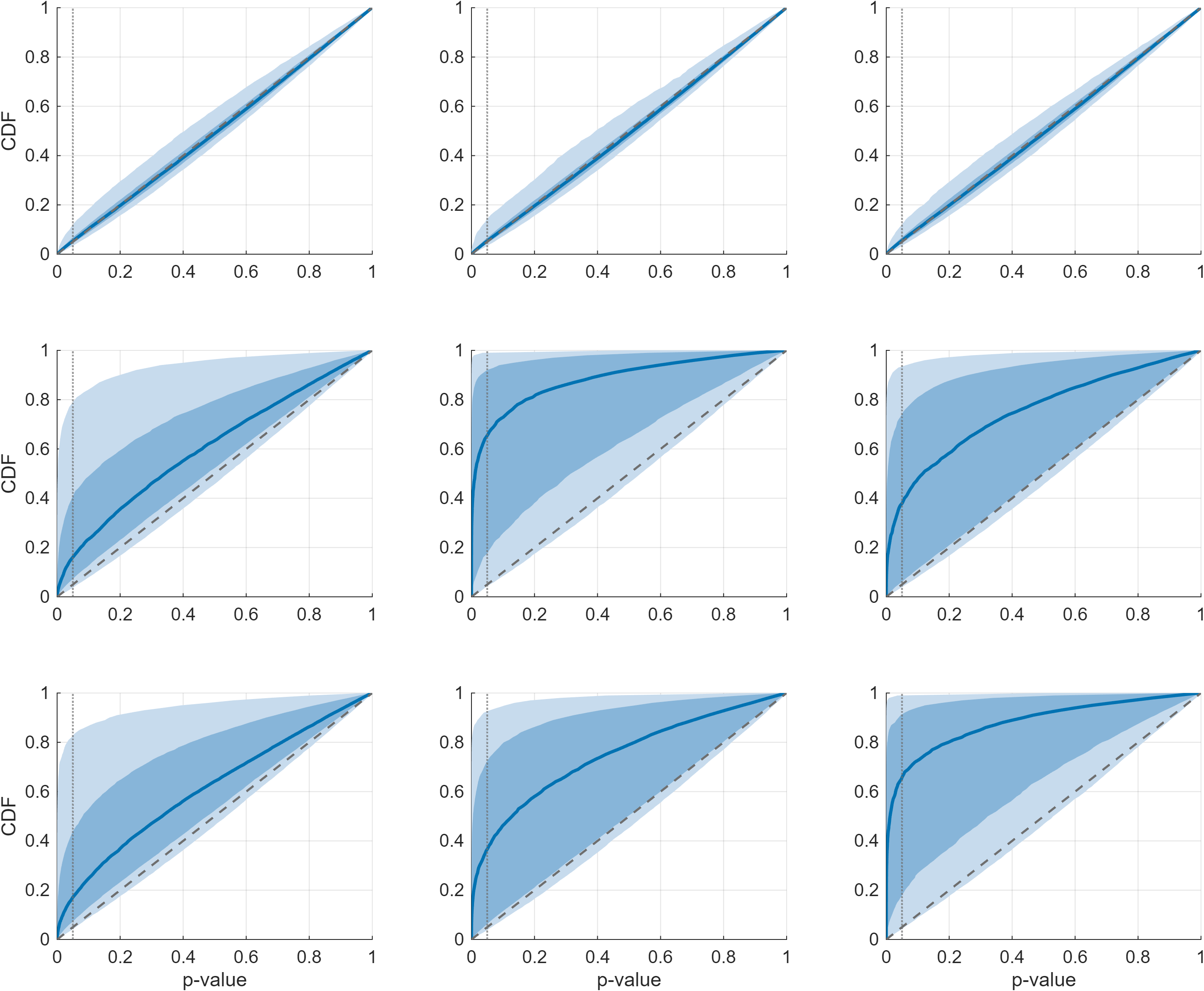}
    \caption{Univariate \textbf{No ICA}: Empirical distribution function of $p^*_{M,T,s}$ from the univariate Doornik--Hansen test for normality of each element of the bootstrap estimates of
 $\hat{B}_{T}^*$ for the specification \textbf{No ICA} (see Section~\ref{sec:dist_spec}), with sample size ${T} = 500$, the sequence length ${M} = \lceil (1/2)T^{3/5} \rceil$,
    and the number of sequences, ${S} = 1000$, across Monte Carlo simulations ${N_S} = 1000$. The $j$-th row of panels collects the three elements of the $j$-th column of $\hat{B}_T$, that is, the on-impact responses to the structural shock $\varepsilon_{j,t}$. The dashed gray line indicates the $45^{\circ}$ line and the dotted vertical line the nominal $5\%$ level. The solid blue line is the median of the empirical distribution of the $p$-values, $\pi^*_{M,T,S}(x)$, while the dark and light shaded areas indicate the $50$ and $90$-percentile intervals, respectively, with the $p$-values on the $x$-axis and their corresponding empirical frequencies on the $y$-axis.}
    \label{fig:uni_NIGWGT500}
\end{figure}

%% file: Tables/weakiduni_v2.tex
\begin{table}[H]
\centering
\renewcommand{\arraystretch}{1.1}
\scalebox{0.85}{
    \begin{tabular}{|c|ccc|ccc|ccc|ccc|}
     \hline
     & & & & & & & & & & \multicolumn{3}{c|}{\textbf{by shock}}\\
        % & \multicolumn{3}{c|}{$\varepsilon_{1,t}$} & \multicolumn{3}{c|}{$\varepsilon_{2,t}$} & \multicolumn{3}{c|}{$\varepsilon_{3,t}$} & \multicolumn{3}{c|}{\textbf{by shock}} \\
        \cline{11-13}
        & \multicolumn{3}{c|}{$\varepsilon_{1,t}$} & \multicolumn{3}{c|}{$\varepsilon_{2,t}$} & \multicolumn{3}{c|}{$\varepsilon_{3,t}$} & & & \\
        & & & & & & & & & & $\varepsilon_{1,t}$ & $\varepsilon_{2,t}$ & $\varepsilon_{3,t}$ \\
        \cline{1-10}
        \rule{0pt}{2.5ex}$\mathbf{B_0}$ & 1.29 & 0.82 & -0.66 & 0.28 & 1.40 & 0.77 & 0.00 & -0.60 & 1.00 & & & \\
        \hline
        \rule{0pt}{2.5ex}$\mathbf{T = 300}$ & 0.056 & 0.075 & 0.070 & 0.235 & 0.484 & 0.403 & 0.209 & 0.386 & 0.458 & 0.067 & 0.374 & 0.351 \\
        $\mathbf{T = 500}$ & 0.057 & 0.060 & 0.058 & 0.325 & 0.578 & 0.479 & 0.286 & 0.456 & 0.530 & 0.058 & 0.461 & 0.424 \\
        \hline
    \end{tabular}
}
\caption{\scriptsize Rejection frequencies of the univariate diagnostic with the Doornik--Hansen test under the weak-identification sequence. The frequencies are reported for each element, and also averaged over the three elements of each column of $\hat{B}_T$. Each entry is the average across Monte Carlo simulations of the within-simulation rejection frequency at the nominal $5\%$ level, over $S = 1000$  sequences. The comparable figures under \textbf{ICA 1} and \textbf{No ICA} are in Table~\ref{tab:uninormtest}.}
\label{tab:weakiduni}
\end{table}

%% file: Tables/wkidp_retention_v2.tex
\begin{table}[H]
\centering
\renewcommand{\arraystretch}{1}
\scalebox{0.90}{
    \begin{tabular}{|c|cc|cc|}
     \hline
    & \multicolumn{2}{c|}{Residual normality pre-test} & \multicolumn{2}{c|}{Bootstrap diagnostic} \\
        \textbf{DGP} & $T=300$ & $T=500$ & $T=300$ & $T=500$ \\
        \hline
        \rule{0pt}{2.5ex}\textbf{ICA 1} & 0.399 & 0.204 & 0.121 & 0.089 \\
        \rule{0pt}{2.5ex}\textbf{Weak ICA 1} & 0.727 & 0.603 & 0.595 & 0.763 \\
        \rule{0pt}{2.5ex}\textbf{No ICA} & 0.615 & 0.461 & 0.641 & 0.779 \\
        \hline
    \end{tabular}
    }
\caption{\scriptsize Rejection frequencies of the residual normality pre-test and the bootstrap diagnostic, at nominal significance level $5\%$. The residual normality pre-test of Section~\ref{subsec:insuff_test} retains a simulation when the univariate Jarque--Bera test rejects normality for every reduced-form residual series; the table reports the share of simulations it does not retain. The bootstrap diagnostic is the decision rule of Section~\ref{subsubsec:wkidp_rule}, applied to the $n^2$ elements of $\mathrm{vec}(\hat{B}_T)$ with Bonferroni corrections and $M = \lceil (1/3)T^{3/5}\rceil$, for the null hypothesis that the normal approximation of the bootstrap distribution of $\hat{B}_T$ is valid.}
\label{tab:wkidp_size_power}
\end{table}

%% file: 06_EMP.tex
In this section, we consider an empirical application which focuses on the role of macroeconomic and financial sources of uncertainty in business cycle fluctuations. The question of uncertainty as an exogenous driver or endogenous response of business cycle fluctuations has been addressed
with different identification schemes by
\citet{angelini_uncertainty_2019,ludvigson_uncertainty_2021}, so an
approach that imposes no economic identifying restrictions is also of independent interest. The impulse responses support the view that financial
uncertainty is a source of fluctuations in real activity, whereas
macroeconomic uncertainty is largely an endogenous response to them. The diagnostic does not
reject the validity of the identification conditions at the $5\%$ level,
which supports inference based on the bootstrap percentile bands. 

\subsection{Data and specification} We consider the trivariate SVAR of
$X_t = (U_{Ft}, Y_t, U_{Mt})'$, with the measures of financial uncertainty
($U_{F}$), real activity ($Y$) and macroeconomic uncertainty ($U_{M}$), as in \citet{ludvigson_uncertainty_2021} and
\citet{angelini_uncertainty_2019}. On the one hand, the former impose inequality restrictions on the
behavior of the shocks during influential episodes, such as October 1987
and the 2008 financial crisis, together with restrictions on the
correlation of the structural shocks with asset market returns and the price of gold.
On the other hand, the latter combine non-recursive zero restrictions
with breaks in unconditional volatility. In this application, we rely on statistical identification through non-Gaussianity without imposing economic identifying restrictions.
The data are of monthly frequency, from 1960:07 to 2025:06, i.e., $T = 780$. For a detailed exposition of the construction of uncertainty measures, we refer the reader to \citet{jurado_measuring_2015} and \citet{ludvigson_uncertainty_2021}. For the measure of real activity, we consider the log of real industrial production (FRED database). The
reduced-form specification:
\begin{equation}
    X_{t} = \Pi_0 + \Pi_{\tau} t + \Pi_{1}X_{t-1} + \cdots + \Pi_{p}X_{t-p} + u_{t},
    \qquad u_t = B\varepsilon_t ,
    \label{eq:emp_VAR}
\end{equation}
with a constant and a linear trend. In Appendix~\ref{app:other_results}, Table~\ref{tab:emp_lagsel} reports the information criteria which select lag order $p = 3$. Also, the multivariate Ljung--Box test for serial correlation (up to 12 lags) in the estimated reduced-form residuals does not show evidence against the null of no serial correlation. As in Section~\ref{sec:MC}, $(\beta_0,\sigma_0)$ and the shape parameters
are estimated by NGML with NIG structural shocks,
and the impulse responses to a one-standard-deviation shock are
\begin{equation}
\underset{n \times n}{\hat{\Psi}_i} = (R\,\hat{\mathbf{\Pi}}^i\,R')\,\hat{B}_T,
\qquad i = 0,1,\dots,h,
\label{eq:emp_IRF}
\end{equation}
where $R = (I_3, \mathbf{0}_{3\times 6})$ selects the first block of the companion form and
$\hat{\Psi}_0 = \hat{B}_T$. All bootstrap quantities are computed from the residual-based
moving block bootstrap of Section~\ref{sec:BSINF}, with $N = 19{,}999$
replications.

\subsection{Estimates}\label{subsec:emp_pretest}
Table~\ref{tab:emp_B} reports the impact matrix $\hat{B}_T$ with bootstrap standard errors, and the estimated shocks are plotted in
Figure~\ref{fig:emp_shocks}. The shocks are strongly leptokurtic, with
sample excess kurtosis of $17.44$, $10.65$ and $26.48$. Although a substantial portion of this excess kurtosis can be attributed to the COVID-19 pandemic and the October 1987 Black Monday crisis, the estimated shocks continue to exhibit substantial excess kurtosis even after excluding these episodes.
\input{Figures/emp_shocks.tex}
\input{Tables/emp_Best.tex}
We allow for asymmetric shocks, while the NIG densities above fix the asymmetry
parameters at zero. Table~\ref{tab:emp_skew} in Appendix~\ref{app:other_results}
shows that the estimates are robust to this choice.

% We adopt two labeling conventions. First, the column order induced by the scheme of Section~\ref{subsec:NGMLest}
% does not carry economic interpretation. We therefore order the columns of $\hat B_T$ so that the element of largest
% absolute value in column $j$ lies in row $j$, and sign-normalize the matrix such that the diagonal is positive. Since the three column maxima of $\hat B_T$ lie
% in different rows, this ordering is unique, so the likelihood and the asymptotic
% results are unaffected. Second, each bootstrap
% replication is realigned to this labeling of $\hat{B}_T$ before the diagnostic is
% computed. This is to avoid spurious multimodality due to a mislabeled replication in which case the diagnostic would reject because of wrong labeling rather than identification
% failure. Moreover, only $1.1\%$ of
% the bootstrap replications require realignment.

\subsection{The bootstrap diagnostic}\label{subsec:emp_diag}
We apply the decision rule of Section~\ref{subsubsec:wkidp_rule} to the
$N$ stored replications of $\hat{B}_T$, with $\alpha = 0.05$ and
$n^2 = 9$. For each element $(i,j)$, the $L$ disjoint sequences of $M$
replications give $\pi^*_{M,T,L,ij}(x)$ in \eqref{eq:wkidp_picurve} and $D_{ij}$
in \eqref{eq:wkidp_ks}. Figure~\ref{fig:emp_picurve} plots the empirical distribution of the $p$-values at $M = 12$ for each element of the impact matrix, and Table~\ref{tab:emp_diag} reports the decision rule
for various values of $M$.

\input{Tables/emp_diag_v2.tex}

\input{Figures/emp_picurve.tex}
Table~\ref{tab:emp_diag} reports the results. At $M = 12$ and $L = 32$, the largest element-wise statistic, $\max_{i,j} D_{ij}$, is $1.450$ against a critical value of $1.683$. The smallest element-wise $p$-value is $0.024$ against the Bonferroni level $0.0056$, and no
element of $\hat{B}_T$ rejects. The conclusion holds at every $M$ considered between $10$ and $28$
at $\ell = 27$. This range includes $M = 19$ and $M = 28$, the values that
the rules $M = \lceil a\,T^{3/5} \rceil$ with $a = 1/3$ and $a = 1/2$ of
Section~\ref{sec:MC} adopt. Since the diagnostic is asymptotically ancillary
(Remark~\ref{rem:pivotal}), conditioning on non-rejection does not asymptotically distort
the bootstrap percentile bands reported below.

\subsection{Impulse responses}\label{subsec:emp_irf}
Figure~\ref{fig:emp_irf} reports the impulse responses with $68\%$ and $90\%$
bootstrap percentile bands. Although NGML achieves statistical identification, the identified
shocks lack direct economic interpretation. Consequently, we use sign restrictions to label the
identified shocks. In line with the literature where macroeconomic and financial uncertainty co-move positively \citep{jurado_measuring_2015,carriero_measuring_2018}, we assume that a positive
shock to either uncertainty index elicits a positive response in both indices. To economically
distinguish the two uncertainty shocks, we impose that macroeconomic uncertainty is countercyclical, i.e., its response to a positive real activity shock is negative \citep{bloom_impact_2009}, while
leaving the response of financial uncertainty unrestricted. These restrictions are satisfied by labeling $\hat{\varepsilon}_{1,t}$,
$\hat{\varepsilon}_{2,t}$ and $\hat{\varepsilon}_{3,t}$ as the financial
uncertainty, real activity and macroeconomic uncertainty shocks,
respectively.

A financial uncertainty shock raises financial
uncertainty sharply. The response peaks within the first quarter
and decays over the following two years. Real activity does not respond
significantly on impact but declines steadily after about one quarter,
while macroeconomic uncertainty rises over the first six months. After a real activity shock, macroeconomic
uncertainty falls sharply, while financial uncertainty does not respond
significantly at any horizon. A
macroeconomic uncertainty shock raises macroeconomic uncertainty to a peak
within two quarters but produces no sustained response of real activity
or financial uncertainty. 

Financial uncertainty depresses real activity with a delay and remains largely
unaffected by the other two shocks. Macroeconomic uncertainty leads to no
sustained effects in either variable, but responds persistently to both other shocks.
This evidence matches the suggestion of \citet{ludvigson_uncertainty_2021}
that financial uncertainty is a likely source of business cycle
fluctuations, whereas higher macroeconomic uncertainty is
largely an endogenous response to them.
\input{Figures/emp_irf.tex}

%% file: Figures/emp_shocks.tex
\begin{figure}[H]
    \centering
    \includegraphics[width=0.70\textwidth]{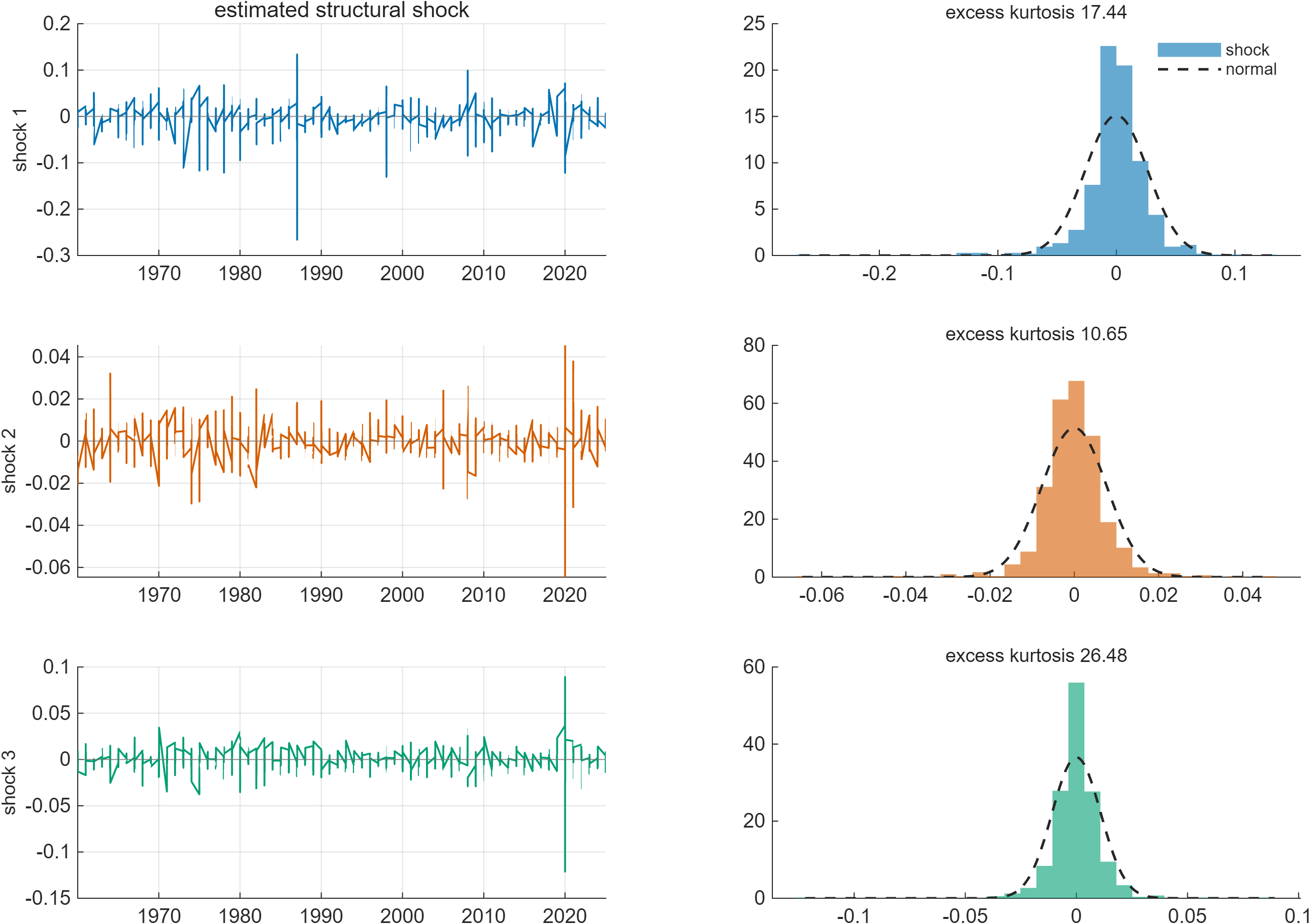}
    \caption{\scriptsize Estimated structural shocks $\hat{\varepsilon}_t = \hat{B}_T^{-1}\hat{u}_t$, 1960:10--2025:06. Left column: estimated structural shocks across time. Right column: their frequency distribution against a Gaussian density with the same mean and variance, with the sample excess kurtosis reported above each panel.}
    \label{fig:emp_shocks}
\end{figure}

%% file: Tables/emp_Best.tex
\begin{table}[H]
\centering
\renewcommand{\arraystretch}{1.05}

\scalebox{0.86}{
\begin{tabular}{|c|ccc|}
\hline
 & $\hat{\varepsilon}_{1,t}$ & $\hat{\varepsilon}_{2,t}$ & $\hat{\varepsilon}_{3,t}$ \\
\hline
\multicolumn{4}{|c|}{\rule{0pt}{2.6ex}\textit{Estimated impact matrix $\hat{B}_T$ (estimates multiplied by $100$)}} \\
\hline
\rule{0pt}{2.6ex}
$U_{F}$ & $2.559$ & $0.121$ & $0.176$ \\
        % & $(0.149)$ & $(0.090)$ & $(0.097)$ \\
 & $[0.243]$ & $[0.179]$ & $[0.361]$ \\
 & $(2.109,\,2.909)$ & $(-0.223,\,0.309)$ & $(-0.082,\,1.161)$ \\[0.3ex]
$Y$ & $-0.044$ & $0.743$ & $0.270$ \\
        % & $(0.033)$ & $(0.053)$ & $(0.113)$ \\
 & $[0.076]$ & $[0.097]$ & $[0.218]$ \\
 & $(-0.168,\,0.084)$ & $(0.550,\,0.850)$ & $(-0.155,\,0.605)$ \\[0.3ex]
$U_{M}$ & $0.257$ & $-0.566$ & $1.024$ \\
        % & $(0.047)$ & $(0.156)$ & $(0.100)$ \\
 & $[0.119]$ & $[0.285]$ & $[0.159]$ \\
 & $(-0.018,\,0.367)$ & $(-0.902,\,0.068)$ & $(0.741,\,1.242)$ \\

\hline
\multicolumn{4}{|c|}{\rule{0pt}{2.6ex}\textit{Estimated structural shocks}} \\
\hline
\rule{0pt}{2.4ex}
$\hat{\kappa}$, full sample & $17.44$ & $26.48$ & $10.65$ \\
$\hat{\kappa}$, excl.\ 2020 & $19.27$ & $1.86$ & $3.76$ \\
$\hat{\kappa}$, excl.\ 2020 and 1987:10--11 & $5.53$ & $1.85$ & $3.81$ \\
\hline
\end{tabular}
}

\vspace{0.2cm}

\caption{\scriptsize The upper panel reports the estimated impact matrix $\hat{B}_T$ of the trivariate SVAR $X_t=(U_{Ft},Y_t,U_{Mt})'$ identified by NGML with NIG-distributed structural shocks, estimated on monthly data over 1960:07--2025:06 with $p=3$ lags, a constant and a linear trend. All entries in the upper panel are multiplied by $100$. The bootstrap standard errors are reported in square brackets, and the $90\%$ bootstrap percentile interval in parentheses, computed from $N=19{,}999$ replications with block length $\ell=27$. The lower panel reports the sample excess kurtosis $\hat{\kappa}$ of the estimated structural shocks $\hat{\varepsilon}_t=\hat{B}_T^{-1}\hat{u}_t$, together with $\hat{\kappa}$ recomputed after removing the twelve months of 2020 and, in the last row, also the two months 1987:10 and 1987:11.}
\label{tab:emp_B}
\end{table}

%% file: Tables/emp_diag_v2.tex
\begin{table}[H]
\centering
\renewcommand{\arraystretch}{1.05}

\scalebox{0.78}{
\begin{tabular}{rccc}
\hline
$M$ & $\max_{i,j} D_{ij}$ & $c_{1-\alpha/n^2}$ & $\phi_{M,L}$ \\
\hline
10 & 1.355 & 1.638 & 0 \\
12 & 1.450 & 1.683 & 0 \\
15 & 1.325 & 1.700 & 0 \\
19 & 1.431 & 1.632 & 0 \\
20 & 1.612 & 1.620 & 0 \\
28 & 1.537 & 1.649 & 0 \\
\hline
\end{tabular}}
\caption{\scriptsize The diagnostic decision rule \eqref{eq:wkidp_report} for a range of $M$, with $L = \lfloor T/(2M) \rfloor$ disjoint
sequences, on the $N = 19{,}999$ moving-block bootstrap replications of
$\hat{B}_T$ with $\ell = 27$. $D_{ij}$ is the statistic
\eqref{eq:wkidp_ks} for element $(i,j)$ of $\hat{B}_T$,
$c_{1-\alpha/n^2}(M,L)$ its critical value at the Bonferroni level
$\alpha/n^2 = 0.0056$, simulated under the null of a valid bootstrap, and
$\phi_{M,L} = 0$ when no element rejects.}
\label{tab:emp_diag}
\end{table}

% \begin{table}[H]
% \centering
% \renewcommand{\arraystretch}{1.05}

% \scalebox{0.78}{
% \begin{tabular}{rrcccc}
% \hline
% $M$ & $L$ & $\max_{i,j} D_{ij}$ & $c_{1-\alpha/n^2}(M,L)$ & $p_{\min}$ & $\phi_{M,L}$ \\
% \hline
% 10 & 38 & 1.355 & 1.638 & 0.042 & 0 \\
% 12 & 32 & 1.450 & 1.683 & 0.024 & 0 \\
% 15 & 25 & 1.325 & 1.700 & 0.049 & 0 \\
% 19 & 20 & 1.431 & 1.632 & 0.026 & 0 \\
% 20 & 19 & 1.612 & 1.620 & 0.006 & 0 \\
% 28 & 13 & 1.537 & 1.649 & 0.010 & 0 \\
% \hline
% \end{tabular}}
% \caption{\scriptsize The decision rule \eqref{eq:wkidp_report} for a range of
% sequence lengths $M$, with $L = \lfloor T/(2M) \rfloor$ disjoint
% sequences, on the $N = 19{,}999$ moving-block bootstrap replications of
% $\hat{B}_T$ with $\ell = 27$. $D_{ij}$ is the statistic
% \eqref{eq:wkidp_ks} for element $(i,j)$ of $\hat{B}_T$,
% $c_{1-\alpha/n^2}(M,L)$ its critical value at the Bonferroni level
% $\alpha/n^2 = 0.0056$, simulated under the null of a valid bootstrap,
% $p_{\min}$ the smallest of the nine element $p$-values, and
% $\phi_{M,L} = 0$ when no element rejects.}
% \label{tab:emp_diag}
% \end{table}

%% file: Figures/emp_picurve.tex
% The figure is written by run_empcfg as generated/figures/<tag>_picurve.pdf
% for the HEADLINE cell, which is the Kolmogorov-Smirnov statistic at M = 12.
% The curve data are in emp27_pi_curve.csv.  See README.md, section 3.
\begin{figure}[H]
    \centering
    \includegraphics[width=0.65\textwidth]{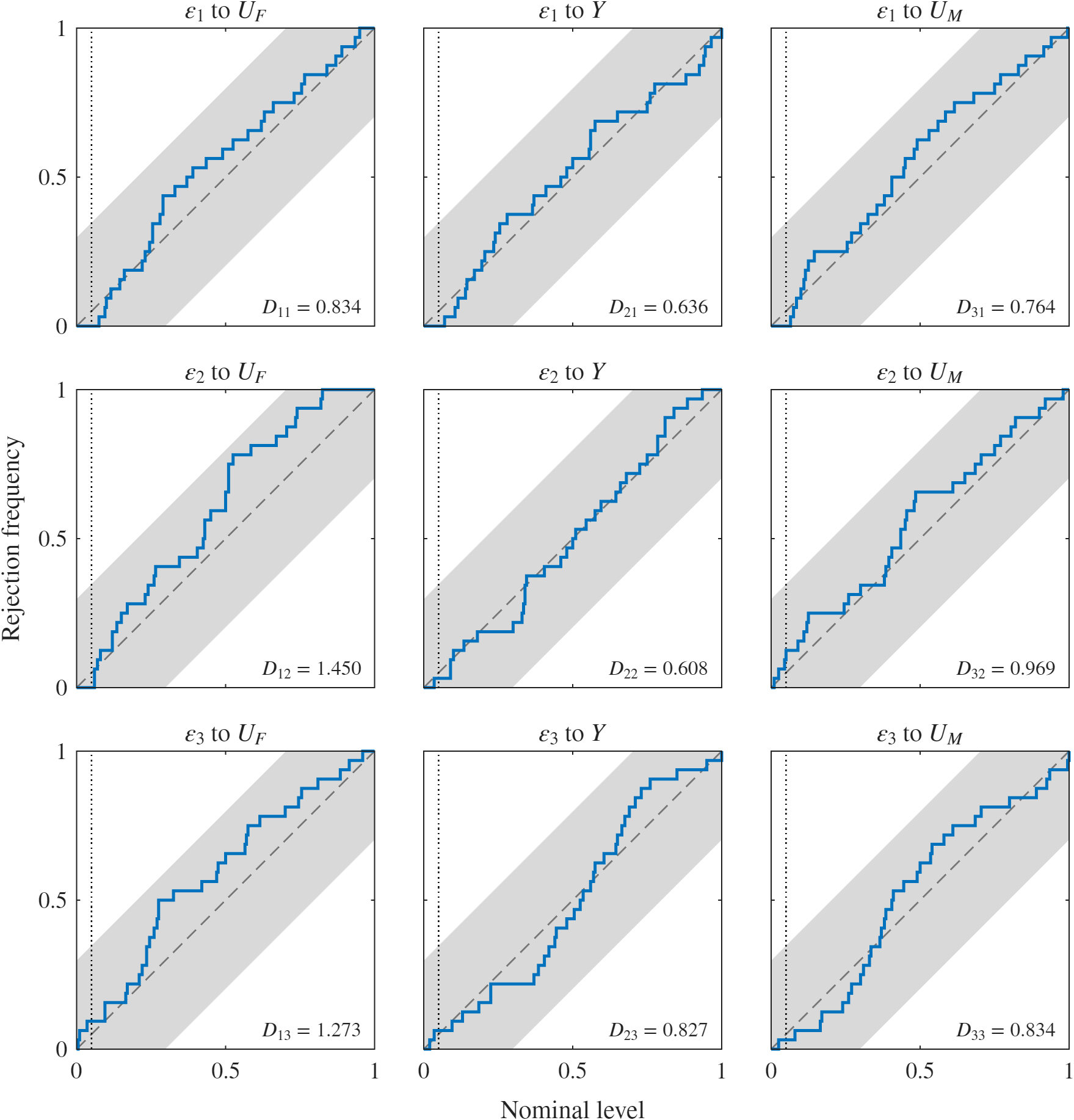}
    \caption{\scriptsize Empirical distribution functions $\pi^*_{M,T,L,ij}(x)$ of the
$p$-values of the Kolmogorov--Smirnov sequence statistic for the elements of $\hat{B}_T$, with $M = 12$,
$L = 32$ disjoint sequences and $\ell = 27$. As in
Figures~\ref{fig:uni_NIGOGT500} and \ref{fig:uni_NIGWGT500}, row $j$
collects the on-impact responses to $\hat{\varepsilon}_{j,t}$. The solid
line is $\pi^*_{M,T,L,ij}(x)$ and the dashed line the $45^{\circ}$ line.
The shaded area is
$|\pi - x| \leq c_{1-\alpha/n^2}(M,L)/\sqrt{L}$ with
$c_{1-\alpha/n^2}(M,L) = 1.683$ and $\alpha = 0.05$, so a curve leaves
it exactly when $D_{ij}$ in \eqref{eq:wkidp_ks} exceeds the critical
value. The dotted vertical line marks the $5\%$ level.}
    \label{fig:emp_picurve}
\end{figure}

%% file: Figures/emp_irf.tex
\begin{figure}[H]
    \centering
    \includegraphics[width=0.65\textwidth]{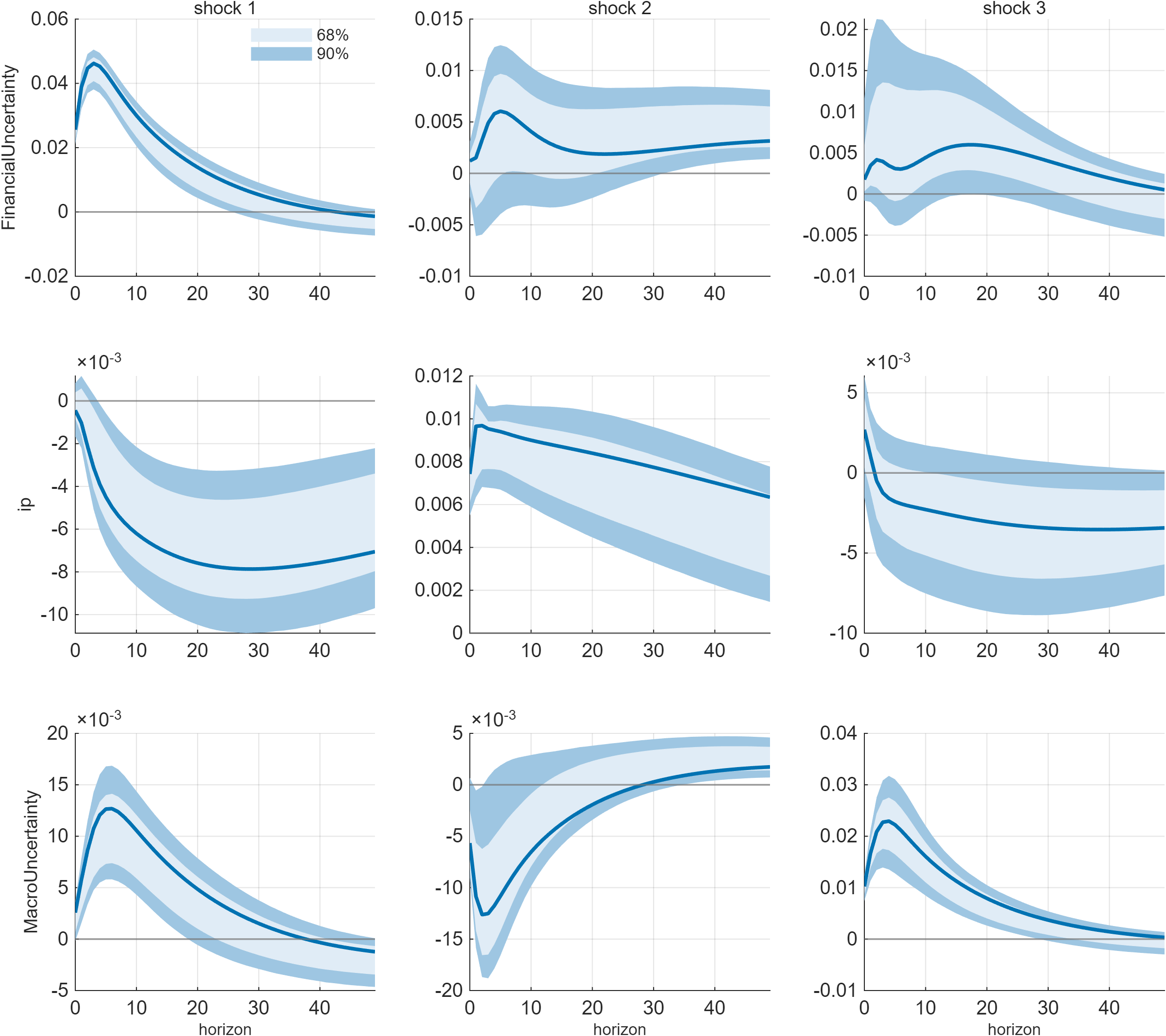}
    \caption{\scriptsize Estimated impulse responses of $X_t = (U_{Ft},Y_t,U_{Mt})'$ to one-standard-deviation positive structural shocks, with $68\%$ and $90\%$ bootstrap percentile bands from $N = 19{,}999$ replications and block length $\ell = 27$. Impulse responses with $\ell = 50$ are reported in
Appendix~\ref{app:other_results} as a robustness check.
Rows denote the variables and columns are the estimated shocks. The horizontal axis is the horizon in months.}
    \label{fig:emp_irf}
\end{figure}

%% file: 11_con.tex
This paper develops a bootstrap-based approach to assess the identification conditions
of non-Gaussian SVARs. Since the reduced-form
innovations are linear mixtures of the structural shocks, their non-Gaussianity is uninformative about the identification failure due to excess Gaussian shocks. We
instead propose a bootstrap-based approach, where under valid identification and maintained regularity conditions, the
conditional bootstrap distribution of the NGML impact-matrix estimator is
asymptotically normal, and it departs from normality when identification
fails. Hence, assessing identification reduces to testing
the normality of the bootstrap replications of the impact matrix.

We provide three theoretical results which underpin the diagnostic. First, the diagnostic is valid
across the entire null. Even at the single-Gaussian boundary of the null, where the full-parameter information matrix is singular, the impact-matrix
estimator retains a nonsingular profile information and remains consistent
and asymptotically normal. Second, by letting the
number of bootstrap replications and the sample size diverge jointly at an appropriate rate rather than
sequentially, the test statistic is asymptotically pivotal and ancillary. Therefore, conditioning on the diagnostic does not induce any pre-testing bias asymptotically. Third, a univariate version of the diagnostic
localizes which shocks are responsible for an identification
failure, extending its use to partially identified non-Gaussian SVARs.

Monte Carlo evidence, based on NIG shocks, shows
that the diagnostic mildly over-rejects in finite samples under valid
identification. Its power against an additional Gaussian shock
increases with the sample size. Normality
tests on the reduced-form residuals, in contrast, do not discriminate between valid and
invalid specifications. The multivariate test rejects in essentially every
sample, and the conventional conditioning rule retains a share of
misspecified samples that grows with the sample size. Under weak identification,
where the non-Gaussianity of one shock vanishes with the sample size, the univariate
diagnostic rejects mainly for the columns associated with the near-Gaussian and the
Gaussian shocks, and close to the nominal level for the column of the strongly
non-Gaussian shock.

In the empirical application with measures of financial uncertainty,
real activity and macroeconomic uncertainty, the estimated shocks display substantial excess kurtosis. The diagnostic does not reject
the validity of the identification conditions at the nominal $5\%$ level, and this
conclusion holds across the range of bootstrap replications considered for the normality test. The
resulting impulse responses are consistent with financial uncertainty being a source of
business cycle fluctuations and macroeconomic uncertainty being largely an
endogenous response to them.

Several extensions remain open. The finite-sample performance of the
diagnostic depends on the joint choice of the sequence length and the
number of sequences, and a data-driven choice that balances size and power
can be explored. The approach is not specific to the likelihood-based,
fully independent shocks considered here. Analogous bootstrap diagnostics
could be developed for moment-based estimators that relax complete
independence, to assess their relevant identifying conditions.

%% file: 12_app.tex
\section{Preliminaries} \label{appendix:a} 
We denote the probability measure for the data as $\mathbb{P}$, with $\mathbb{E}(.)$ and $\operatorname{Var}(.)$ denoting the expectation and variance operators with respect to $\mathbb{P}$. Let $\mathbb{E}^*(.)$ and $\operatorname{Var}^*(.)$ denote the expectation and variance operators with respect to the bootstrap probability measure $\mathbb{P}^*$, conditional on the observed data. For $\epsilon >0$, define $\delta^*(\epsilon) \coloneqq \mathbb{P}^*(\|\hat{\lambda}_T^* - \hat{\lambda}_T \| > \epsilon)$, where $\hat{\lambda}_T^*$ is the bootstrap analog of the NGML estimator $\hat{\lambda}_T$, and $\|.\|$ is the Euclidean norm. We denote $\hat{\lambda}_T^* - \hat{\lambda}_T \xrightarrow{p^*}_p 0$, as $\hat{\lambda}_T^* - \hat{\lambda}_T$ converges in $p^*$-probability to $0$, in probability, if for every $\epsilon >0$, $\delta^*(\epsilon) \xrightarrow{p} 0$, as $T \rightarrow \infty$. Similarly, consider a random variable $X_T$, with CDF denoted as $G_{X_T}(x) \coloneqq P(X_T \leq x)$, and let $X_T^*$ be its bootstrap analog, with CDF denoted as $G_{X_T}^*(x) \coloneqq \mathbb{P}^*(X_T^* \leq x)$, conditional on the observed data. We say that $X_T^* \xrightarrow{d^*}_p X_T$, as $X_T^*$ converges in conditional distribution to $X_T$, in probability, if for every continuity point $x$ of $G_{X_T}(x)$, $G_{X_T}^*(x) \xrightarrow{p} G_{X_T}(x)$, as $T \rightarrow \infty$. For $x \in \mathbb{R}$, $\lfloor x \rfloor$ and $\lceil x \rceil$ denote the
largest integer not greater than $x$ and the smallest integer not less than $x$,
respectively.

\subsection{Regularity conditions for NGML estimation and bootstrap validity} \label{subsec:regcond}

    \begin{enumerate}[itemsep=1em]
    \item The distribution of the structural shock, $\varepsilon_{i,t}$ has a density $f_{i,\sigma_i}(x;\theta_i) = \sigma_i^{-1}f_i(\sigma_i^{-1}x;\theta_i)$
    which depends on the parameter vector $\theta_i$.

    \item The true parameter value $\lambda_0$ belongs to the parameter space $\Theta = \Theta_\Pi \times \Theta_\beta \times \Theta_\sigma \times
    \Theta_\theta$, where, (i) $\Theta_\Pi = \mathbb{R}^n \times \Theta_{\Pi_2}$ with $\Theta_{\Pi_2} \subseteq \mathbb{R}^{n^2p}$ such that the VAR stability
    condition holds for every $\Pi_2 \in \Theta_{\Pi_2}$, (ii)$\Theta_\beta = \operatorname{vecd}(\mathcal{B})$\footnote{$\mathcal{B}$ denotes the set of $n \times n$ 
    nonsingular
    matrices which have a specific permutation and scaling transformation to identify a unique $B$ matrix from each equivalence class corresponding
    to observationally equivalent SVAR processes, for details see \citet{ilmonen_semiparametrically_2011} and \citet{lanne_identification_2017}.} 
    $= \{\beta \in \mathbb{R}^{n(n-1)}: \beta = \operatorname{vecd}(B)
    \text{for some } B \in \mathcal{B}\}$, (iii)$\Theta_\sigma = \mathbb{R}^{n}_{+}$ and (iv)
    $\Theta_\theta = \Theta_{\theta_1} \times \cdots \times \Theta_{\theta_n}$,
    where 
    $\Theta_{\theta_i} = \left\{ (\alpha_i, \gamma_i, \delta_i, \mu_i) \in 
    (0,\infty) \times \mathbb{R} \times (0,\infty) \times \mathbb{R} \; : \; |\gamma_i| < \alpha_i, \alpha_i\delta_i \geq c \right\}, \text{for some} \ c > 0, 
    i = 1,2,\cdots,n.$

    \item  $f_i(x;\theta_i) > 0$ and $f_{i}(x;\theta_i)$ is twice continuously differentiable 
    with respect to $x$ and $\theta_i$, for all $x \in \mathbb{R}$ and
     $\theta_i \in \Theta_{\theta_i}, i = 1,2,\cdots,n$.

     \item  $ \int |\nabla f_{i,x}(x;\theta_{i,0})| dx < \infty$, where $\nabla f_{i,x}(x;\theta_{i,0}) = {\partial f_i(x;\theta_{i,0})}/{\partial x}$.
    
     \item For all $x \in \mathbb{R}$, $x^2\frac{\nabla f_{i,x}^2(x;\theta_{i,0})}{f_i(x;\theta_{i,0})}, \quad \frac{|| \nabla f_{i,\theta_i}(x;\theta_{i,0})||^2}{f_i^2(x;\theta_{i,0})}$ are dominated by 
     $c_1(1  + |x|^{c_2})$ \text{ for some constants } $c_1 > 0$ and $c_2 > 0$ \text{ and } $\int |x|^{c_2}f_i(x;\theta_{i,0}) dx < \infty$, where
     $\nabla f_{i,\theta_i}(x;\theta_{i,0}) = {\partial f_i(x;\theta_{i,0})}/{\partial \theta_i}$.

     \item $\int \text{sup}_{\theta_i \in \Theta_{0,\theta_i}} ||\nabla f_{i,\theta_i}(x;\theta_{i,0})|| dx < \infty $.
     
  \item The matrix $\mathbb{E}[\nabla \ell_{\lambda,t}(\lambda_0)\nabla \ell'_{\lambda,t}(\lambda_0)]$ is positive definite,
    where $\nabla \ell_{\lambda,t}(\lambda_0) = {\partial \ell_t(\lambda_0)}/{\partial \lambda}$.
    This condition can hold only when all structural shocks are non-Gaussian.

    \item[(7a)] When exactly one structural shock, say shock $k$, is Gaussian: the score
    covariance of the autoregressive and impact parameters $\psi$, and the score covariance
    of the shape parameters of the $n-1$ non-Gaussian shocks, are positive definite.
    \end{enumerate}

\subsection{Residual-based moving block bootstrap algorithm}\label{subsec:RMBB}

    \begin{enumerate}
        \item Compute the residuals from the fitted VAR model (in companion form), 
        
        \begin{equation}
            \hat{u}_t = y_t - \hat{\Pi}_1 y_{t-1}, \ t = 1,2,\cdots,T
        \end{equation} 
        (we write the algorithm for $p = 1$ without deterministic terms; with $p$ lags and deterministic terms the steps are the same with the companion form \eqref{eq:companionVAR})
        
        where $\hat{\Pi}_1$ is the $n \times n$ estimate of the autoregressive parameters, and obtain the NGML estimates of $\hat{\beta}_T$ (the off-diagonal elements of the normalized impact matrix $B(\beta)$) and $\hat{\sigma}_T$; consequently obtaining the estimate of the impact matrix, $\hat{B}_T \coloneqq B(\hat{\beta}_T)\operatorname{diag}(\hat{\sigma}_T)$, by maximizing the log-likelihood function \eqref{eq:loglik} and \eqref{eq:lik}, see Section~\ref{sec:ICA}.

        \item Choose a block length\footnote{We fix the block length in the residual-based MBB algorithm to the largest integer
        smaller than $5.03T^{1/4}$, see \citet{hall_blocking_1995}, \citet{jentsch_proxy_2016}, \citet{jentsch_dynamic_2019}, \citet{mertens_dynamic_2019} and
        \citet{angelini_identification_2024}. As a robustness check,  block lengths of $\ell=15,35$ (versus the rule-implied
        $\lfloor 5.03\cdot 500^{1/4}\rfloor=23$) yield similar results.}, $\ell \leq T$ and let  $L= \lceil T/\ell \rceil$ be the number of blocks, such that $lL \geq T$. 
        Form ($n \times \ell$)-dimensional blocks of residuals, $\hat{U}_{i,\ell} = [\hat{u}_{i+1}, \hat{u}_{i+2}, \cdots, \hat{u}_{i+\ell}], i = 0, \cdots, T-\ell$,
        and let $i_0,\cdots,i_{L-1}$ be i.i.d. random variables uniformly distributed on {$0,1,2,\cdots, T-\ell$}.

        \item Arrange the blocks $\hat{U}_{i_0,\ell},\cdots, \hat{U}_{i_{L-1},\ell}$ end to end, and discard the last $(lL -T)$ residuals, and obtain the bootstrap residuals,
        $(\hat{u}_1^*, \hat{u}_2^*, \cdots, \hat{u}_T^*)$.

        \item Center the bootstrap residuals, $(\hat{u}_1^*, \hat{u}_2^*, \cdots, \hat{u}_T^*)$, by computing:
        \begin{equation}
            \hat{u}_{jl+s}^{*} = \hat{u}_{jl+s}^* - \mathbb{E}^*(\hat{u}_{jl+s}^*) = \hat{u}_{jl+s}^* - \frac{1}{T-\ell+1}\sum_{r=0}^{T-\ell}\hat{u}_{s+r}
        \end{equation}
        for $s = 1,2,\cdots, \ell$, $j = 0,1,\cdots, L-1$ and $t = 1,2,\cdots,T$ to get $\mathbb{E}^*(u_t^*) = 0$.

        \item Set bootstrap pre-sample values (here, $y_0^*$) equal to zero, and generate the bootstrap time series as:
         $y_t^* = \hat{\Pi}_1 y_{t-1}^* + \hat{u}_t^{*}, \quad t = 1,2,\cdots,T.$ 
         \item Compute the residuals from the bootstrap fitted VAR model, 
         \begin{equation}
            \hat{u}_t^{**} = y_t^* - \hat{\Pi}_1^* y_{t-1}^*, \ t = 1,2,\cdots,T
        \end{equation}
        where, $\hat{\Pi}_1^*$ is the $n \times n$ bootstrap analog of $\hat{\Pi}_1$.

                \item Using the bootstrap residuals $\hat{u}_t^{**}$, compute the bootstrap NGML
                estimates
                \begin{equation}
                \hat{\lambda}_T^{*}
                = \arg\max_{\beta,\sigma,\theta}\ \mathcal{L}^{*}_T(\lambda),
                \qquad
                \mathcal{L}^{*}_T(\lambda) \coloneqq T^{-1}\sum_{t=1}^{T}\ell_t^{*}(\lambda),
                \end{equation}
                \begin{equation}
                \ell_t^{*}(\lambda)
                = \sum_{j=1}^{n}\log f_j\!\bigl(\sigma_j^{-1}\imath_j'B(\beta)^{-1}\hat{u}_t^{**};\,\theta_j\bigr)
                    \;-\;\log|\det B(\beta)|\;-\;\sum_{j=1}^{n}\log\sigma_j .
                \end{equation}
                This is the bootstrap analog of \eqref{eq:lik}: the bootstrap enters \emph{only}
                through the resampled residuals $\hat{u}_t^{**}$. The bootstrap analog is then
                $\hat{B}_T^{*}=B(\hat{\beta}_T^{*})\operatorname{diag}(\hat{\sigma}_T^{*})$.

                \item Repeat steps 2-7 independently to generate the bootstrap impact-matrix estimates $\{\hat{B}^{*}_{T,b}\}_{b=1}^{N}$, equivalently the studentized statistics
                    $\{Q^{*}_{T,b}\}_{b=1}^{N}$ with
                    \begin{equation}
                     Q^{*}_{T,b}=\hat{\Sigma}^{-1/2}\,T^{1/2}\operatorname{vec}\big(\hat{B}_T-B_0\big)
                    \end{equation}
                    where $\hat{\Sigma}_{B_T}$ is the estimated covariance of the impact-matrix estimator, see \eqref{eq:SigmaB}. The empirical bootstrap distribution $\hat{G}^{*}_{T,N}$ is formed from these draws as in \eqref{eq:edf}. For
                    confidence intervals and coverage we use $N$ replications; for the diagnostic
                    (Section~\ref{sec:BSdiag}) a normality test is applied to $M$ replications per test
                    sequence, and $S$ such sequences are used to summarize the empirical distribution of the $p$-values.

    \end{enumerate}

\section{Proofs}\label{appendix:Proofs}
\subsection{Proof of Proposition~\ref{prop:onegaussian}}\label{proof:onegaussian}
\begin{lemma}[Information at the single-Gaussian boundary]\label{lem:profileinfo}
Let conditions (1)--(6) and (7a) of Appendix~\ref{subsec:regcond} and Assumptions
\ref{as:stability}--\ref{as:indep} hold, and suppose exactly one structural shock,
say shock $k$, is Gaussian. Partition $\lambda=(\psi',\theta')'$ with
$\psi=(\operatorname{vec}(\Pi)',\beta',\sigma')'$ the autoregressive and impact parameters and $\theta=(\theta_1',\dots,\theta_n')'$ the shape parameters, and let
$\mathcal{I}(\lambda_0)=-\mathbb{E}[\nabla^2_{\lambda\lambda}\ell_t(\lambda_0)]$ and
$\mathcal{J}(\lambda_0)=\mathbb{E}[\nabla_\lambda\ell_t(\lambda_0)\nabla_\lambda\ell_t(\lambda_0)']$
be partitioned conformably. Then:
\begin{enumerate}
\item[(i)] the shape block $\mathcal{I}_{\theta\theta}(\lambda_0)$ is singular, with
null space $\mathcal{N}=\ker\mathcal{I}_{\theta\theta}$ equal to the shape directions
along which shock $k$'s Gaussian-limit density is invariant;
\item[(ii)] the cross-information blocks annihilate $\mathcal{N}$:
$\mathcal{I}_{\psi\theta}(\lambda_0)\,\mathcal{N}=\{0\}$ and
$\mathcal{J}_{\psi\theta}(\lambda_0)\,\mathcal{N}=\{0\}$;
\item[(iii)] the profile information
$\mathcal{I}_\psi=\mathcal{I}_{\psi\psi}-\mathcal{I}_{\psi\theta}\mathcal{I}_{\theta\theta}^{+}\mathcal{I}_{\theta\psi}$
is independent of the choice of generalized inverse $\mathcal{I}_{\theta\theta}^{+}$
and positive definite, with $\mathcal{J}_\psi$ the corresponding score
covariance.
\end{enumerate}
\end{lemma}
Proof of Lemma~\ref{lem:profileinfo}:
\subsubsection{Singularity and the null space of the shape block}
By \eqref{eq:NIGpdf} the NIG density of
shock $k$ converges to $\mathcal{N}(0,\sigma_k^2)$ as $\alpha_k\to\infty$ with
$\gamma_k=0$ and $\delta_k=\sigma_k^2\alpha_k$. In this limit, the shock depends on
$(\alpha_k,\delta_k)$ only through $\delta_k/\alpha_k=\sigma_k^2$. We can represent
$\rho_k\coloneqq1/\alpha_k\in[0,\infty)$, so the Gaussian limit is the
\emph{finite} boundary point $\rho_k=0$. Along the curve $\{(\alpha_k,\delta_k):\delta_k/\alpha_k=\sigma_k^2\}$, the NIG
density of shock $k$ approaches $\mathcal{N}(0,\sigma_k^2)$ with excess kurtosis
$3/(\delta_k\alpha_k)=3/(\sigma_k^2\alpha_k^2)=O(\rho_k^2)$ and all higher
cumulants of order $O(\rho_k^2)$.  The shock's deviation from the Gaussian is therefore
\emph{second order} in $\rho_k=1/\alpha_k$, so the first derivative of
$f_k$  along this curve vanishes at $\rho_k=0$.

Let us denote this curve direction by $v_\rho$. At $\gamma_k=0$ the $\gamma_k$-score is
$\varepsilon_{k,t}-\mu_k$ (Block~1 of the proof of Lemma~\ref{lemma:scoreconditions}),
which is not zero, while in the Gaussian limit the $\mu_k$-score converges to the
Gaussian location score $(\varepsilon_{k,t}-\mu_k)/\sigma_k^2$. The two scores are
therefore collinear at $\lambda_0$. The direction $v_\gamma$ that raises $\gamma_k$ by one
unit and lowers $\mu_k$ by $\sigma_k^2$, which changes the asymmetry while holding the
mean of shock $k$ fixed, satisfies $s_{\theta,t}'v_\gamma=0$ a.s. Hence
$\mathcal{N}=\operatorname{span}\{v_\rho,v_\gamma\}$ is a non-trivial subspace of
shape directions of shock $k$ along which the directional score satisfies
$s_{\theta,t}\,'v = 0$ a.s.\ at $\lambda_0$, so

$\mathcal{I}_{\theta\theta}\mathcal{N}=\{0\}$, i.e.\ $\mathcal{N}=\ker\mathcal{I}_{\theta\theta}$.
For the remaining $n-1$ non-Gaussian shocks the NIG family with finite tail index
is regular, so their shape blocks are nonsingular. This proves (i).

\subsubsection{Cross-information block annihilation} The $\psi$-score
$s_{\psi,t}=\partial\ell_t/\partial\psi$ depends on $\theta$ only through the shock $\varepsilon_t(\psi)$ and the densities $f_j(\cdot;\theta_j)$.
At $\lambda_0$, shock $k$'s density is exactly $\mathcal{N}(0,\sigma_k^2)$ and, by
(i), is exactly invariant along $\mathcal{N}$: the score
$s_{\theta,t}\,'v = 0$ a.s.\ identically for $\psi$ in a neighborhood of
$\psi_0$. Differentiating with respect to $\psi$ gives
$(\partial s_{\psi,t}/\partial\theta)\,v=0$ a.s., where
$\mathcal{I}_{\psi\theta}v=-\mathbb{E}[\partial s_{\psi,t}/\partial\theta]\,v=0$;
and $\mathcal{J}_{\psi\theta}v=\mathbb{E}[s_{\psi,t}\,(s_{\theta,t}'v)]=0$ since
$s_{\theta,t}'v = 0$. The annihilation
$\mathcal{I}_{\psi\theta}\mathcal{N}=\mathcal{J}_{\psi\theta}\mathcal{N}=\{0\}$ is
therefore \emph{exact}. Since the true parameter lies at the boundary
$\alpha_k=\infty$, there is no $O(\alpha_k^{-1})$ remainder (the leading departure
from invariance is $O(\alpha_k^{-2})$ and vanishes at first order). This proves (ii).

\subsubsection{Profile information}
The information matrix $\mathcal{I}_{\theta\theta}$ is symmetric and positive semidefinite, so
$\operatorname{range}(\mathcal{I}_{\theta\theta})=\mathcal{N}^{\perp}$, the orthogonal complement of the null space. By (ii)
every row of $\mathcal{I}_{\psi\theta}$ is orthogonal to $\mathcal{N}$, i.e.\
$\operatorname{range}(\mathcal{I}_{\theta\psi})\subseteq\operatorname{range}(\mathcal{I}_{\theta\theta})$,
which is why
$\mathcal{I}_{\psi\theta}\mathcal{I}_{\theta\theta}^{+}\mathcal{I}_{\theta\psi}$
is invariant to the choice of generalized inverse and equals its value at the
Moore--Penrose inverse. Hence, $\mathcal{I}_\psi$ is well-defined. For positive
definiteness, $\mathcal{I}_\psi$ is the information for $\psi$ in the model with
$\theta$ concentrated out along its identified directions. The impact matrix, equivalently $\psi$ up to the fixed signed column permutation, is point-identified
whenever at most one shock is Gaussian,
and the $n-1$ non-Gaussian shocks supply strictly negative criterion curvature in
every rotation direction in which they participate (the local-concavity condition
of \citet{gourieroux_statistical_2017} (Assumption A.4) holds for each
non-Gaussian pair). No non-zero $\psi$-direction therefore
lies in the null space of the profile curvature, so $\mathcal{I}_\psi\succ0$. This
proves (iii).

\begin{lemma}[Parity cancellation at the Gaussian boundary]\label{lem:parity}
Under $\mathcal{H}_0$ with shock $k$ Gaussian and the other shocks independent and
non-Gaussian, $\mathbb{E}[\partial^3\ell_t/\partial\psi\,\partial\rho_k^2|_{\rho_k=0}]=0$.
\end{lemma}
\begin{proof}
The order-$\rho_k^2$ term of $\log f_k$ is $c_0 He_4(\varepsilon_k)+o(\rho_k^2)$ for a
constant $c_0$\footnote{$He_n$ denotes the Hermite polynomial of
degree $n$, $He_n(x)=(-1)^n\phi(x)^{-1}\mathrm{d}^n\phi(x)/\mathrm{d}x^n$ with
$\phi$ the standard normal density. In particular, $He_3(x)=x^3-3x$ and
$He_4(x)=x^4-6x^2+3$. They follow $He_n'(x)=n\,He_{n-1}(x)$ and, for
$\varepsilon\sim\mathcal{N}(0,1)$, the orthogonality
$\mathbb{E}[He_m(\varepsilon)He_n(\varepsilon)]=n!\,\delta_{mn}$, so
$\mathbb{E}[He_n(\varepsilon)\,p(\varepsilon)]=0$ for every polynomial $p$ of degree
below $n$.}, so $\frac{\partial^2\ell_t}{\partial\rho_k^2}=2c_0He_4(\varepsilon_k)$ and
$\frac{\partial^3\ell_t}{\partial\psi\,\partial\rho_k^2}|_{\psi_0}=2c_0\,He_4'(\varepsilon_k)\,\frac{\partial\varepsilon_k}{\partial\psi}
=8c_0\,He_3(\varepsilon_k)\,\frac{\partial\varepsilon_k}{\partial\psi}$. The score $\frac{\partial\varepsilon_k}{\partial\psi}$ is, for the impact and scale
blocks $(\beta,\sigma)$, a linear function of the structural shocks,
and, for the autoregressive block, a function of the lagged data $Y_{t-1}$
independent of the current shock $\varepsilon_{k,t}$. In the first case
$\mathbb{E}[He_3(\varepsilon_k)\,\frac{\partial\varepsilon_k}{\partial\psi}]=0$ because $He_3$ is
orthogonal under the standard Gaussian distribution to every polynomial of degree at most
two, so $\mathbb{E}[He_3(\varepsilon_k)\,\varepsilon_{m}]=0$ for all $m$ (including
$m=k$); in the second, $\mathbb{E}[He_3(\varepsilon_{k,t})\,\partial\varepsilon_{k,t}/\partial\psi]
=\mathbb{E}[He_3(\varepsilon_{k,t})]\,\mathbb{E}[\partial\varepsilon_{k,t}/\partial\psi]=0$
by independence and $\mathbb{E}[He_3]=0$. Hence, $\mathbb{E}[\partial^3\ell_t/\partial\psi\,\partial\rho_k^2|_{\rho_k=0}]=0$.
\end{proof}

\subsubsection{Proof of Proposition~\ref{prop:onegaussian}}
Under the representation $\rho_k=1/\alpha_k$ the Gaussian limit is the finite
boundary point $\rho_k=0$ and the criterion is smooth in $\psi$ for $\theta$ in a
neighborhood of this point.
\begin{enumerate}
    \item At $\rho_k=0$, the NIG density departs from the Gaussian limit case at order $\rho_k^2$ (its
    excess kurtosis is $3\rho_k^2/\sigma_k^2$). Hence, the $\rho_k$-score vanishes and the
    information for $\rho_k$ is zero: the shape parameter is locally unidentified and is
    estimated at the slow rate $\hat\rho_{k,T}=O_p(T^{-1/4})$. The boundary theory of
    \citet{andrews_estimation_1999}, which requires a nonsingular information at the
    boundary, therefore does not apply, and we argue directly.

    \item Since $\mathcal{I}_{\psi\rho_k}=0$ (Lemma~\ref{lem:profileinfo}(ii)) and the
    cross-Hessian $\partial^2\ell_t/\partial\psi\,\partial\rho_k$ vanishes at $\rho_k=0$,
\begin{equation}
  \Delta s_\psi=\tfrac12\,\bar c\,\hat\rho_{k,T}^{2}+o_p(\hat\rho_{k,T}^{2}),
  \qquad
  \bar c=\mathbb{E}\!\left[\partial^3\ell_t/\partial\psi\,\partial\rho_k^2\big|_{\rho_k=0}\right].    
\end{equation}
The leading correction is the kurtosis term, proportional to $He_4(\varepsilon_k)$, so
$\partial^3\ell_t/\partial\psi\,\partial\rho_k^2|_{0}\propto
He_3(\varepsilon_k)\,\partial\varepsilon_k/\partial\psi$. By Lemma~\ref{lem:parity},
$\bar c=0$, because $He_3$ is orthogonal under the Gaussian distribution to the degree-one
sensitivities $\partial\varepsilon_k/\partial\psi$. Thus, $\sqrt{T}\,\Delta s_\psi=o_p(1)$:
the zero-information nuisance does not enter the distribution of $\hat\psi_T$.

\item Finally, since profiling out $\rho_k$ perturbs the $\psi$-score by $o_p(T^{-1/2})$,
the profile score for $\psi$ is asymptotically equivalent to the score evaluated at a
known $\rho_k$. The estimation of $\psi$ is then a regular ML object with nonsingular profile information $\mathcal{I}_\psi$
of Lemma~\ref{lem:profileinfo}(iii); standard QMLE asymptotics from
\citet{white_maximum_1982} give us
\[
  T^{1/2}(\hat\psi_T-\psi_0)\xrightarrow{d}
  \mathcal{N}\!\bigl(0,\;\mathcal{I}_\psi^{-1}\mathcal{J}_\psi\mathcal{I}_\psi^{-1}\bigr),
\]
where the profile-likelihood construction follows \citet{murphy_profile_2000}, and the delta method through $B(\beta)\operatorname{diag}(\sigma)$ yields
\eqref{eq:Bblock_normal}. When all shocks are non-Gaussian the model is a correctly
specified, the information identity gives
$\mathcal{I}_\psi=\mathcal{J}_\psi$, and $\mathcal{V}_B$ collapses to \eqref{eq:SigmaB}.
\end{enumerate}

\subsection{Proof of Proposition~\ref{prop:BSvalidity}} \label{subsec:proofBSvalidity}

We collect all the model parameters in a vector form as $\lambda \coloneqq (\operatorname{vec}(\Pi), \operatorname{vecd}(B), \sigma, \theta)$, where $\operatorname{vec}(\Pi)$ is a vector of autoregressive coefficients, $\operatorname{vecd}(B)$ is a vector containing the off-diagonal elements of the normalized impact matrix $B$, $\sigma$ is a vector of standard deviations of shocks' distributions and $\theta = (\theta_1', \theta_2', \cdots, \theta_n')'$ is a vector collecting the shape parameters of each of the $n$ independent, structural shocks, see Section~\ref{subsec:NGMLest}. Let $\lambda_0$ be the true parameter vector and $\hat{\lambda}_T$ be the NGML estimator defined in (\ref{eq:NGMLestimator}) and (\ref{eq:BT}), and its bootstrap analog $\hat{\lambda}_T^{\ast}$ defined in Algorithm~\ref{subsec:RMBB}. Furthermore, 
let $\mathcal{L}_T(\lambda) \coloneqq T^{-1}\sum_{t = 1}^{T}\ell_t(\lambda)$ be the log-likelihood function defined in (\ref{eq:loglik}) and (\ref{eq:lik}), and its bootstrap analog $\mathcal{L}_T^{\ast}(\lambda) \coloneqq T^{-1}\sum_{t = 1}^{T}\ell_t^{\ast}(\lambda)$, where $\ell_t^{\ast}(\lambda)$ is defined in step 7 of Algorithm~\ref{subsec:RMBB}. Also, let the score functions be defined as $S_T(\lambda) \coloneqq T^{-1} \sum_{t = 1}^{T}\nabla_{\lambda}\ell_t(\lambda)$ and its bootstrap analog as $\hat{S}_T^{\ast}(\lambda) \coloneqq T^{-1}\sum_{t = 1}^{T}\nabla_{\lambda}\ell_t^{\ast}(\lambda)$. 

We establish Proposition~\ref{prop:BSvalidity} through two lemmas. Lemma~\ref{lemma:scoreconditions} verifies that the NGML score satisfies the conditions required for a bootstrap central limit theorem, and Lemma~\ref{lemma:boot_score} uses them to establish the asymptotic normality of the bootstrap score under the residual-based moving block bootstrap. Combining these with the consistency of $\hat{\lambda}_T$ \citet{lanne_identification_2017}(Theorem 1) and of the bootstrap estimator $\hat{\lambda}^{\ast}_T$ \citet{bruggemann_inference_2016}(Theorem~4.1), a standard Taylor expansion of the bootstrap score completes the proof.

\begin{lemma}[Score conditions under NIG shocks]\label{lemma:scoreconditions}
For the likelihood function defined through Equations (\ref{eq:NIGpdf})-(\ref{eq:lik}) under the parameter space $\Theta_{\theta_i}$ of the NIG distributed shocks, 
\[
\Theta_{\theta_i}
=
\left\{
(\alpha_i,\gamma_i,\delta_i,\mu_i)
\in
(0,\infty)\times\mathbb{R}\times(0,\infty)\times\mathbb{R}
:
|\gamma_i|<\alpha_i, \alpha_i\delta_i \geq c
\right\}, \ \text{for some} \ c > 0,
\]
and Assumptions~\ref{as:stability}--\ref{as:gau} with the regularity conditions of Appendix~\ref{subsec:regcond}, the following hold for each structural shock $i=1,2,\ldots,n$:

\begin{enumerate}
\item[(i)] The score function $\nabla_\lambda \ell_t(\lambda_0)$ is a stationary, ergodic martingale difference sequence with respect to the natural filtration
\[
\mathcal{F}_{t-1}
=
\sigma(\{\varepsilon_s:s\le t-1\}),
\]
equivalently, $\sigma(\{Y_s:s\le t-1\})$, at the true parameter $\lambda_0$, satisfying
\[
\mathbb{E}\!\left[
\nabla_\lambda \ell_t(\lambda_0)
\mid
\mathcal{F}_{t-1}
\right]
=
0
\quad \text{a.s.}
\]

\item[(ii)] The score is square-integrable:
\[
\mathbb{E}\!\left[
\|\nabla_\lambda \ell_t(\lambda_0)\|^2
\right]
<\infty.
\]

\item[(iii)] The score components with respect to the NIG parameters
$(\alpha_i,\gamma_i,\delta_i,\mu_i)$ satisfy the integrability conditions (4) and (5) of Appendix~\ref{subsec:regcond}.
\end{enumerate}
\end{lemma}
\begin{remark}
The proofs of Parts (i), (ii), and (iii) are mutually independent. Part (i) invokes the integrability conditions of Appendix~\ref{subsec:regcond}, which are verified for the NIG distribution in Part (iii); and invokes finite first moments, which follow from Part (ii) via Jensen's inequality $\mathbb{E}[\|\nabla_\lambda \ell_t\|] \leq \sqrt{\mathbb{E}[\|\nabla_\lambda \ell_t\|^2]} < \infty$. These forward references are non-circular since Parts (ii) and (iii) are established without reference to Part (i).
\end{remark}
\subsubsection{Proof of Part (i)}

At the true parameter $\lambda_0$, the score $\nabla_\lambda \ell_t(\lambda_0)$ is a function of $\varepsilon_t$ and $\lambda_0$ (up to some constant). Therefore, $\nabla_\lambda \ell_t(\lambda_0)$ is a measurable function of $\varepsilon_t$ and $\lambda_0$.  By Assumption~\ref{as:indep}, $\varepsilon_t$ is \textit{i.i.d.}, so it is independent of $\mathcal{F}_{t-1}$. Since $\nabla_\lambda \ell_t(\lambda_0)$ is a measurable function of $\varepsilon_t$, and independent of $\mathcal{F}_{t-1}$, it follows from \citet{billingsley_probability_1995} (Section 35.2) that $\nabla_\lambda \ell_t(\lambda_0)$ is a martingale difference sequence with respect to $\mathcal{F}_{t-1}$, since:
\[
\mathbb{E}\!\left[
\nabla_\lambda \ell_t(\lambda_0)
\mid
\mathcal{F}_{t-1}
\right]
=
\mathbb{E}\!\left[
\nabla_\lambda \ell_t(\lambda_0)
\right].
\]
The latter expectation equals zero by the Fisher score identity, through conditions (3)--(5) of Appendix~\ref{subsec:regcond}. Stationarity and ergodicity follow from the \textit{i.i.d.} structure of $\varepsilon_t$ (\citet{billingsley_probability_1995}, Theorem 24.1).
\begin{remark}\label{remark:ergodicity}
Finite first moments $\mathbb{E}[\|\nabla_\lambda \ell_t(\lambda_0)\|] < \infty$ follow from 
\begin{equation}
    \mathbb{E}[\|\nabla_\lambda \ell_t(\lambda_0)\|] \leq \sqrt{\mathbb{E}[\|\nabla_\lambda \ell_t(\lambda_0)\|^2]} < \infty
\end{equation}
by Jensen's inequality, where the square-integrability is established independently in Part (ii). The expectation operator is applied componentwise to each element of $\nabla_\lambda \ell_t(\lambda_0)$ (\citet{billingsley_probability_1995}, Theorem 22.1), with vector convergence following from the equivalence of componentwise and joint convergence in $\mathbb{R}^d$ (\citet{billingsley_probability_1995}, Theorem 29.4).
\end{remark}

\subsubsection{Proof of Part (ii)}

We establish $\mathbb{E}\!\left[
\|\nabla_\lambda \ell_t(\lambda_0)\|^2
\right]
<\infty$ by examining each block of the parameter vector $\lambda=(\Pi, \beta, \sigma, \theta)$ separately. The parameter vector $\theta_i = (\alpha_i,\gamma_i,\delta_i,\mu_i)$, for each shock $i$, also contains four components, each requiring separate treatment. We first establish a preliminary result on Bessel function ratios that is used throughout all blocks.

\textbf{Bessel Ratio Uniform Boundedness:} The score components with respect $(\alpha_i,\delta_i, \mu_i)$ involve ratios of modified Bessel functions $K_\nu(r_{i,t})/K_1(r_{i,t})$ for $\nu \in \{0,2\}$, where:
\begin{equation}
r_{i,t} \coloneqq \alpha_i\sqrt{\delta_i^2 + (\varepsilon_{i,t} - \mu_i)^2}
\end{equation}

We establish that these ratios are uniformly bounded over all $\varepsilon_{i,t} \in \mathbb{R}$ and all $(\alpha_i,\gamma_i,\delta_i) \in \Theta_{\theta_i}$.
\begin{enumerate}
\item Lower bound on $r_{i,t}$:
Since $\delta_i > 0$ and $\alpha_i > 0$ on $\Theta_{\theta_i}$:
$r_{i,t} = \alpha_i\sqrt{\delta_i^2+(\varepsilon_{i,t}-\mu_i)^2} \geq \alpha_i\delta_i \geq c > 0 \quad \text{for all } \varepsilon_{i,t} \in \mathbb{R}$.
This rules out $r_{i,t} \to 0$, where $K_1$ has a singularity of the form $K_1(z) \sim 1/z$ as $z \to 0$. The modified Bessel functions $K_\nu$ are strictly positive for all $r > 0$, so the ratios $K_\nu(r_{i,t})/K_1(r_{i,t})$
are well-defined. Each ratio is continuous on $[\alpha_i\delta_i,\infty)$ since $K_\nu$ is continuous on $(0,\infty)$
and $K_1(r) > 0$ for all $r > 0$ \citep{abramowitz_handbook_1974} (henceforth, \textbf{AM-SI}), Section 9.6.

\item Behavior as $r_{i,t} \to \infty$:
By the uniform asymptotic expansion of modified Bessel functions \textbf{AM-SI} (Equation 9.7.2):

\[
K_\nu(z) = \sqrt{\frac{\pi}{2z}}e^{-z}\left(1 + \frac{4\nu^2-1}{8z} + O(z^{-2})\right) \quad \text{as } z \to \infty
\]

Taking ratios for $\nu \in \{0,1,2\}$ and simplifying:

\begin{equation}\label{eq:ratio1bound}
\frac{K_0(r)}{K_1(r)} = \frac{1 - \frac{1}{8r} + O(r^{-2})}{1 + \frac{3}{8r} + O(r^{-2})} \longrightarrow 1 \quad \text{as } r \to \infty    
\end{equation}

\begin{equation}\label{eq:ratio2bound}
\frac{K_2(r)}{K_1(r)} = \frac{1 + \frac{15}{8r} + O(r^{-2})}{1 + \frac{3}{8r} + O(r^{-2})} \longrightarrow 1 \quad \text{as } r \to \infty
\end{equation}

\item Uniform boundedness on $[\alpha_i\delta_i, \infty)$:
Each ratio $K_\nu(r)/K_1(r)$ is continuous on $[\alpha_i\delta_i,\infty)$ from equations (\ref{eq:ratio1bound}) and (\ref{eq:ratio2bound}). On any compact subset $[\alpha_i\delta_i,R]$ for finite $R$, continuity implies boundedness. On $[R,\infty)$, both ratios converge to 1, so they are bounded there as well. Combining these, there exist finite constants $C_0,C_2<\infty$ depending only on the lower bound $c$ in $\Theta_{\theta_i}$ such that:

\begin{equation}
    \label{eq:BesselRatioBound}
\frac{K_\nu(r_{i,t})}{K_1(r_{i,t})} \leq C_\nu \quad \text{for all } \varepsilon_{i,t} \in \mathbb{R}, \text{ all } (\alpha_i,\gamma_i,\delta_i,\mu_i) \in \Theta_{\theta_i}, \quad \nu \in \{0,2\}
\end{equation}
We use (\ref{eq:BesselRatioBound}) throughout all subsequent blocks.
\end{enumerate}

\textbf{Block 1: Score with respect to $\gamma_i$.}
By direct differentiation of \eqref{eq:NIGpdf} with respect to $\gamma_i$:
\begin{equation}
\frac{\partial\log f_i(\varepsilon_{i,t};\theta_{i,0})}{\partial\gamma_i} = (\varepsilon_{i,t}-\mu_i) - \frac{\delta_i\gamma_i}{\sqrt{\alpha_i^2-\gamma_i^2}}
\end{equation}
The second term is a constant and finite in $\varepsilon_{i,t}$ since $|\gamma_i| < \alpha_i$ on $\Theta_{\theta_i}$
ensures $\alpha_i^2-\gamma_i^2 > 0$, and $\delta_i < \infty$
on any compact subset of $\Theta_{\theta_i}$. Square-integrability requires $\mathbb{E}[(\varepsilon_{i,t}-\mu_i)^2] < \infty$. The NIG distribution has finite variance: $\operatorname{Var}(\varepsilon_{i,t}) = \frac{\delta_i\alpha_i^2}{(\alpha_i^2-\gamma_i^2)^{3/2}} < \infty$
for all $(\alpha_i,\gamma_i,\delta_i) \in \Theta_{\theta_i}$, since $|\gamma_i| < \alpha_i$ ensures $(\alpha_i^2-\gamma_i^2) > 0$ by \citet{barndorff-nielsen_normal_1997} (Proposition 2). Therefore:
\begin{equation}
\mathbb{E}\left[\left(\frac{\partial\log f_i}{\partial\gamma_i}\right)^2\right] \leq 2\mathbb{E}[(\varepsilon_{i,t}-\mu_i)^2] + 2\left(\frac{\delta_i\gamma_i}{\sqrt{\alpha_i^2-\gamma_i^2}}\right)^2 = 2\,\operatorname{Var}(\varepsilon_{i,t}) + 2\left(\frac{\delta_i\gamma_i}{\sqrt{\alpha_i^2-\gamma_i^2}}\right)^2 < \infty
\end{equation}
where the inequality uses $(a+b)^2 \leq 2a^2+2b^2$.

\textbf{Block 2: Score with respect to $\alpha_i$.}
By direct differentiation of \eqref{eq:NIGpdf} with respect to $\alpha_i$, applying the modified Bessel recurrence relation:
\begin{equation}\label{eq:besselreccurence}
    K_1'(z) = -\frac{1}{2}(K_0(z)+K_2(z))
\end{equation}
\textbf{AM-SI} (Equation 9.6.26) via the chain rule to differentiate $\log K_1(r_{i,t})$
 with respect to $\alpha_i$ and noting that $\partial r_{i,t}/\partial\alpha_i = \sqrt{\delta_i^2+(\varepsilon_{i,t}-\mu_i)^2}$, we have:
\begin{equation}
\frac{\partial\log f_i}{\partial\alpha_i} = \frac{1}{\alpha_i} + \frac{\delta_i\alpha_i}{\sqrt{\alpha_i^2-\gamma_i^2}} - \frac{\sqrt{\delta_i^2+(\varepsilon_{i,t}-\mu_i)^2}}{2}\left[\frac{K_0(r_{i,t})}{K_1(r_{i,t})} + \frac{K_2(r_{i,t})}{K_1(r_{i,t})}\right]    
\end{equation}

Applying the uniform bound \eqref{eq:BesselRatioBound}, the Bessel-ratio sum
satisfies,
\begin{equation}
    |K_0(r_{i,t})/K_1(r_{i,t})+K_2(r_{i,t})/K_1(r_{i,t})|\le C_0+C_2<\infty
\end{equation}
and $\sqrt{\delta_i^2+(\varepsilon_{i,t}-\mu_i)^2}\le\delta_i+|\varepsilon_{i,t}-\mu_i|$
by $\sqrt{a^2+b^2}\le a+b$ for $a,b\ge0$. Hence,
\begin{equation}
\left|\frac{\partial\log f_i}{\partial\alpha_i}\right|
 \le \frac{1}{\alpha_i}+\frac{\delta_i\alpha_i}{\sqrt{\alpha_i^2-\gamma_i^2}}
   +\frac{C_0+C_2}{2}\bigl(\delta_i+|\varepsilon_{i,t}-\mu_i|\bigr)
 \;=\; A_i + B_i\,|\varepsilon_{i,t}-\mu_i|,
\end{equation}
where,
$A_i \coloneqq \dfrac{1}{\alpha_i}+\dfrac{\delta_i\alpha_i}{\sqrt{\alpha_i^2-\gamma_i^2}}
+\dfrac{(C_0+C_2)\delta_i}{2}$ and $B_i \coloneqq \dfrac{C_0+C_2}{2}$ are finite constants
on $\Theta_{\theta_i}$. 
Squaring, using $(a+b)^2\le2a^2+2b^2$, and taking expectations,
\begin{equation}
\mathbb{E}\!\left[\left(\frac{\partial\log f_i}{\partial\alpha_i}\right)^2\right]
 \le 2A_i^2 + 2B_i^2\,\mathbb{E}\!\bigl[(\varepsilon_{i,t}-\mu_i)^2\bigr] < \infty,
\end{equation}
where $\mathbb{E}[(\varepsilon_{i,t}-\mu_i)^2]
=\operatorname{Var}(\varepsilon_{i,t})+\bigl(\delta_i\gamma_i/\sqrt{\alpha_i^2-\gamma_i^2}\bigr)^2<\infty$
by the finite mean and variance of the NIG distribution on $\Theta_{\theta_i}$.

\textbf{Block 3: Score with respect to $\mu_i$.}
By direct differentiation of \eqref{eq:NIGpdf} with respect to $\mu_i$, applying the Bessel recurrence relation \eqref{eq:besselreccurence} via the chain rule to differentiate $\log K_1(r_{i,t})$ with respect to $\mu_i$, noting that $\partial r_{i,t}/\partial\mu_i = -\alpha_i(\varepsilon_{i,t}-\mu_i)/\sqrt{\delta_i^2+(\varepsilon_{i,t}-\mu_i)^2}$, and differentiating the remaining terms in equation \eqref{eq:lik}:
\begin{equation}
    \frac{\partial\log f_i}{\partial\mu_i} = \frac{(\varepsilon_{i,t}-\mu_i)}{\delta_i^2+(\varepsilon_{i,t}-\mu_i)^2} + \frac{\alpha_i(\varepsilon_{i,t}-\mu_i)}{2\sqrt{\delta_i^2+(\varepsilon_{i,t}-\mu_i)^2}}\left[\frac{K_0(r_{i,t})}{K_1(r_{i,t})} + \frac{K_2(r_{i,t})}{K_1(r_{i,t})}\right] - \gamma_i
\end{equation}

We establish that all three terms are uniformly bounded over all $\varepsilon_{i,t} \in \mathbb{R}$.
\begin{enumerate}
    \item Term 1: By the inequality $ a^2+b^2 \geq 2ab$, $\delta_i^2+(\varepsilon_{i,t}-\mu_i)^2 \geq 2\delta_i|\varepsilon_{i,t}-\mu_i|$ for all $\varepsilon_{i,t} \in \mathbb{R}$, so:
\begin{equation}
\left|\frac{\varepsilon_{i,t}-\mu_i}{\delta_i^2+(\varepsilon_{i,t}-\mu_i)^2}\right| \leq \frac{|\varepsilon_{i,t}-\mu_i|}{2\delta_i|\varepsilon_{i,t}-\mu_i|} = \frac{1}{2\delta_i} < \infty    
\end{equation}
uniformly over all $\varepsilon_{i,t} \in \mathbb{R}$, with $\delta_i > 0$ guaranteed on $\Theta_{\theta_i}$.
\item Term 2: The prefactor satisfies:
\begin{equation}
\left|\frac{\alpha_i(\varepsilon_{i,t}-\mu_i)}{2\sqrt{\delta_i^2+(\varepsilon_{i,t}-\mu_i)^2}}\right| \leq \frac{\alpha_i}{2}
\end{equation}
for all $\varepsilon_{i,t} \in \mathbb{R}$, since $\sqrt{\delta_i^2+(\varepsilon_{i,t}-\mu_i)^2} \geq |\varepsilon_{i,t}-\mu_i|$ for all $\varepsilon_{i,t}$. Combined with the uniform bound \eqref{eq:BesselRatioBound} on the Bessel ratios, Term 2 is uniformly bounded by $\alpha_i(C_0+C_2)/2 < \infty$ over all $\varepsilon_{i,t} \in \mathbb{R}$.
\item Term 3: $|-\gamma_i| < \alpha_i < \infty$. Since all three terms are uniformly bounded, the score with respect to $\mu_i$ satisfies:
\begin{equation}
\left|\frac{\partial\log f_i}{\partial\mu_i}\right| \leq \frac{1}{2\delta_i} + \frac{\alpha_i(C_0+C_2)}{2} + |\gamma_i| \coloneqq D < \infty
\end{equation}
where $D$ depends only on $\Theta_{\theta_i}$ and the Bessel bounds $C_0,C_2$. Therefore:
\begin{equation}
\mathbb{E}\left[\left(\frac{\partial\log f_i}{\partial\mu_i}\right)^2\right] \leq D^2 < \infty    
\end{equation}
\end{enumerate}

This bound requires no moment conditions on $\varepsilon_{i,t}$ beyond the existence of its density, which is a stronger result than was needed for $(\gamma_i,\alpha_i,\delta_i)$. The uniform boundedness arises because both $(\varepsilon_{i,t}-\mu_i)/[\delta_i^2+(\varepsilon_{i,t}-\mu_i)^2]$
and $(\varepsilon_{i,t}-\mu_i)/\sqrt{\delta_i^2+(\varepsilon_{i,t}-\mu_i)^2}$ are bounded functions of $\varepsilon_{i,t}$: in both cases the denominator grows at least as fast as the numerator as $|\varepsilon_{i,t}| \to \infty$.

\textbf{Block 4: Score with respect to $\delta_i$.}
By direct differentiation of equation \eqref{eq:lik} with respect to $\delta_i$,
applying the Bessel recurrence relation \eqref{eq:besselreccurence} to
differentiate $\log K_1(r_{i,t})$, and noting
\begin{equation}
 \frac{\partial r_{i,t}}{\partial\delta_i}
   =\frac{\alpha_i\delta_i}{\sqrt{\delta_i^2+(\varepsilon_{i,t}-\mu_i)^2}}   
\end{equation}
we obtain
\begin{equation}
\frac{\partial\log f_i}{\partial\delta_i}
 = \frac{1}{\delta_i}
 - \frac{\delta_i}{\delta_i^2+(\varepsilon_{i,t}-\mu_i)^2}
 + \sqrt{\alpha_i^2-\gamma_i^2}
 - \frac{\alpha_i\delta_i}{2\sqrt{\delta_i^2+(\varepsilon_{i,t}-\mu_i)^2}}
   \left[\frac{K_0(r_{i,t})}{K_1(r_{i,t})}+\frac{K_2(r_{i,t})}{K_1(r_{i,t})}\right].
\end{equation}
The first and third terms are constants on $\Theta_{\theta_i}$. The second term
satisfies $\bigl|\delta_i/(\delta_i^2+(\varepsilon_{i,t}-\mu_i)^2)\bigr|\leq 1/\delta_i$,
attained at $\varepsilon_{i,t}=\mu_i$. For the fourth term, the coefficient of the
Bessel-ratio sum satisfies
\begin{equation}
\left|\frac{\alpha_i\delta_i}{2\sqrt{\delta_i^2+(\varepsilon_{i,t}-\mu_i)^2}}\right|
 \leq \frac{\alpha_i\delta_i}{2\delta_i}=\frac{\alpha_i}{2},
\end{equation}
using $\sqrt{\delta_i^2+(\varepsilon_{i,t}-\mu_i)^2}\geq\delta_i$; combined with the
uniform bound \eqref{eq:BesselRatioBound}, the fourth term is bounded by
$\alpha_i(C_0+C_2)/2$. Every term is therefore uniformly bounded, and
\begin{equation}
\left|\frac{\partial\log f_i}{\partial\delta_i}\right|
 \leq \frac{2}{\delta_i}+\sqrt{\alpha_i^2-\gamma_i^2}+\frac{\alpha_i(C_0+C_2)}{2}
 \;=:\; D' \;<\;\infty .
\end{equation}
As in Block 3, uniform boundedness yields square-integrability with no moment
condition on $\varepsilon_{i,t}$.

\textbf{Block 5: Score with respect to $\sigma_i$.}
The parameter $\sigma_i$ enters the log-likelihood contribution $\ell_t(\lambda_0)$ in equation \eqref{eq:lik} through the standardization of the structural shock and through the Jacobian term $\log\sigma_i$. At the true parameter $\lambda_0$, differentiating $\ell_t(\lambda_0)$ with respect to $\sigma_i$ via the chain rule yields:
\begin{equation}
\frac{\partial \ell_t(\lambda_0)}{\partial\sigma_i} = -\frac{\varepsilon_{i,t}}{\sigma_{i,0}^2} \frac{\partial\log f_i(s;\theta_{i,0})}{\partial s}\bigg|_{s=\sigma_{i,0}^{-1}\varepsilon_{i,t}} - \frac{1}{\sigma_{i,0}}
\end{equation}
By direct differentiation of equation \eqref{eq:lik} with respect to the argument $s_i$, the score with respect to $s_i$ has the identical structure as the score with respect to $\mu_i$ in Block 3, with the same Bessel ratios and bounded fractions. Therefore, by the same argument as in Block 3, the score with respect to $s_i$ is uniformly bounded by a finite constant $D''$ depending only on $\Theta_{\theta_i}$ and the Bessel bounds $C_0,C_2$. Hence: 
\begin{equation}\label{eq:score_s_bound}
\left|\frac{\partial\log f_i(s;\theta_{i,0})}{\partial s}\right| \leq \frac{1}{2\delta_i} + \frac{\alpha_i(C_0+C_2)}{2} + |\gamma_i| =: D'' < \infty 
\end{equation}
Substituting this into the score with respect to $\sigma_i$:
\begin{equation}
\left|\frac{\partial \ell_t(\lambda_0)}{\partial\sigma_i}\right| \leq \frac{|\varepsilon_{i,t}|}{\sigma_{i,0}^2}D + \frac{1}{\sigma_{i,0}}
\end{equation}
Squaring (and applying $(a+b)^2 \leq 2a^2+2b^2$) and taking expectations:
\begin{equation}
    \mathbb{E}\left[(\frac{\partial \ell_t(\lambda_0)}{\partial\sigma_i})^2\right] \leq \frac{2D^2}{\sigma_{i,0}^4}\mathbb{E}[\varepsilon_{i,t}^2] + \frac{2}{\sigma_{i,0}^2} < \infty
\end{equation}
With the finite variance of $\varepsilon_{i,t}$, the score with respect to $\sigma_i$ is square-integrable.

\textbf{Block 6: Score with respect to autoregressive parameters $\Pi$ and off-diagonal elements $\beta$.}
At the true parameter $\lambda_0$, the score components with respect to $(\Pi,\beta)$ are linear functions of the structural score $\frac{\partial\log f_i(s;\theta_{i,0})}{\partial s}$ evaluated at $s = \sigma_{i,0}^{-1}\varepsilon_{i,t}$, scaled by $Y_{t-1}$ and elements of $B_0^{-1}$. By (\ref{eq:score_s_bound}), the structural score is uniformly bounded by $D < \infty$ over all $\varepsilon_{i,t} \in \mathbb{R}$. Therefore,
\begin{equation}
    \mathbb{E}[||Y_{t-1}||^2||\nabla_\varepsilon\log f_i(\varepsilon_{i,t};\theta_{i,0})||^2] \leq D^2\mathbb{E}[||Y_{t-1}||^2]
\end{equation}
Under Assumption~\ref{as:stability} and the stationarity of $Y_t$, $\mathbb{E}[||Y_{t-1}||^2] < \infty$. Therefore, the score components with respect to $(\Pi,\beta)$ are square-integrable.

\subsubsection{Proof of Part (iii)}
Since $\nabla_{f_i,x}(x;\theta_{i,0})=f_i(x;\theta_{i,0})\,\partial\log f_i(x;\theta_{i,0})/\partial x$,
the uniform bound gives
$\int|\nabla_{f_i,x}(x;\theta_{i,0})|\,dx
  \leq D\int f_i(x;\theta_{i,0})\,dx = D < \infty$. Condition (4) of Appendix~\ref{subsec:regcond} requires $\int|\nabla_{f_i,x}(x;\theta_{i,0})|\,dx < \infty$, where $\nabla_{f_i,x}(x;\theta_{i,0})$ is the score of the density $f_i$ with respect to its argument $x$.
By direct differentiation of equation \eqref{eq:lik} with respect to $x$, the score $\partial\log f_i(x;\theta_{i,0})/\partial x$ has the same algebraic structure as $\partial\log f_i/\partial\mu_i$
established in Block 3 of Part (ii). Specifically, every occurrence of $(x-\mu_i)$
is divided by $\delta_i^2+(x-\mu_i)^2$ or $\sqrt{\delta_i^2+(x-\mu_i)^2}$, which are uniformly bounded fractions as shown in Block 3. Therefore, $|\partial\log f_i/\partial x| \leq D < \infty$ uniformly over all $x \in \mathbb{R}$, where $D$ is the same constant as Block 3. Condition (4) in Appendix~\ref{subsec:regcond} holds.

Condition (5) of Appendix~\ref{subsec:regcond} requires $\|\nabla_{f_i,\theta_i}(x;\theta_{i,0})\|^2/f_i^2(x;\theta_{i,0})$ to be dominated by $c_1(1+|x|^{c_2})$ for constants $c_1,c_2 > 0$, with $\int|x|^{c_2}f_i(x;\theta_{i,0})\,dx < \infty$.

The ratio $\|\nabla_{f_i,\theta_i}(x;\theta_{i,0})\|^2/f_i^2(x;\theta_{i,0}) = \|\nabla_{\theta_i}\log f_i(x;\theta_{i,0})\|^2$ is the squared norm of the score with respect to $\theta_i = (\alpha_i,\gamma_i,\delta_i,\mu_i)$. From Blocks 1-4: the $\gamma_i$ and $\alpha_i$
components satisfy $\partial\log f_i/\partial\theta_{i,k}| \leq A_k+B_k|x-\mu_i|$; the $\delta_i$ and $\mu_i$ components are uniformly bounded by finite constants $D',D$ respectively. Therefore:
\begin{equation}
    \|\nabla_{\theta_i}\log f_i(x;\theta_{i,0})\|^2 = \sum_{k=1}^4\left(\frac{\partial\log f_i}{\partial\theta_{i,k}}\right)^2 \leq (A_1+B_1|x-\mu_i|)^2 + (A_2+B_2|x-\mu_i|)^2 + D'^2 + D^2
\end{equation}
Applying $(a+b)^2 \leq 2a^2+2b^2$ to the first two terms and combining constants:
\begin{equation}
\|\nabla_{\theta_i}\log f_i(x;\theta_{i,0})\|^2 \leq c_1(1+|x-\mu_i|^2) \leq c_1'(1+|x|^2)    
\end{equation}
for finite constants $c_1,c_1'$ depending only on $\Theta_{\theta_i}$, establishing $c_2=2$. The required moment condition $\int|x|^2 f_i(x;\theta_{i,0})\,dx = \mathbb{E}[\varepsilon_{i,t}^2] < \infty$
holds by ${Var}(\varepsilon_{i,t}) < \infty$. \\

\paragraph{Condition (MBB)}
The block length $\ell=l_T$ in Algorithm~\ref{subsec:RMBB} satisfies $\ell\to\infty$ and $\ell/T\to 0$ as $T\to\infty$ (e.g.\ $\ell=\lfloor 5.03\,T^{1/4}\rfloor$), so that the number of blocks $L=\lfloor T/\ell\rfloor\to\infty$, and the block start indices $i_1,\dots,i_L$ are drawn independently and uniformly from $\{0,1,\dots,T-\ell\}$, independently of the data.

\paragraph{Condition (BV) [Bootstrap validity of the two-step estimator]}
The residual-based MBB of Algorithm~\ref{subsec:RMBB}, composed with
the NGML second step, is first-order valid: conditional on the data and under
$\mathcal{H}_0$, the bootstrap estimator $\hat{\lambda}^{\ast}_T$ admits the same
first-order expansion as $\hat{\lambda}_T$, with the bootstrap score
$\hat{S}^{\ast}_T$ and Hessian $\nabla_\lambda\hat{S}^{\ast}_T$ following, respectively,
the moving block bootstrap central limit theorem of \citet{goncalves_maximum_2004}. This two-step composition is maintained; the $\beta$-mixing and $(2+\eta)$-moment conditions it
requires are verified in Lemma~\ref{lemma:boot_score}.

\begin{lemma}[Bootstrap score convergence]\label{lemma:boot_score}
   Under Assumptions~\ref{as:stability}--\ref{as:gau}, the regularity conditions of Appendix~\ref{subsec:regcond}, Conditions (MBB) and (BV), and Lemma~\ref{lemma:scoreconditions}:
    \[
     T^{1/2}\big(\hat{S}^{\ast}_T(\hat{\lambda}_T)-S_T(\hat{\lambda}_T)\big) \  \xrightarrow{d^*}_p \ \mathcal{N}(0, \mathcal{I}(\hat{\lambda}))
    \]
where $\mathcal{I}(\hat{\lambda}_T) = -\mathbb{E}[\nabla^2 \ell_{\lambda\lambda,t}(\hat{\lambda}_T)]$ and $\nabla^2 \ell_{\lambda\lambda,t}(\hat{\lambda}_T) = {\partial^2 \ell_t(\hat{\lambda}_T)}/{\partial \lambda \partial \lambda'}$.
\end{lemma}

 \begin{proof} By Lemma~\ref{lemma:scoreconditions}, at $\lambda_0$ the score $\nabla_\lambda\ell_t$ is a strictly stationary, ergodic martingale-difference sequence with $\mathbb{E}\|\nabla_\lambda\ell_t(\lambda_0)\|^{2+\eta}<\infty$ for some $\eta>0$ (the NIG distribution has finite moments of all orders, \citet{barndorff-nielsen_normal_1997}) and nonsingular variance $\mathcal{I}(\lambda_0)$ by condition~(7) of Appendix~\ref{subsec:regcond}. Under Assumption~\ref{as:stability} the companion process is stable and, the NIG density being strictly positive on $\mathbb{R}$, geometrically ergodic; hence $\nabla_\lambda\ell_t$ is geometrically $\beta$-mixing, as a measurable function of the geometrically ergodic companion process. These are the conditions under which the {residual-based}
moving block bootstrap of \citet{bruggemann_inference_2016} is first-order valid for the reduced-form VAR. The accompanying central limit
theorem for the smooth extremum estimator follows from the moving block bootstrap
theory for extremum estimators of \citet{goncalves_maximum_2004}, whose
data-blocking arguments we invoke for the residual-recursive construction. By
Condition~(BV) this two-step composition is first-order valid, and the $\beta$-mixing
and $(2+\eta)$-moment conditions it requires hold by Lemma~\ref{lemma:scoreconditions}. Consequently, evaluated at the bootstrap recentering point $\hat{\lambda}_T$ and under Condition (MBB),
\[
T^{1/2}\big(\hat{S}^{\ast}_T(\hat{\lambda}_T)-S_T(\hat{\lambda}_T)\big) \ \xrightarrow{d^*}_p \ \mathcal{N}\big(0,\mathcal{I}(\hat{\lambda})\big),
\]
the limiting covariance being the long-run variance of the score. Since the
score is a martingale difference (Lemma~\ref{lemma:scoreconditions}), it is
serially uncorrelated, so this long-run variance carries no autocovariance terms
and equals the contemporaneous score covariance
$\mathcal{J}(\lambda_0)=\mathbb{E}[\nabla_\lambda\ell_t\nabla_\lambda\ell_t']$.
Under valid identification with all structural shocks non-Gaussian,
the information identity
$\mathcal{J}(\lambda_0)=-\mathbb{E}[\nabla^2_{\lambda\lambda}\ell_t]=\mathcal{I}(\lambda_0)$
holds, giving the stated $\mathcal{I}(\hat{\lambda}_T)$; when one shock is
Gaussian the fitted density is misspecified on the boundary, $\mathcal{J}\neq\mathcal{I}$,
and the robust sandwich of Proposition~\ref{prop:onegaussian} applies to the autoregressive and impact parameters $\psi$.
\end{proof}

\subsubsection{Proof of Proposition~\ref{prop:BSvalidity}}
A Taylor expansion of the bootstrap first-order condition $\hat{S}^{\ast}_T(\hat{\lambda}^{\ast}_T)=0$ about $\hat{\lambda}_T$, using $S_T(\hat{\lambda}_T)=0$, gives
\[
T^{1/2}\big(\hat{\lambda}^{\ast}_T-\hat{\lambda}_T\big) = -\big[\nabla_\lambda\hat{S}^{\ast}_T(\tilde{\lambda}_T)\big]^{-1} T^{1/2}\big(\hat{S}^{\ast}_T(\hat{\lambda}_T)-S_T(\hat{\lambda}_T)\big) + o^{\ast}_p(1),
\]
where $\tilde{\lambda}_T$ lies between $\hat{\lambda}^{\ast}_T$ and $\hat{\lambda}_T$; the bootstrap Hessian satisfies $\nabla_\lambda\hat{S}^{\ast}_T(\tilde{\lambda}_T)\xrightarrow{p^*}_p -\mathcal{I}(\lambda_0)$ by a uniform law of large numbers under Condition~(MBB) and the moment bounds of Lemma~\ref{lemma:scoreconditions} \citep{goncalves_maximum_2004}, hence the limit is nonsingular. By the consistency of $\hat{\lambda}_T$ through \citet{lanne_identification_2017}(Theorem~1) and of $\hat{\lambda}^{\ast}_T$ through \citet{bruggemann_inference_2016}(Theorem~4.1), combining with Lemma~\ref{lemma:boot_score} yields
\[
T^{1/2}\big(\hat{\lambda}^{\ast}_T-\hat{\lambda}_T\big) \ \xrightarrow{d^*}_p \ \mathcal{N}\big(0,\mathcal{I}(\hat{\lambda}_T)^{-1}\big),
\]
and pre-multiplying by $\hat{\Sigma}^{-1/2}_{\lambda_T}$, a consistent estimator of the asymptotic covariance of $\hat{\lambda}_T$, gives the studentized statement of Proposition~\ref{prop:BSvalidity}.

The argument presumes the regularity conditions hold for the full
    $\lambda$, i.e. all structural shocks are non-Gaussian so that the score
    covariance equals the Hessian information,
    $\mathcal{J}(\lambda_0)=\mathbb{E}[\nabla_\lambda\ell_t\nabla_\lambda\ell_t']
    =\mathcal{I}(\lambda_0)$; when one shock is Gaussian the same steps apply
    to the autoregressive and impact parameters $\psi=(\operatorname{vec}(\Pi)',\beta',\sigma')'$ and yield the robust (sandwich)
    covariance of Corollary~\ref{cor:Bblock} (Proposition~\ref{prop:onegaussian}).
\begin{remark}[Why a moving block bootstrap]\label{rem:whyMBB}
Because the score is a martingale difference (Lemma~\ref{lemma:scoreconditions}),
it is serially uncorrelated, and its long-run variance is the contemporaneous
$\mathcal{J}(\lambda_0)$; a block bootstrap is therefore not \emph{necessary} to
reproduce serial dependence of the score, and under the i.i.d.\ design the
residual-based moving block bootstrap and the i.i.d.\ residual bootstrap share the
same first-order limit. We nonetheless use the residual-based moving block
bootstrap throughout as the robust default: it remains first-order valid for the
reduced-form VAR under conditional heteroscedasticity of unknown form
\citep{bruggemann_inference_2016}, where the i.i.d.\ residual bootstrap would be
inconsistent. The NGML score conditions (Lemma~\ref{lemma:scoreconditions}) are verified
here under the i.i.d.\ design; their extension to conditionally heteroscedastic shocks in the sense of \citet{lanne_identification_2017} is maintained. Moreover, the
studentization of the impact-matrix estimator uses the robust covariance estimator $\hat{\Sigma}_{B_T}$, which
reduces to the efficient \eqref{eq:SigmaB} under correct specification.
\end{remark}

\subsection{Proof of Proposition~\ref{prop:bootstraptest}}\label{subsec:proofbootstraptest}
\subsubsection{Validity under $\mathcal{H}_0$}
By Proposition~\ref{prop:BSvalidity}, under $\mathcal{H}_0$ and the regularity conditions, the conditional distribution of
$Q^{*}_{T}$ given $D_T$ converges in probability to the standard
normal:
\begin{equation}
  \label{eq:bsvalid}
  \sup_{x \in \mathbb{R}^{n^2}}\,
  \bigl|\hat{G}^{*}_{T}(x) - \Phi(x)\bigr|
  \;\xrightarrow{p^{*}}_{p}\; 0,
  \qquad \text{as } T \to \infty.
\end{equation}
This follows from the asymptotic normality of the NGML estimator
(equations \eqref{eq:asymp}--\eqref{eq:fisher}), the consistency of
$\hat{\Sigma}_{B_T}$, and the asymptotic validity of the residual-based
MBB established in
\citet{bruggemann_inference_2016}, which ensures that conditional
heteroscedasticity of unknown form does not invalidate the bootstrap
approximation.

\subsubsection{CLT for the empirical bootstrap distribution}
For fixed $T$ and any $x \in \mathbb{R}^{n^{2}}$, the bootstrap
replications $\{\mathbf{1}(\mathrm{vec}(Q^{*}_{T,b}) \leq x)\}_{b=1}^{M}$
are conditionally i.i.d.\ given the data $D_T$, each with conditional mean
$\hat{G}^{*}_{T}(x)$ and conditional variance
$\Omega_{T}(x) = \hat{G}^{*}_{T}(x)(1 - \hat{G}^{*}_{T}(x))$.
By the conditional Lindeberg--Lévy central limit theorem, as
$M \to \infty$ for fixed $T$,
\begin{equation}
  \label{eq:clt}
  M^{1/2}\,\hat{\Omega}^{-1/2}_{T}(x)
  \Bigl(\hat{G}^{*}_{T,M}(x) - \hat{G}^{*}_{T}(x)\Bigr)
  \;\xrightarrow{d}\;
  \mathcal{N}(0,1).
\end{equation}

\begin{condition}[Convergence rate of the conditional bootstrap distribution]\label{cond:BErate}
The consistent bootstrap used (the i.i.d.\ residual bootstrap, or the moving block
bootstrap of Condition~(MBB) under serial dependence) is such that, under
$\mathcal{H}_0$ and the regularity conditions of Appendix~\ref{subsec:regcond}, the
conditional distribution of $Q^{\ast}_T=\hat{\Sigma}^{-1/2}_{B_T}T^{1/2}(\hat{B}^{\ast}_T-\hat{B}_T)$
satisfies
\[
  \sup_{x\in\mathbb{R}^{n^2}}\bigl|\hat{G}^{\ast}_T(x)-\Phi(x)\bigr|=O_p\!\bigl(T^{-\rho}\bigr),
  \qquad \text{for some } \rho>0 .
\]
\end{condition}

\begin{remark} $Q^{\ast}_T$ is a twice continuously differentiable function of
sample averages of the NGML score (Lemma~\ref{lemma:scoreconditions}), studentized
by the nonsingular robust covariance of Proposition~\ref{prop:onegaussian}. The
NIG density is analytic, so the Cramér condition holds and the studentized bootstrap statistic admits a multivariate Edgeworth expansion uniform in
$x\in\mathbb{R}^{n^2}$ \citep{bhattacharya_normal_1976,hall_bootstrap_1992}. The
leading departure of $\hat{G}^{\ast}_T$ from $\Phi$ is the $O(T^{-1/2})$ term
governed by the standardized third cumulant of the statistic $Q^{\ast}_T$, giving $\rho=\tfrac12$ which we maintain throughout;
for the moving block scheme the same second-order rate follows from
\citet{gotze_second-order_1996}. At the single-Gaussian boundary this rate is preserved. The boundary shape nuisance
$\rho_k$ is information-orthogonal to $\psi$ and parity-cancelling
(Lemma~\ref{lem:parity}), so it does not enter $\hat{B}_T$ and
hence does not affect the $O(T^{-1/2})$ Edgeworth term of $Q^{\ast}_T$. Thus,
$\rho=\tfrac12$ holds across the entire null, including the boundary.
The same expansion controls the standardized cumulants of $\hat{G}^{\ast}_T$: its
standardized skewness and excess kurtosis satisfy $\gamma_{1,T}=O_p(T^{-1/2})$ and
$\gamma_{2,T}=O_p(T^{-1})$, both $O_p(T^{-\rho})$ with $\rho=\tfrac12$. These cumulant
rates are used in Corollary~\ref{cor:sizetest}.
\end{remark}

Now, decompose the test statistic as
\begin{align}
  d^{*}_{T,M}(x)
  &= M^{1/2}\,\hat{\Omega}^{-1/2}_{T}(x)
     \Bigl(\hat{G}^{*}_{T,M}(x) - \hat{G}^{*}_{T}(x)\Bigr)
  \notag \\
  &\quad\;+\;
     M^{1/2}\,\hat{\Omega}^{-1/2}_{T}(x)
     \Bigl(\hat{G}^{*}_{T}(x) - \Phi(x)\Bigr).
  \label{eq:decomp}
\end{align}
The first term on the right-hand side of \eqref{eq:decomp} converges in
distribution to $\mathcal{N}(0,1)$ by \eqref{eq:clt}. The second term
represents a potential source of pre-testing bias: it is a
data-dependent centering error of order
$M^{1/2}\sup_{x}|\hat{G}^{*}_{T}(x) - \Phi(x)|$.
Under $\mathcal{H}_0$, Condition~\ref{cond:BErate} gives
     $\sup_{x}|\hat{G}^{*}_{T}(x) - \Phi(x)| = O_p(T^{-\rho})$ with $\rho=\tfrac12$.
     The condition $MT^{-2\rho} \to 0$
therefore ensures that
\begin{equation}
  M^{1/2}\,\hat{\Omega}^{-1/2}_{T}(x)
  \Bigl(\hat{G}^{*}_{T}(x) - \Phi(x)\Bigr)
  \;=\;
  O_p\!\bigl(M^{1/2}T^{-\rho}\bigr)
  \;=\;
  o_p(1),
\end{equation}
as $T, M \to \infty$ jointly. Consequently, the second term in
\eqref{eq:decomp} is asymptotically negligible, eliminating pre-testing
bias. The asymptotic null distribution \eqref{eq:testdist} then follows
from Slutsky's theorem and the distributional result of CLT.

\begin{lemma}[Characterization under the alternative hypothesis]\label{lemma:altlaw}
Under $\mathcal{H}_1$ with $k\ge 2$ Gaussian shocks, the regularity conditions of Appendix~\ref{subsec:regcond} holding for the non-Gaussian block, conditional on $D_T$ and as
$T\to\infty$:
\begin{enumerate}
\item[(i)] the pseudo-log-likelihood is invariant to rotations $Q\in\mathcal{O}(k)$ of
the Gaussian subspace, so $\mathcal{I}(\lambda_0)$ is singular along $\mathcal{O}(k)$;
the identified directions of $\hat{B}^{\ast}_T$ concentrate at rate $T^{1/2}$, the
rotation $Q$ does not.
\item[(ii)] the affected off-diagonal replications converge to the non-degenerate distribution
induced by $Q\in\mathcal{O}(k)$ through the unit-diagonal impact map, a non-affine,
ratio/tangent-type transform, with non-zero (divergent for $n=2$) standardized excess
kurtosis; hence they do not converge to $\mathcal{N}(0,I_{n^2})$.
\item[(iii)] any consistent normality test therefore diverges (Jarque--Bera on an
affected element $\xrightarrow{p}\infty$; Doornik--Hansen inherits the non-normality),
so the diagnostic has power tending to one against every fixed alternative in
$\mathcal{H}_1$.
\end{enumerate}
\end{lemma}

\begin{proof}
\emph{(i)} With $k\ge2$ Gaussian shocks the Gaussian sub-block of the density is
rotation-invariant \citet{comon_independent_1994}, so the NGML
criterion is flat along $\mathcal{O}(k)$ and that information block is singular
\citep{gourieroux_statistical_2017,maxand_identification_2020}. The bootstrap residuals
are asymptotically Gaussian in the subspace, so the bootstrap criterion is
asymptotically flat there and $Q$ does not concentrate while the identified directions
concentrate at $T^{1/2}$. 

\emph{(ii)} The replications thus converge conditionally to $B_0$ acted on by a
non-degenerate $Q\in\mathcal{O}(k)$. Under the unit-diagonal normalization the
coefficients coupling two Gaussian directions are non-affine (ratio/tangent-type)
functions of $Q$, each with non-zero standardized excess kurtosis, which are divergent in the
$n=2$ case (Remark~\ref{rem:n2alt}). For general $n$ with $k\ge2$, non-normality holds
for each pair of Gaussian directions, so it is detected element-wise by the univariate
test of part~(iii); the fully general multivariate asymptotic distribution of the replications is
maintained.

\emph{(iii)} Jarque--Bera is consistent against any distribution with non-zero
standardized skewness or excess kurtosis, so on an affected element it diverges, giving
power $\to1$; Doornik--Hansen inherits the non-normality.
\end{proof}

\begin{remark}[The $n=2$ case]\label{rem:n2alt}
For $n=2$ with $k=2$ Gaussian shocks, the rotation is a single angle $\varphi$, asymptotically
non-degenerate on the circle. Under the unit-diagonal normalization the off-diagonal
replication $\hat{b}^{\ast}_{12}$ is a tangent-type transform of $\varphi$, with
Cauchy-type tails where the normalizing diagonal approaches zero; its sample skewness
and kurtosis diverge, so Jarque--Bera on the replications diverges. This is the
simplest configuration in which the non-normality is explicit.
\end{remark}

\subsubsection{Divergence under $\mathcal{H}_1$}
Under $\mathcal{H}_1$, Lemma~\ref{lemma:altlaw} shows that the conditional distribution
of the replications $\{\hat{B}^{\ast}_{T,b}\}$ converges to the non-normal distribution induced
by an unidentified rotation of the Gaussian subspace, so any consistent normality test
(Remark~\ref{rem:DH}) diverges and the diagnostic has power tending to one. This is the
mechanism by which the bootstrap diagnostic detects identification failure in
the framework of \citet{cavaliere_bootstrap_2025}; see also
\citet{angelini_bootstrap_2022,angelini_identification_2024}.

Although Corollary~\ref{cor:sizetest} is stated for the studentized $Q^{\ast}_T$, the
Doornik--Hansen (affine-invariant) and Jarque--Bera (location and scale-invariant)
statistics take the same value on the raw replications $\{\hat{B}^{\ast}_{T,b}\}$,
which differ from $\{Q^{\ast}_{T,b}\}$ only by the fixed affine map
$\hat{\Sigma}_{B_T}^{-1/2}T^{1/2}(\cdot-\hat{B}_T)$, so the size guarantee transfers to
the implemented test.

The joint-limit result \eqref{eq:testdist} of part~(i) and the divergence
\eqref{eq:altdist} of part~(ii) complete the proof.

\subsection{Proof of Corollary~\ref{cor:sizetest}}\label{subsec:proofcorroboottest}
\begin{proof}
The Jarque--Bera statistic is $\mathcal{T}_{T,M}=M(g_1^2/6+g_2^2/24)$, with $g_1,g_2$
the sample skewness and excess kurtosis of the $M$ studentized replications;
Doornik--Hansen is the analogous quadratic in the multivariate sample skewness and
kurtosis. Conditionally on the data the replications are i.i.d.\ from
$\hat{G}^{\ast}_T$, whose standardized skewness $\gamma_{1,T}$ and excess kurtosis
$\gamma_{2,T}$ are $O_p(T^{-\rho})$ under $\mathcal{H}_0$ by the Edgeworth expansion of
the remark following Condition~\ref{cond:BErate}.
By the central limit theorem for sample moments,
$M^{1/2}(g_1-\gamma_{1,T})\xrightarrow{d^\ast}_p\mathcal{N}(0,6)$ and
$M^{1/2}(g_2-\gamma_{2,T})\xrightarrow{d^\ast}_p\mathcal{N}(0,24)$, so $\mathcal{T}_{T,M}$
is asymptotically non-central $\chi^2_2$ with non-centrality
\[
  \lambda_{T,M}=M\Bigl(\tfrac{\gamma_{1,T}^2}{6}+\tfrac{\gamma_{2,T}^2}{24}\Bigr)
  +O_p\!\bigl((M\gamma_{i,T}^2)^{1/2}\bigr)=MT^{-2\rho}\,O_p(1).
\]
Under $MT^{-2\rho} \to 0$, $\lambda_{T,M}\xrightarrow{p}0$, so
$\mathcal{T}_{T,M}\xrightarrow{d}\chi^2_2$ and
$\Pr(\mathcal{T}_{T,M}>c_{1-\alpha}\mid\mathcal{H}_0)\to\alpha$. Under $\mathcal{H}_1$,
$\gamma_{1,T}$ or $\gamma_{2,T}$ is bounded away from zero (Lemma~\ref{lemma:altlaw}),
so $\lambda_{T,M}\to\infty$ and the power tends to one.
\end{proof}

\section{Other results} \label{subsec:resultsICA0}
\subsection{Probability coverages for specification \textbf{ICA 0}} \label{subsec:coverageica0}
\input{Tables/EstandCovICA0_v2.tex}

\subsection{Multivariate diagnostic with the Doornik--Hansen test ($T = 300$)}\label{subsec:dhmultit300}
Table~\ref{tab:mvnormtest} reports the average rejection frequencies of the null hypothesis of multivariate normality at the nominal $5\%$ and $1\%$ levels, with $T = 300, 500$, and residual-based MBB replications $N = 1999$. The results show that under the null of valid identification, in \textbf{ICA 0} and \textbf{ICA 1}, i.e., with zero and one Gaussian shock, respectively, the average rejection frequency of multivariate normality is close to the nominal levels. However, under the alternative of identification failure in \textbf{No ICA}, i.e., with two Gaussian shocks, the empirical rejection frequency of the null hypothesis of multivariate normality
increases with the sample size. At these sample sizes the power is moderate, because the rate condition
$M/T \to 0$ keeps $M$ small. Power increases with $T$, as
Proposition~\ref{prop:bootstraptest}(ii) suggests, and with the constant $a$ in
the rule for $M$, slightly oversized: at $T = 500$, raising $a$ from $1/3$
to $1/2$ increases the rejection frequency under \textbf{No ICA} from $0.170$
to $0.271$, and under \textbf{ICA 1} from $0.054$ to $0.064$.
\input{Tables/mvnormtest_v2.tex}
\input{Figures/fanchart300_v2.tex}
Table~\ref{tab:uninormtest} reports the corresponding rejection frequencies at the nominal $5\%$ level, averaged over the three elements of each column of $\hat{B}^*_T$. Under \textbf{ICA 1} the Doornik--Hansen diagnostic retains close to nominal size for every shock, with rejection frequencies of $0.058, 0.074$ and $0.076$ at $T = 500$, and the mild over-rejection for the second and third shocks shrinks as $T$ grows. Under \textbf{No ICA} the localization is sharp: at $T = 500$ the rejection frequency for the column of the single non-Gaussian shock is $0.070$, while the two columns associated with the Gaussian shocks have rejection frequency of $0.415$ and $0.418$.
\input{Tables/uninormtest_v2.tex}

\subsection{Univariate diagnostic with an alternative normality test}\label{subsec:JBuni}
This subsection presents the results of the univariate diagnostic based on the Jarque--Bera (JB) test for normality of each element of the bootstrap impact-matrix estimator $\hat{B}_T^*$. We use $M = \lceil (5/3)T^{3/5} \rceil$, and the results are qualitatively similar to those obtained from the Doornik--Hansen test for normality presented in Section~\ref{sec:MC}, Figures~\ref{fig:uni_NIGOGT500} and \ref{fig:uni_NIGWGT500}. 
\input{Figures/jbuni_NIGOGT500_v2.tex}
\input{Figures/jbuni_NIGWGT500_v2.tex}

\subsection{Other results from the empirical application}\label{app:other_results}
Table~\ref{tab:emp_lagsel} reports the information criteria for the lag order of the reduced-form VAR in Section~\ref{sec:EMP}.
\input{Tables/emp_lagselect.tex}
Figure~\ref{fig:emp_irf_50} reports the impulse responses with block length $\ell = 50$. The estimated responses are essentially unchanged.
\input{Figures/emp_irf_50.tex}
Table~\ref{tab:emp_skew} compares the impact matrix estimated with symmetric NIG shocks to the one estimated with the asymmetry parameters left unrestricted.
\input{Tables/emp_skew_v2.tex}

%% file: Tables/EstandCovICA0_v2.tex
\begin{table}[H]
\centering

    \renewcommand{\arraystretch}{0.9}

\scalebox{0.65}{

         \begin{tabular}{|c|c|cc|c|ccc|}
         \hline
         & & & & & & \multicolumn{2}{c|}{\textbf{Bootstrap Coverage}}\\

         & $\mathbf{B_{0}}$ & $\mathbf{\hat{B}_T}$ & $\mathbf{\hat{B}^*_T}$ & $\mathbf{90\%}$\textbf{CI} & \textbf{Asymp. Cov}. & \textbf{Studentized} & \textbf{Percentile} \\
         & (a) & (b) & (c) & (d) & (e) & (f) & (g) \\

         \hline
        \rule{0pt}{1ex}
                           & 1.29 & 1.26 & 1.23 & [0.94, 1.48] & 0.83 & 0.82 & 0.74 \\
                           && (0.19) & (0.19) &  & (0.60) & (0.63) & (0.53)\\
                           & 0.82 & 0.79 & 0.78 & [0.46, 1.08] & 0.88 & 0.81 & 0.85 \\
                           && (0.18) & (0.19) &  & (0.58) & (0.60) & (0.61)\\
                           & -0.66 & -0.64 & -0.62 & [-0.92, -0.31] & 0.88 & 0.81 & 0.86 \\
                           && (0.18) & (0.19) &  & (0.58) & (0.57) & (0.61)\\
                           & 0.28 & 0.27 & 0.23 & [0.10, 0.38] & 0.91 & 0.85 & 0.91 \\
                           && (0.08) & (0.09) &  & (0.25) & (0.25) & (0.29)\\
        $\mathbf{T = 150}$ & 1.40 & 1.36 & 1.31 & [1.00, 1.61] & 0.90 & 0.82 & 0.83 \\
                           && (0.19) & (0.20) &  & (0.61) & (0.60) & (0.62)\\
                           & 0.77 & 0.76 & 0.37 & [0.33, 0.99] & 0.88 & 0.83 & 0.88 \\
                           && (0.19) & (0.19) &  & (0.58) & (0.57) & (0.62)\\
                           & 0.00 & 0.00 & 0.05 & [-0.13, 0.20] & 0.88 & 0.85 & 0.91 \\
                           && (0.08) & (0.09) &  & (0.26) & (0.27) & (0.31)\\
                           & -0.88 & -0.86 & -0.47 & [-1.05, -0.44] & 0.85 & 0.85 & 0.89 \\
                           && (0.17) & (0.17) &  & (0.50) & (0.52) & (0.59)\\
                           & 1.47 & 1.42 & 1.29 & [1.08, 1.57] & 0.87 & 0.83 & 0.79 \\
                           && (0.15) & (0.16) &  & (0.47) & (0.47) & (0.50)\\
\hline
        \rule{0pt}{1ex}
                           & 1.29 & 1.28 & 1.28 & [1.11, 1.43] & 0.88 & 0.87 & 0.85 \\
                           && (0.11) & (0.11) &  & (0.36) & (0.36) & (0.33)\\
                           & 0.82 & 0.81 & 0.80 & [0.65, 0.96] & 0.89 & 0.85 & 0.86 \\
                           && (0.10) & (0.10) &  & (0.32) & (0.31) & (0.31)\\
                           & -0.66 & -0.66 & -0.65 & [-0.80, -0.49] & 0.91 & 0.88 & 0.88 \\
                           && (0.10) & (0.10) &  & (0.31) & (0.30) & (0.31)\\
                           & 0.28 & 0.28 & 0.27 & [0.21, 0.34] & 0.91 & 0.88 & 0.90 \\
                           && (0.04) & (0.04) &  & (0.13) & (0.13) & (0.13)\\
        $\mathbf{T = 500}$ & 1.40 & 1.39 & 1.38 & [1.22, 1.55] & 0.91 & 0.87 & 0.88 \\
                           && (0.10) & (0.10) &  & (0.33) & (0.32) & (0.33)\\
                           & 0.77 & 0.76 & 0.68 & [0.58, 0.91] & 0.88 & 0.86 & 0.87 \\
                           && (0.10) & (0.10) &  & (0.32) & (0.31) & (0.32)\\
                           & 0.00 & 0.00 & 0.01 & [-0.07, 0.07] & 0.90 & 0.90 & 0.90 \\
                           && (0.04) & (0.04) &  & (0.13) & (0.13) & (0.14)\\
                           & -0.88 & -0.87 & -0.79 & [-1.00, -0.72] & 0.90 & 0.89 & 0.89 \\
                           && (0.08) & (0.09) &  & (0.27) & (0.27) & (0.28)\\
                           & 1.47 & 1.46 & 1.43 & [1.32, 1.57] & 0.91 & 0.87 & 0.88 \\
                           && (0.08) & (0.08) &  & (0.26) & (0.25) & (0.25)\\
\hline
\hline
    \end{tabular}
    }
    \caption{\scriptsize Estimates of the true impact matrix, $B_{0}$, from the specification: \\
    \textbf{ICA 0:} $\epsilon_{1,t} \sim NIG(0.34,0.57), \epsilon_{2,t} \sim NIG(0.51,1.00)$, and \\
     $\epsilon_{3,t} \sim NIG(0.83,1.79)$ \\
     (see Section~\ref{sec:dist_spec}),
    with ${T} = 150, 500$, Monte Carlo simulations ${N_S} = 1000$ and residual-based moving-block bootstrap replications ${N} = 1999$.\\
    (a) True parameter values, $\mathrm{vec}(B_0)$.\\
    (b) NGML estimates, averaged across Monte Carlo simulations; the parentheses below contain the average estimated standard error.\\
    (c) Bootstrap estimates, averaged across Monte Carlo simulations; the parentheses below contain the average bootstrap standard error.\\
    (d) Median lower and upper limits of the bootstrap $90\%$ percentile confidence intervals, across Monte Carlo simulations.\\
    (e) Empirical probability coverage of the asymptotic $90\%$ confidence intervals.\\
    (f) Bootstrap (studentized) probability coverage of nominal $90\%$ confidence intervals.\\
    (g) Bootstrap (percentile) probability coverage of nominal $90\%$ confidence intervals.\\
    For (e), (f) and (g) the parentheses below contain the median length of the respective confidence interval, across Monte Carlo simulations.
    Coverage is reported after resolving the column labeling of $\hat{B}_T$: each replication is matched to $B_0$ by its best signed column permutation, measured by Amari distance (for $T=150$, 17.3\% of replications were realigned; for $T = 500$, 3.5\% of replications were realigned). The realignment uses $B_0$ and is therefore a Monte Carlo convention rather than a procedure available to a practitioner.}
    \label{tab:EstandCovICA0}
\end{table}

%% file: Tables/mvnormtest_v2.tex
\begin{table}[H]
\centering
\renewcommand{\arraystretch}{1}
\textit{Panel A: nominal 5\% level}

\scalebox{0.73}{
    \begin{tabular}{|c|cc|cc|}
     \hline
        & \multicolumn{2}{c|}{$\mathbf{T = 300}$} & \multicolumn{2}{c|}{$\mathbf{T = 500}$} \\
        \textbf{DGP} & $\mathbf{M = \lceil (1/3)T^{3/5} \rceil}$ & $\mathbf{M = \lceil (1/2)T^{3/5} \rceil}$ & $\mathbf{M = \lceil (1/3)T^{3/5} \rceil}$ & $\mathbf{M = \lceil (1/2)T^{3/5} \rceil}$ \\
        \hline
        \rule{0pt}{2.5ex}
        \textbf{ICA 0} & 0.058 & 0.069 & 0.057 & 0.070 \\
        \textbf{ICA 1} & 0.056 & 0.064 & 0.054 & 0.064 \\
        \textbf{No ICA} & 0.154 & 0.218 & 0.170 & 0.271 \\
        \hline
    \end{tabular}
}

\vspace{0.4cm}

\textit{Panel B: nominal 1\% level}

\scalebox{0.73}{
    \begin{tabular}{|c|cc|cc|}
     \hline
        & \multicolumn{2}{c|}{$\mathbf{T = 300}$} & \multicolumn{2}{c|}{$\mathbf{T = 500}$} \\
        \textbf{DGP} & $\mathbf{M = \lceil (1/3)T^{3/5} \rceil}$ & $\mathbf{M = \lceil (1/2)T^{3/5} \rceil}$ & $\mathbf{M = \lceil (1/3)T^{3/5} \rceil}$ & $\mathbf{M = \lceil (1/2)T^{3/5} \rceil}$ \\
        \hline
        \rule{0pt}{2.5ex}
        \textbf{ICA 0} & 0.014 & 0.018 & 0.013 & 0.019 \\
        \textbf{ICA 1} & 0.013 & 0.016 & 0.013 & 0.016 \\
        \textbf{No ICA} & 0.105 & 0.149 & 0.111 & 0.187 \\
        \hline
    \end{tabular}
}
\caption{\scriptsize Empirical rejection frequencies of the multivariate Doornik--Hansen test for normality of $S = 1000$ bootstrap estimate sequences, ${\{\hat{B}_{T,1}^*, \hat{B}_{T,2}^*, \cdots, \hat{B}_{T,M}^*\}}_{s}, s = 1, 2, \cdots, S$, for different sample sizes, ${T} = 300, 500$, across ${N_S} = 1000$ Monte Carlo simulations, at the nominal 5\% and 1\% levels, respectively, for the three distributional specifications (DGPs) described in Section~\ref{sec:dist_spec}, and for the values of $M$ as discussed in Remark~\ref{rem:Mchoice}. Each value is the average over Monte Carlo simulations of the within-simulation rejection frequency. The two columns of each block are computed from the same bootstrap draws and differ only in $M$. A test statistic that cannot be computed is counted as a rejection; this affects at most 0.6\% of tests, and only under \textbf{No ICA}.}
\label{tab:mvnormtest}
\end{table}

%% file: Figures/fanchart300_v2.tex
\begin{figure}[H]
    \centering
          \begin{subfigure}{0.45\textwidth}
          \centering
          \includegraphics[width = 0.75\textwidth]{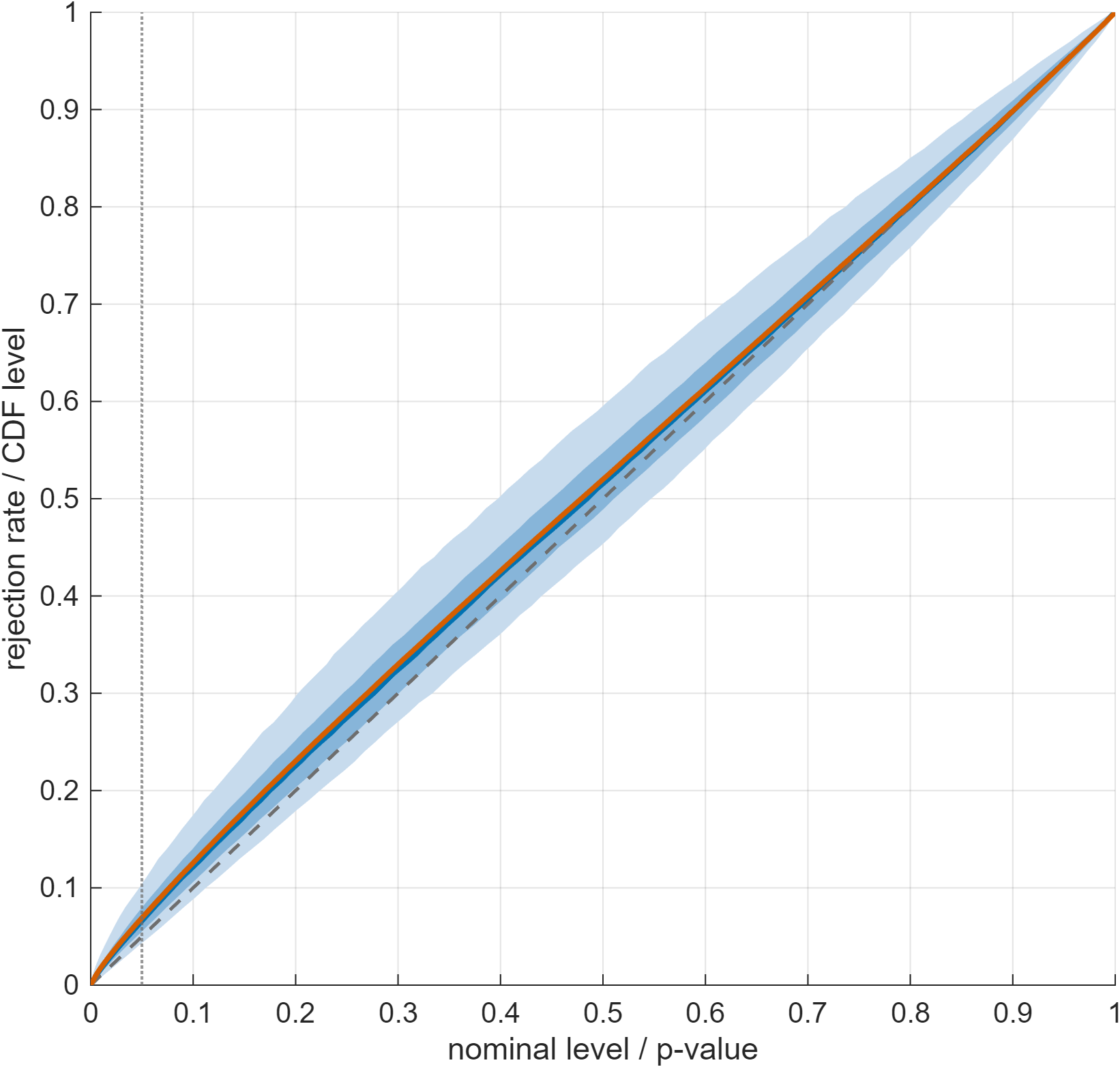}
          \subcaption{\textbf{ICA 0}}
          \label{fig:300ICA0}
          \end{subfigure}
      \quad
          \begin{subfigure}{0.45\textwidth}
          \centering
          \includegraphics[width = 0.75\textwidth]{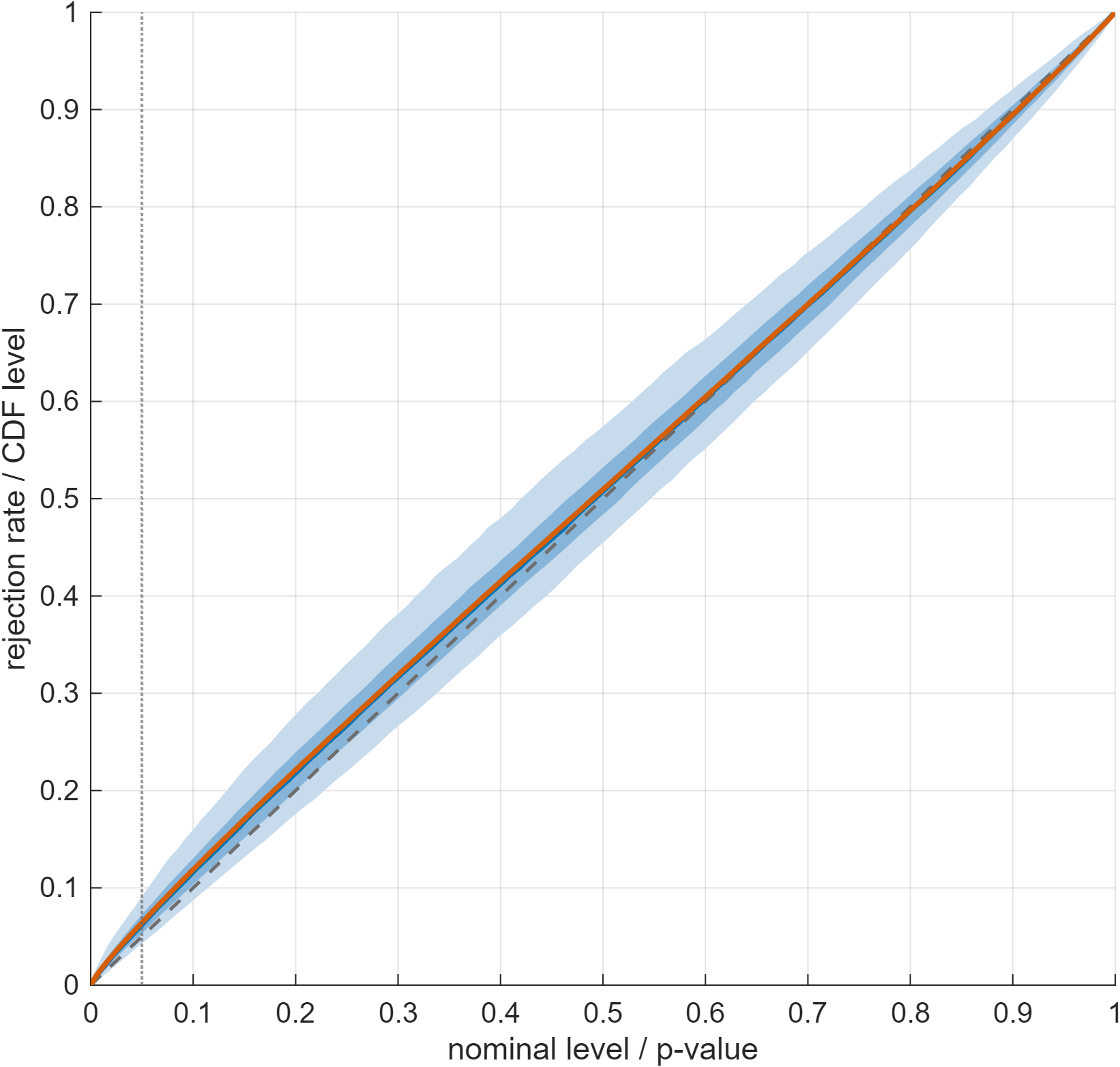}
          \subcaption{\textbf{ICA 1}}
          \label{fig:300ICA1}
          \end{subfigure}
      \quad
      % \vspace{0.2cm}
          \begin{subfigure}{0.45\textwidth}
          \centering
          \includegraphics[width = 0.75\textwidth]{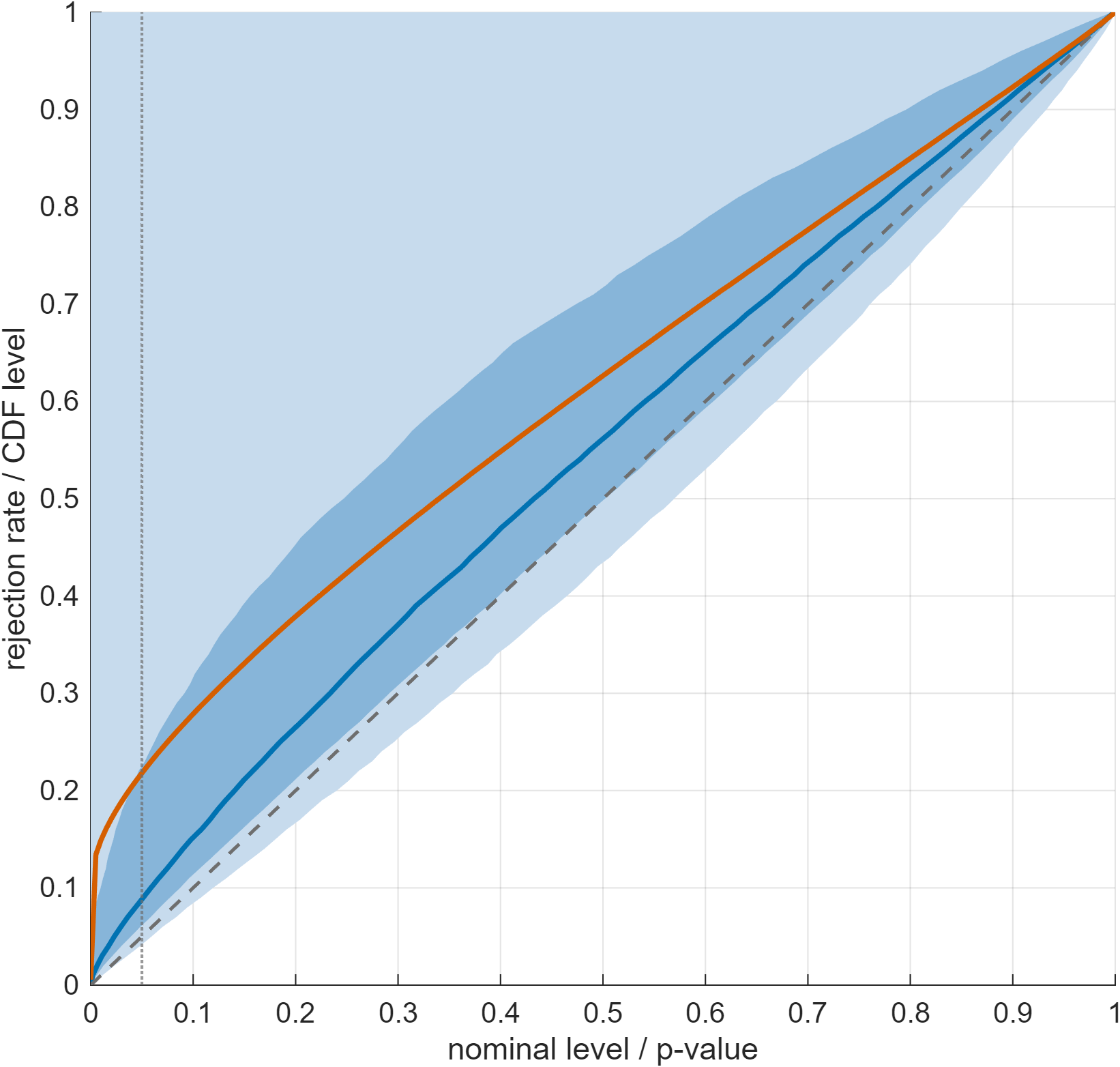}
          \subcaption{\textbf{No ICA}}
          \label{fig:300NoICA}
          \end{subfigure}
\caption{Empirical distribution function of the $p$-values, $p^*_{M,T,s}$, $s = 1,2,\cdots, \mathbf{S}$ of the multivariate Doornik--Hansen normality test for the bootstrap estimates of $\hat{B}_T^{*}$ for the three specifications - \textbf{ICA 0}, \textbf{ICA 1} and \textbf{No ICA} - with sample size $T = 300$ and $M = \lceil (1/2)T^{3/5} \rceil$, for the number of sequences, $S = 1000$, across Monte Carlo simulations, $N_S = 1000$. The solid blue (orange) line is the median (mean) of the empirical rejection frequencies, the dashed gray line is the $45^{\circ}$ line, and the dark and light shaded areas indicate the $50$ and $90$-percentile intervals, respectively, across $N_S$. The dotted vertical line marks the nominal $5\%$ level, with the $p$-values on the $x$-axis and their corresponding empirical frequencies on the $y$-axis.}
\label{fig:fanchart300}
\end{figure}

%% file: Tables/uninormtest_v2.tex
\begin{table}[H]
\centering
\renewcommand{\arraystretch}{1}
\scalebox{0.9}{
    \begin{tabular}{|c|ccc|ccc|}
     \hline
        & \multicolumn{3}{c|}{$\mathbf{T = 300}$} & \multicolumn{3}{c|}{$\mathbf{T = 500}$} \\
        \textbf{DGP} & $\varepsilon_{1}$ & $\varepsilon_{2}$ & $\varepsilon_{3}$ & $\varepsilon_{1}$ & $\varepsilon_{2}$ & $\varepsilon_{3}$ \\
        \hline
        \textit{DH} ($M = \lceil (1/2)T^{3/5} \rceil$) & & & & & & \\
        \rule{0pt}{2.5ex}
        \textbf{ICA 0} & 0.071 & 0.099 & 0.102 & 0.062 & 0.081 & 0.080 \\
        \textbf{ICA 1} & 0.066 & 0.091 & 0.097 & 0.058 & 0.074 & 0.076 \\
        \textbf{No ICA} & 0.071 & 0.348 & 0.353 & 0.070 & 0.415 & 0.418 \\
        \hline
        \textit{JB} ($M = \lceil (5/3)T^{3/5} \rceil$) & & & & & & \\
        \rule{0pt}{2.5ex}
        \textbf{ICA 0} & 0.113 & 0.189 & 0.194 & 0.087 & 0.140 & 0.137 \\
        \textbf{ICA 1} & 0.097 & 0.168 & 0.186 & 0.078 & 0.120 & 0.128 \\
        \textbf{No ICA} & 0.106 & 0.505 & 0.508 & 0.098 & 0.573 & 0.579 \\
        \hline
    \end{tabular}
}
\caption{\scriptsize Empirical rejection frequencies of the univariate Doornik--Hansen and Jarque--Bera tests for normality of the bootstrap impact-matrix estimates, $\hat{B}^*_T$, at the nominal 5\% level, averaged over the three elements of the column of $\hat{B}^*_T$ associated with each structural shock $\varepsilon_j$, $j = 1,2,3$, over $S = 1000$ sequences and across ${N_S} = 1000$ Monte Carlo simulations. $M = \lceil (1/2) T^{3/5} \rceil$ for DH and $M = \lceil (5/3) T^{3/5} \rceil$ for JB. In \textbf{No ICA} the Gaussian shocks are $\varepsilon_{2}$ and $\varepsilon_{3}$; in \textbf{ICA 1} the Gaussian shock is $\varepsilon_{3}$.}
\label{tab:uninormtest}
\end{table}

%% file: Figures/jbuni_NIGOGT500_v2.tex
\begin{figure}[H]
    \centering
    \includegraphics[width = 0.65\textwidth]{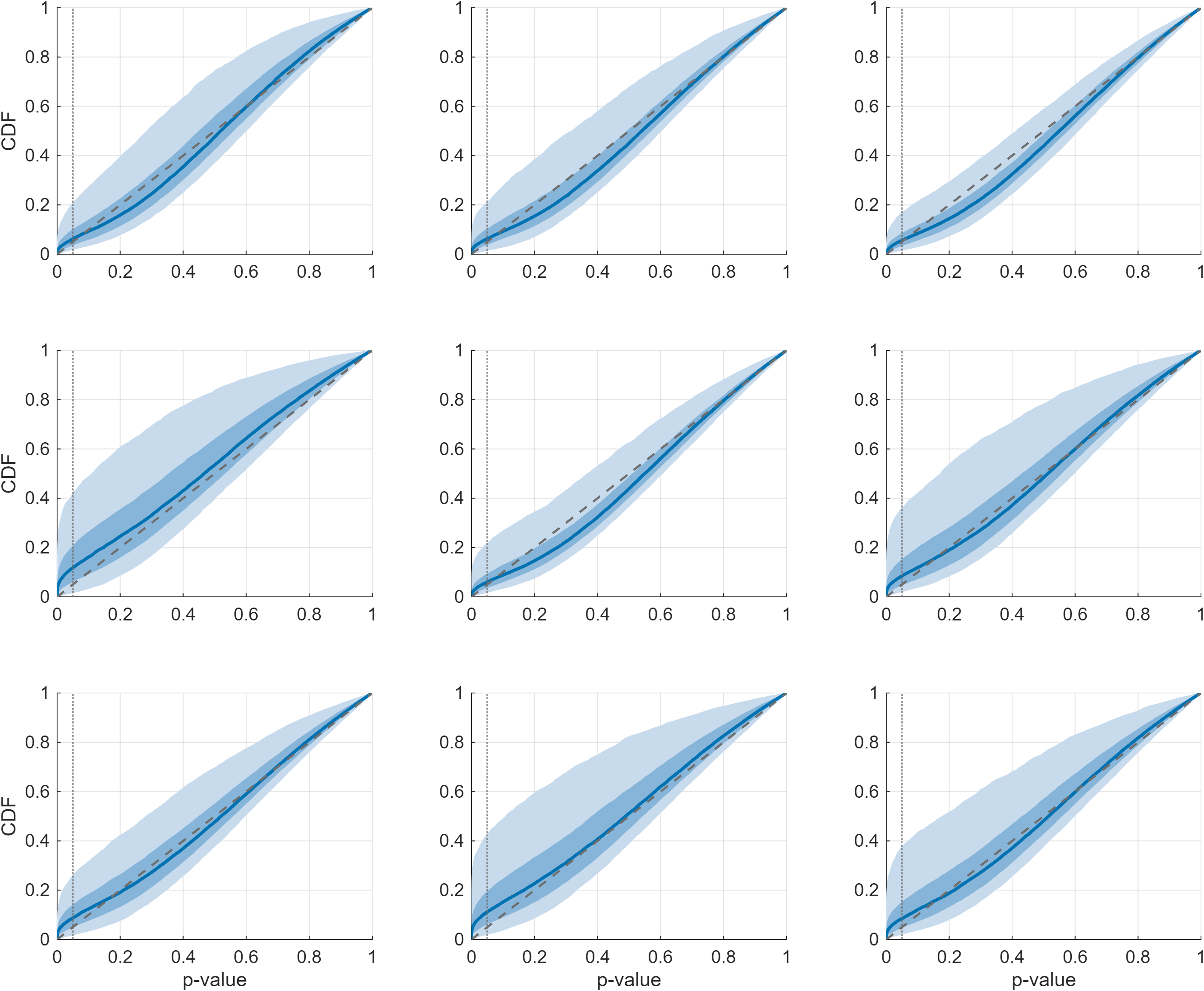}
    \caption{Univariate \textbf{ICA 1}: Empirical distribution function of $p^*_{M,T,s}$ from the univariate Jarque--Bera test for normality of each element of the bootstrap estimates of
 $\hat{B}_{T}^*$ for the specification \textbf{ICA 1} (see Section~\ref{sec:dist_spec}), with sample size ${T} = 500$, the sequence length ${M} = \lceil (5/3)T^{3/5} \rceil$,
    and the number of sequences, ${S} = 1000$, across Monte Carlo simulations ${N_S} = 1000$. The $j$-th row of panels collects the three elements of the $j$-th column of $\hat{B}_T$, that is, the on-impact responses to the structural shock $\varepsilon_{j,t}$. The dashed gray line indicates the $45^{\circ}$ line and the dotted vertical line the nominal $5\%$ level. The solid blue line is the median of the empirical distribution of the $p$-values, $\pi^*_{M,T,S}(x)$, while the dark and light shaded areas indicate the $50$ and $90$-percentile intervals, respectively, with the $p$-values on the $x$-axis and their corresponding empirical frequencies on the $y$-axis.}
    \label{fig:jbuni_NIGOGT500}
\end{figure}

%% file: Figures/jbuni_NIGWGT500_v2.tex
\begin{figure}[H]
    \centering
    \includegraphics[width = 0.65\textwidth]{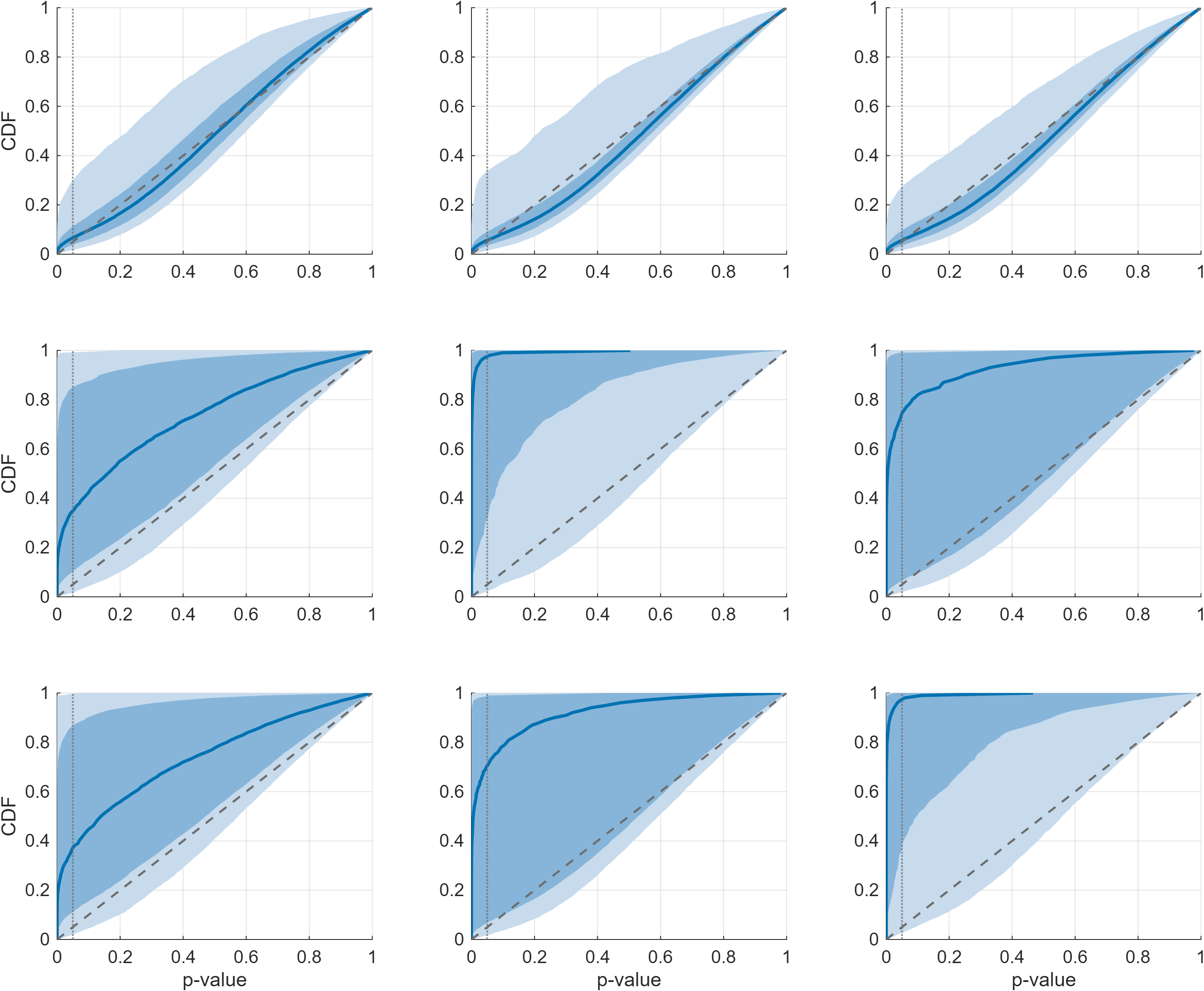}
    \caption{Univariate \textbf{No ICA}: Empirical distribution function of $p^*_{M,T,s}$ from the univariate Jarque--Bera test for normality of each element of the bootstrap estimates of
 $\hat{B}_{T}^*$ for the specification \textbf{No ICA} (see Section~\ref{sec:dist_spec}), with sample size ${T} = 500$, the sequence length ${M} = \lceil (5/3)T^{3/5} \rceil$,
    and the number of sequences, ${S} = 1000$, across Monte Carlo simulations ${N_S} = 1000$. The $j$-th row of panels collects the three elements of the $j$-th column of $\hat{B}_T$, that is, the on-impact responses to the structural shock $\varepsilon_{j,t}$. The dashed gray line indicates the $45^{\circ}$ line and the dotted vertical line the nominal $5\%$ level. The solid blue line is the median of the empirical distribution of the $p$-values, $\pi^*_{M,T,S}(x)$, while the dark and light shaded areas indicate the $50$ and $90$-percentile intervals, respectively, with the $p$-values on the $x$-axis and their corresponding empirical frequencies on the $y$-axis.}
    \label{fig:jbuni_NIGWGT500}
\end{figure}

%% file: Tables/emp_lagselect.tex
\begin{table}[H]
\centering
\renewcommand{\arraystretch}{0.85}
\scalebox{0.88}{
\begin{tabular}{|c|cccc|}
\hline
$\mathbf{p}$ & $\log|\hat{\Sigma}_u|$ & \textbf{AIC} & \textbf{BIC} & \textbf{HQ} \\
\hline
\rule{0pt}{2.5ex}
$0$ & $-12.949$ & $-12.934$ & $-12.897$ & $-12.920$ \\
$1$ & $-24.624$ & $-24.585$ & $-24.495$ & $-24.551$ \\
$2$ & $-25.653$ & $-25.591$ & $\mathbf{-25.447}$ & $-25.535$ \\
$3$ & $-25.713$ & $\mathbf{-25.627}$ & $-25.429$ & $\mathbf{-25.551}$ \\
$4$ & $-25.726$ & $-25.617$ & $-25.364$ & $-25.520$ \\
$5$ & $-25.742$ & $-25.610$ & $-25.302$ & $-25.491$ \\
$6$ & $-25.754$ & $-25.599$ & $-25.237$ & $-25.460$ \\
$7$ & $-25.774$ & $-25.595$ & $-25.180$ & $-25.435$ \\
$8$ & $-25.781$ & $-25.579$ & $-25.109$ & $-25.398$ \\
\hline
\end{tabular}
}

\caption{\scriptsize Lag order selection for the trivariate VAR in $X_t = (U_{Ft},Y_t,U_{Mt})'$ with a constant and a linear trend, monthly data over 1960:07--2025:06. All criteria are evaluated on the common sample implied by $p_{\max} = 8$, so that they are directly comparable across rows. Bold marks the minimizing lag order for each criterion: AIC and Hannan--Quinn select $p = 3$, BIC selects $p = 2$. The specification used throughout the section is $p = 3$.}
\label{tab:emp_lagsel}
\end{table}

%% file: Figures/emp_irf_50.tex
\begin{figure}[H]
    \centering
    \includegraphics[width=0.65\textwidth]{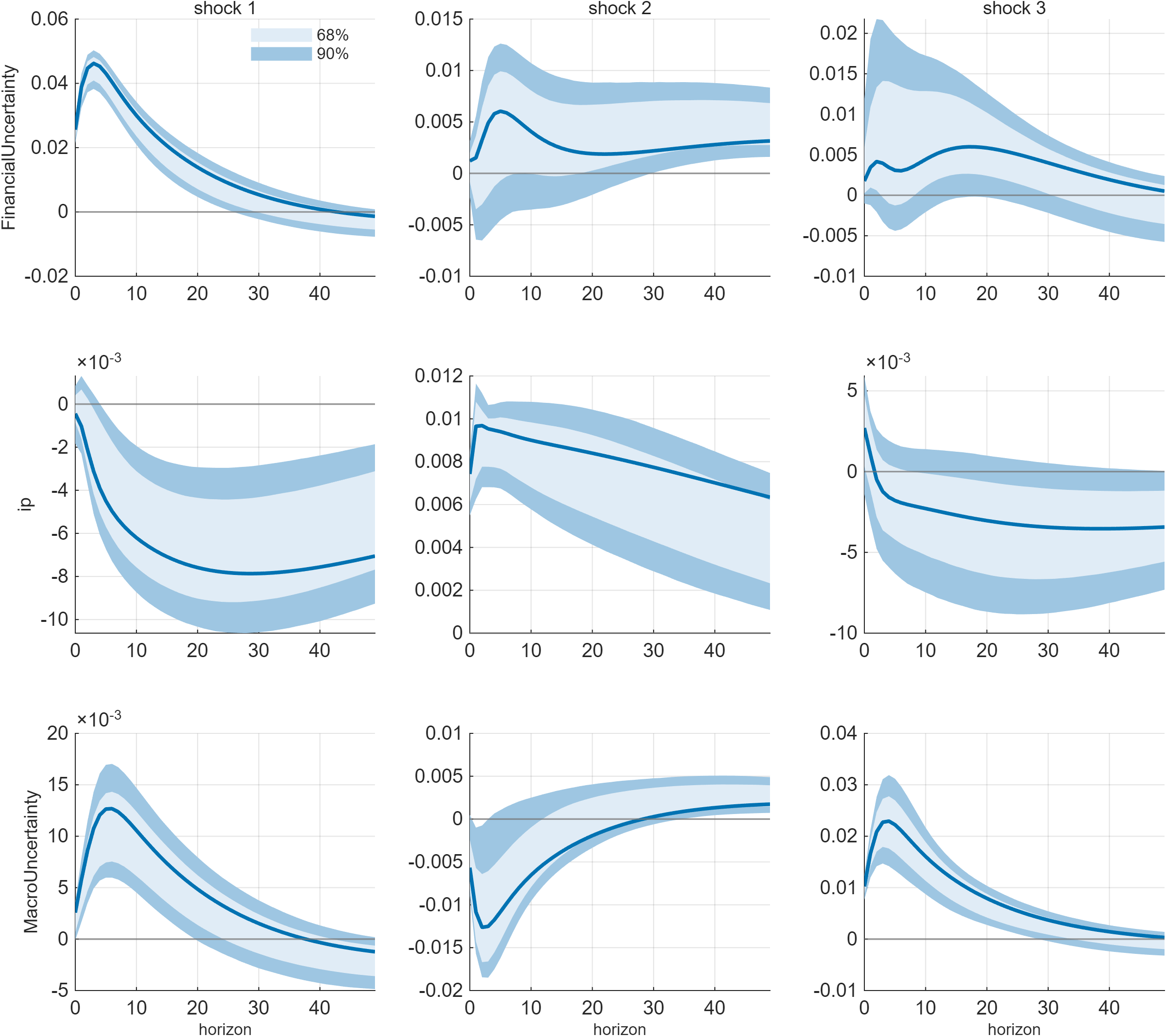}
    \caption{\scriptsize Estimated impulse responses of $X_t = (U_{Ft},Y_t,U_{Mt})'$ to one-standard-deviation positive structural shocks, with $68\%$ and $90\%$ bootstrap percentile bands from $N = 19{,}999$ replications and block length $\ell = 50$.
Rows denote the variables and columns are the estimated shocks. The horizontal axis is the horizon in months.}
    \label{fig:emp_irf_50}
\end{figure}

%% file: Tables/emp_skew_v2.tex
% Panel A from icadiag/out_27_skew/empirical_shape.csv and the symmetric run's
% implied excess kurtosis; Panel B from the two empirical_B.csv files, with the
% log likelihoods read from the two empirical.mat files.

\begin{table}[H]
\centering
\renewcommand{\arraystretch}{1.05}

\scalebox{0.80}{
\begin{tabular}{|l|cc|cc|c|}
\hline
 & \multicolumn{2}{c|}{\textbf{Symmetric}} & \multicolumn{2}{c|}{\textbf{Asymmetry free}} & \textbf{Change} \\
 & $\hat{B}_{T,ik}$ & s.e. & $\hat{B}_{T,ik}$ & s.e. & in s.e. \\
\hline
\rule{0pt}{2.5ex}
$U_{Ft} \leftarrow \hat{\varepsilon}_{1,t}$ & $0.02559$ & $0.00243$ & $0.02553$ & $0.00246$ & $-0.03$ \\
$U_{Ft} \leftarrow \hat{\varepsilon}_{2,t}$ & $0.00121$ & $0.00179$ & $0.00106$ & $0.00202$ & $-0.08$ \\
$U_{Ft} \leftarrow \hat{\varepsilon}_{3,t}$ & $0.00176$ & $0.00361$ & $0.00184$ & $0.00372$ & $0.02$ \\
$Y_{t} \leftarrow \hat{\varepsilon}_{1,t}$ & $-0.00044$ & $0.00076$ & $-0.00044$ & $0.00083$ & $0.00$ \\
$Y_{t} \leftarrow \hat{\varepsilon}_{2,t}$ & $0.00743$ & $0.00097$ & $0.00717$ & $0.00125$ & $-0.26$ \\
$Y_{t} \leftarrow \hat{\varepsilon}_{3,t}$ & $0.00270$ & $0.00218$ & $0.00329$ & $0.00242$ & $0.27$ \\
$U_{Mt} \leftarrow \hat{\varepsilon}_{1,t}$ & $0.00257$ & $0.00119$ & $0.00264$ & $0.00121$ & $0.06$ \\
$U_{Mt} \leftarrow \hat{\varepsilon}_{2,t}$ & $-0.00566$ & $0.00285$ & $-0.00650$ & $0.00318$ & $-0.30$ \\
$U_{Mt} \leftarrow \hat{\varepsilon}_{3,t}$ & $0.01024$ & $0.00159$ & $0.00978$ & $0.00218$ & $-0.29$ \\
\hline
\rule{0pt}{2.5ex}
Log likelihood & \multicolumn{2}{c|}{$7062.83$} & \multicolumn{2}{c|}{$7064.33$} & \\
% LR test of $\beta = 0$ & \multicolumn{4}{c|}{$\chi^2(3) = 3.00$, \quad $p = 0.39$} & \\
\hline
\end{tabular}
}

\caption{\scriptsize The trivariate system re-estimated with the NIG asymmetry
parameter $\gamma$ free. Each structural shock is standardized by the mean and
standard deviation of its own fitted density, so
that $\sigma_j$ remains the standard deviation of shock $j$ and
$\hat{B}_T = \hat{B}_{\mathrm{stand}}\,\mathrm{diag}(\hat{\sigma})$ remains the
impact matrix of unit-variance shocks. At $\gamma = 0$ the mean is zero and the
standard deviation is $(\delta/\alpha)^{1/2}$. The table reports both
impact matrices with bootstrap standard errors from $N = 19{,}999$ replications
at $\ell = 27$; the last column expresses the difference in units of the
symmetric bootstrap standard error.}
\label{tab:emp_skew}
\end{table}